\newif\ifjournal
\journalfalse

\ifjournal
	\documentclass[smallextended]{svjour3} 
\else 
	\documentclass[aps,pra,10pt,twocolumn,floatfix,nofootinbib]{revtex4-1}
\fi

\usepackage{amsmath}
\usepackage{amssymb}
\usepackage{graphicx}
\usepackage{amsfonts}
\usepackage{dutchcal}
\usepackage{braket}
\usepackage{enumitem}

\usepackage{tikz}
\usetikzlibrary{calc}
\usepackage{calculator}
\usepackage{standalone}

\numberwithin{equation}{section}

\ifjournal
	\spnewtheorem{assump}{Assumption}{\bf}{\it}
	
	\spnewtheorem{prop}[equation]{Proposition}{\bf}{\rm}
	\spnewtheorem{thrm}[equation]{Theorem}{\bf}{\it}

	\newenvironment{rationale}{\emph{Rationale}.}{\hfill\(\qed\)}
	\newenvironment{justification}{\emph{Justification}.}{\hfill\(\qed\)}
	\renewenvironment{proof}{\emph{Proof}.}{\hfill\(\qed\)}
\else
	\usepackage{amsthm}

	\newtheorem{assump}{Assumption}
	
	\newtheorem{thrm}[equation]{Theorem}

	\theoremstyle{definition}
	\newtheorem{prop}[equation]{Proposition}

	\newenvironment{rationale}{\emph{Rationale}.}{\qed}
	\newenvironment{justification}{\emph{Justification}.}{\qed}
	\renewenvironment{proof}{\emph{Proof}.}{\qed}

\fi

\newcommand{\id}{\textrm{id}}

\newcommand{\journal}[1]{\ifjournal#1\fi}
\newcommand{\arxiv}[1]{\ifjournal\else#1\fi}

\begin{document}

\title{From physical assumptions to classical and quantum \\ Hamiltonian and Lagrangian particle mechanics}
\author{Gabriele Carcassi, Christine A. Aidala, David J. Baker, Lydia Bieri}

\ifjournal
	\institute{G. Carcassi, C. A. Aidala \at
		Physics Department, University of Michigan, Ann Arbor, MI 48109 \\
		\email{carcassi@umich.edu}
		\and
		D. J. Baker \at
		Philosophy Department, University of Michigan, Ann Arbor, MI 48109
		\and
		L. Bieri \at
		Mathematics Department, University of Michigan, Ann Arbor, MI 48109 \\
	}
\else
	\affiliation{University of Michigan, Ann Arbor, MI 48109}
\fi

\ifjournal
	% As the page margins are so large, the full title does not fit
	\titlerunning{From physical assumptions to classical and quantum particle mechanics}
\fi

\date{\today}

% Journal format has this before the abstract
\journal{\maketitle}
	
\begin{abstract}
The aim of this work is to show that particle mechanics, both classical and quantum, Hamiltonian and Lagrangian, can be derived from few simple physical assumptions. Assuming deterministic and reversible time evolution will give us a dynamical system whose set of states forms a topological space and whose law of evolution is a self-homeomorphism. Assuming the system is infinitesimally reducible---specifying the state and the dynamics of the whole system is equivalent to giving the state and the dynamics of its infinitesimal parts---will give us a classical Hamiltonian system. Assuming the system is irreducible---specifying the state and the dynamics of the whole system tells us nothing about the state and the dynamics of its substructure---will give us a quantum Hamiltonian system. Assuming kinematic equivalence, that studying trajectories is equivalent to studying state evolution, will give us Lagrangian mechanics and limit the form of the Hamiltonian/Lagrangian to the one with scalar and vector potential forces.
\end{abstract}

\arxiv{\maketitle}

\section{Introduction}

The past century has seen a resounding success in the use of mathematics in fundamental physics. In many cases, mathematical ideas were what allowed progress in theoretical physics and ultimately led to experimental discoveries. While these ideas are indeed useful---anything that can advance our understanding is welcome---they cannot be the whole story.

The issue is that mathematical structures are not enough to characterize what physical system is being studied. For example, the same mathematical framework for linear circuits can be used to describe electric, hydraulic, thermal or mechanical systems~\cite{Roland,Barron,Borutzky}. The mathematical treatment is the same because it captures aspects of the description that are common to all cases; that indeed is its power. But for the very same reason, we cannot infer what system is being described if all we are given is the mathematical model. In the jargon of linguistics/computer science/philosophy~\cite{Freidin,Friedman}: math captures the syntax (the relationships between the objects) but does not capture the semantics (the meaning of each object).

The real problem begins when the mathematical description essentially becomes the foundation of a physical theory. While some branches of physics, such as Newtonian mechanics, thermodynamics and special relativity, are founded on physical laws or assumptions, others, such as Hamiltonian, Lagrangian and quantum mechanics, start by setting their mathematical structure. Nothing tells us why there should be  a Lagrangian function whose path integral is minimized during the motion, or when a state is described by conjugate pairs $(q^i, p_i)$ instead of a vector in a complex inner product space. Those are taken as given.

This means that if we have a system in front of us, we do not have a rigorous way to conceptually establish whether the system is Hamiltonian or not, whether it is classical or quantum. We have some heuristics we can apply, but, at the end the day, a system is Hamiltonian because it follows the Hamiltonian framework.\footnote{The lack of a properly defined, and uniquely accepted, semantics is of course very evident in quantum mechanics with its numerous interpretations.}

Unsatisfied with this situation, we started creating a sort of math-to-physics dictionary so that each mathematical concept could be clearly associated to a crisp physical concept. As we pushed further back toward the basic mathematical definitions, we ended up with physical concepts of a more general nature and found that, in some cases, a single physical idea could have more than one mathematical consequence. That is: by putting physics at the center, we could achieve a unification of ideas that was not evident before. The unexpected success of the approach led us to think that, perhaps, the picture we were developing was a better framework to understand and link different areas of basic science and mathematics. This work is the systematized result of this endeavor.

The main goal is to show that it is possible to derive classical/quantum Hamiltonian/Lagrangian particle mechanics starting from few physical assumptions. We do \emph{not} claim that we have perfectly achieved this result, with no room for improvement. We \emph{do} claim that the whole architectural framework is solid enough to be convinced that not only is this indeed possible, but it allows a more unified picture that facilitates intuitive connections within and across different areas of physics and mathematics. It enables a consistent narrative across several domains of knowledge.

Rederiving classical, quantum and (to some extent) relativistic mechanics within the same work is therefore crucial to show that the concepts we develop are truly general, and they can be applied to many different cases. If we could derive classical Hamiltonian mechanics on assumptions that are inapplicable to the quantum case, what good would it do? Quantum systems are Hamiltonian systems, so clearly we wouldn't have understood what makes a system Hamiltonian. This breadth of applicability is what makes the work valuable.

As a brief overview, there are four assumptions we will consider. Assuming deterministic and reversible evolution gives us a dynamical system: states and a law of evolution. Additionally assuming infinitesimal reducibility, studying the whole system is equivalent to studying its infinitesimal parts, gives us Hamiltonian mechanics. Further assuming kinematic equivalence, trajectories in space tell us everything about the state of the system, gives us Lagrangian mechanics and constrains the motion to massive particles under scalar/vector potential forces. Relativistic motion arises when studying time dependent laws of evolution with no further assumptions. Alternatively, as the second assumption we can instead take irreducibility, studying the whole system tells us nothing about its parts, which leads to quantum (Hamiltonian) mechanics.

As a more detailed overview: in section \ref{sec:fundamental_model} we describe in general what it means to study a physical system and how the assumption of determinism and reversibility plays into it. We develop an abstract conceptual model that we use when discussing the later assumptions.

In section \ref{sec:dynamical_systems} we show how the physical concept of physical distinguishability leads to the mathematical concept of topological spaces. Deterministic and reversible evolution is an invertible continuous map (i.e. a homeomorphism) as it needs to preserve what is physically distinguishable.

In section \ref{sec:reducibility} we introduce the notion of a classical material, one for which we assume infinitesimal reducibility. The state of an arbitrary amount of material is a distribution over the states of the infinitesimal parts, the classical particles. As such a distribution must be invariant under the arbitrary choice of units, we see that for each variable $q^i$ that defines a unit there is a corresponding variable $k^i$ defined on the inverse of that unit. Classical particle states can be identified by conjugate pairs of state variables $(q^i, p_i)$ (i.e. a point in the cotangent bundle $T^*\mathcal{Q}$) and the number of possible initial conditions for each pair is given by $dq^i \wedge dp_i$ (i.e. the canonical symplectic form $\omega$). During evolution, the distribution (and marginal distributions) must be mapped point by point which leads to Hamiltonian mechanics (i.e. a symplectomorphism on $T^*\mathcal{Q}$).

In section \ref{sec:relativity} we show that a relativistic version of Hamilton's equations arises naturally when we consider time dependent laws of evolution for distributions of classical material. This also gives the notion of deterministic and reversible evolution a characterization that is independent of any transformation that mixes time and state variables $(q^i, p_i)$. Moreover it gives us a classical equivalent for anti-particle states.

In section \ref{sec:Lagrangian} we introduce the assumption of kinematic equivalence, studying the motion in space is equivalent to studying the evolution of the state. Under that assumption there is a link between the state variables $(q^i, p_i)$ and the kinematic variables $(x^i, u^i)$ (i.e. position and velocity). This link allows us to express the number of possible initial conditions for a degree of freedom (d.o.f.) in terms of $dx^i$ and $du^i$, leading to the metric tensor $g$. Such an expression also allows us to constrain the Hamiltonian to the one for massive particles under scalar/vector potential forces and an equivalent Lagrangian formulation.

In section \ref{sec:quantum_prelude} we introduce some ideas in classical mechanics to more gently transition to quantum mechanics. We see how distributions over phase space already have features that are typically associated with quantum systems and we will show how to construct a classical uncertainty principle that helps to build intuition used in the later section.

In section \ref{sec:irreducibility} we introduce the notion of an irreducible material, one for which studying the whole distribution does not tell us anything about the motion of its infinitesimal parts, the fragments. A unit amount of this material will correspond to a quantum particle. As this means that the states of such a material are invariant under permutations of the configurations of the fragments, we show that the state space of an irreducible material is a complex vector space. These states can still be compared to each other in a way that gives the state space an inner product. Such comparisons will be preserved by deterministic and reversible evolution which leads to Schroedinger's equation (i.e. a unitary transformation).

In section \ref{sec:discussion} we present a few higher level remarks about the work.

\section{Organization and style}

The overall layout of this work is a balance between the need to show that the derivation follows, which requires starting from basic principles and proceeding using formal arguments, and the need to give an intuitive understanding, which often requires starting from specific cases and then generalizing.

The deductive aspect is formally organized into the following:
\begin{description}
  \item[Assumptions] these characterize the physical system we are studying and constitute the premise of our discussion. A \textbf{rationale} follows each assumption, which uses physical and sometimes philosophical arguments to motivate why (or not) such an assumption makes sense in a particular case.
  
  \item[Propositions] these capture the properties of our physical objects into mathematical language. A \textbf{justification} follows each proposition to show the necessity of such a characterization. While it is a formal argument, it is not a strictly mathematical proof as it contains physical arguments.
  
  \item[Theorems] these are pure mathematical statements that are used to simplify propositions. A mathematical \textbf{proof} follows each theorem. No mathematical breakthrough should be expected as we mostly use well known results from different fields. Theorems are grouped in the appendix.
\end{description}

As these sections are clearly demarcated, they allow us to make sure that each justification only references past propositions and avoid circular arguments. It also makes it easier to read the work at the desired level of detail. Interspersed between each step of the more formal treatment, there are numerous discussions that conceptually explain motivations and results in a more approachable way. These often provide more concrete and intuitive examples which may, in some cases, foreshadow later material. As both the formal and more intuitive parts need to cover the same material, this will inevitably lead to some repetition.

The work uses ideas and tools from a broad range of mathematical and scientific subjects: topology~\cite{LeeTM}, measure theory~\cite{Halmos}, differential geometry (symplectic and Riemannian)~\cite{Lee}, group theory and vector spaces~\cite{Lang,Young}, statistics~\cite{Grimmet}, information theory~\cite{Pierce}, Hamiltonian mechanics, Lagrangian mechanics~\cite{classical_dynamics}, quantum mechanics~\cite{Landau,Weinberg}, thermodynamics~\cite{Van Ness}, special relativity~\cite{French} and so on.\footnote{The references given provide the reader a list of resources consistent with the development of this work.} Using ideas from all these areas is necessary and valuable, but it does create a number of problems.

The first one is terminology and prerequisite knowledge. Since we touch numerous areas, we cannot expect all readers to be equally comfortable with all the technical terms. Yet we cannot introduce all of them either. The compromise we adopted is to keep the informal discussion more accessible while keeping the formal derivation more precise, without needing more than the basic definitions and few major results within each specific area. This may leave an expert in a particular field not fully satisfied but, as stated in the introduction, our goal is to lay out the general picture.

The second problem is notation. As each field has developed its own set of conventions, the common usage is now an empty set. We ended up with a compromise that hopefully feels natural enough to a broad audience of physicists, mathematical physicists and philosophers of physics. Within formal parts we may use notation and terminology more specific to the appropriate area of math.

We have tested the prose on a number of colleagues and students in physics, philosophy and mathematics to confirm its accessibility at different levels and believe we have struck a reasonable compromise.

\section{On studying a physical system}
\label{sec:fundamental_model}

% Status: Conceptual work done. Text ready. Incorporated Dave, Isaac and Christine feedback

Our first task is to develop a conceptual model that applies to all realms of physics we'll be considering: classical, statistical and, later, quantum mechanics. We will take the standard picture of system plus environment and extend it to differentiate between the state of the system (i.e. the aspects under study) from the unstated part of the system (i.e. the aspects missing from our description).

We will assume that the state evolves according to a deterministic and reversible law (i.e. for each present state there is one and only one future state). While the unstated part does not influence the state evolution, we'll see how it constrains what states are available and what type of description can or cannot be given to the system.

\subsection{States and their evolution}

We start by fixing a \emph{physical system}, meaning something we can interact with and perform measurements on (e.g. a planet, a fluid, ...). We call \emph{environment} everything else. We set what particular aspect we want to study (e.g. the motion around a star, the flow in a pipe, ...). We call \emph{state} a physically distinguishable configuration of the aspect under study at a particular time (e.g. position/momentum of center of mass, velocity field, ...). Since the state does not, in general, exhaust the description of the system, a part remains \emph{unstated}, and as such we'll call it, for lack of a better word (e.g. the chemical composition, the motion of each of its molecules, ...). Note that the environment plays an essential role here as it's what allows us to define two states as physically distinguishable: we can find an external process (i.e. part of the environment) whose outcome changes depending on the different state. As such processes are what we can use to perform measurements, we consider the experimental apparatus (and us performing measurements) independent of the system, part of the environment.\footnote{This is true even in the case of general relativity, where we can imagine multiple researchers on small spaceships collecting data without greatly influencing the motion of stars and planets. The case where the physical system is the whole universe and there is no environment presents practical problems  and conceptual challenges that the current physical theories do not seem to be equipped to address, and therefore will be absent from our discussion. For example, what physical device can we use to store and process the state of the whole universe to make predictions and compare? How do we define physically distinguishable? Do the physical laws determine which measurements we are going to make and does that limit what is actually distinguishable?}
 
In this context, we call the evolution of the system \emph{deterministic} if the state at a given time uniquely identifies states at future times, and \emph{reversible} if it uniquely identifies states at past times. While this is a common enough definition, we need to be clear how this applies to the unstated part. Note that the concept of determinism outlined here is context dependent because the state itself represents only the part of the system that we choose to (or can) describe. In this sense, the unstated part is \emph{always} non-deterministic (and non-reversible) in the sense that the state of the system does not determine its evolution. For example, suppose we define the state as position and momentum of the center of mass of a cannonball. Suppose that the evolution is deterministic on that state. What does it tell us about the cannonball temperature, or about the motion of each of its atoms? Nothing. In this sense, the unstated part is non-deterministic and non-reversible. Could we extend the state and the laws of evolution to account for temperature? Yes, but that would be a different evolution defined on a different state. Does it mean the unstated part is always evolving randomly? Not at all. The temperature may remain constant throughout the motion of the cannonball. Yet we wouldn't know, since we are not studying it: the evolution is deterministic and/or reversible only as far as the state is concerned. With this in mind, we will restrict ourselves to the cases where the following is valid:

\begin{assump}[Determinism and reversibility]\label{ass:determinism}
The state of the physical system under study undergoes deterministic and reversible evolution.
\end{assump}

\begin{rationale}
As it is an assumption, we first need to discuss when it is valid. More specifically, we need to understand that the non-deterministic/non-reversible evolution of the unstated part plays as much of a fundamental role as the deterministic/reversible evolution of the state. In fact, the non-deterministic part contributes in determining what states are available to the system.

Suppose we study the motion of a cannonball; its state under gravitational and (inertial) inertial forces will be properly described by the position and momentum of the center of mass. While light and air molecules may scatter off its surface unpredictably, its trajectory is not greatly affected as it is a massive rigid body. Suppose we study the motion of a small particle, small enough that the random scattering does influence the trajectory and it undergoes Brownian motion: its state will be a probability distribution for position and momentum of the center of mass. Gravitational and inertial forces have not changed, yet the states have changed from ``pure" to statistical ensembles. In other words, the set of states must be closed under both the deterministic evolution of the state and the non-deterministic evolution of the unstated part. If the Brownian motion is not negligible, we do not end in a well defined position/momentum pair, even if we start from one.

A similar more drastic effect: consider a book and its motion under gravitation and inertial forces, its state being the position and momentum of the center of mass. As we increase the temperature of the air around the book, its motion remains unaffected until, at some point, the book burns. Clearly, the non-deterministic evolution has pushed one of the states outside the set of states, to the point that the system is no longer recognizable.

As we have hinted, sometimes the state is identified by a distribution (either statistical or actual). Even in this case, the state can be deterministic and reversible. That is, given the distribution at one time we can determine the distribution at future times. The shape and the parameters of the distribution can be deterministic, even if the evolution of the parts are not as they fall within the unstated part. Note that we \emph{cannot} assume trajectories and states are always defined for the unstated part, as this includes also the unknown unknowns. We will return to this aspect when discussing quantum systems.

It should also be clear that what constitutes state and unstated part does not depend only on the system under study, but also on the processes we are considering. In some circumstances, the chemical composition of a fluid may be relevant, in others it may not. The choices of environment, state and unstated part are not independent from each other. By choosing a particular set of states, we are not only saying that the state evolution is well approximated by a deterministic/reversible map from initial to final state, but we are also saying that the non-deterministic/non-reversible evolution of the unstated part does not change the nature of the system, and processes that do not satisfy these conditions are not under consideration.

As with all assumptions, we should also ask whether it is necessary. That is, could we define a set of physically distinguishable states and yet have no deterministic and reversible processes defined on them? The claim is that this assumption is indeed needed, as without it we cannot properly define states or write useful physics laws. We can provide different arguments that point in the same direction.

First, to be able to identify the system, we must be able to tell it apart from anything else. Intuitively, we can distinguish between two chairs because we can move the first to another room and sit on it without having touched the second. We can manipulate the state of the first system without affecting the second, and vice-versa. So, to identify a system it has to be sufficiently isolated from everything else. This means that the system future and past states are with good approximation determined only by its own state: the state undergoes deterministic and reversible evolution.

Second, the aim of physics is to write laws that can be used to make predictions that can be validated experimentally. If I drop an anvil from a tower, it will accelerate at $9.81$~m/s$^2$; if I want the anvil to reach the ground at $x$~m/s I have to drop it from $y$~m. To the extent that we want to make predictions in time, we need to have a correspondence between initial and final states.

Third, operationally we must reliably prepare and measure states. That is, we need a process for which the input settings of our preparing device determine the outgoing state of the system; and a process for which the incoming state of the system can be reconstructed by the output of the measuring device. That is, our system must, at least in some cases, be able to participate in a deterministic and reversible process with the preparing and measuring device. Without it we wouldn't be able to calibrate our experimental apparatus.

Fourth, to be able to ascribe a property to a system we need to claim that, at least for a finite interval of time, the system either held or did not hold such property. That is, there is a deterministic and reversible process for that finite period of time for such property.

This link between state definition and deterministic processes should not be too surprising as the state, in the context of control theory and dynamical systems theory, is often defined as \emph{the set of variables needed to determine the future evolution of the system}~\cite{Katok,Sontang}. As we saw before, this applies also to statistical processes: the distribution (the ensemble) as a whole can indeed be calculated, measured and prepared. We can also describe the evolution of each element of the distribution provided that: we have a way to isolate it and study it under deterministic and reversible motion (so that we can define microstates); the non-deterministic motion does not alter the system (the set of microstates is preserved by the evolution).

As with many assumptions, we should stress that it's an idealization: it can never be completely achieved in practice. A system can be prepared or measured up to a certain level of precision. Perfect determinism and isolation of a system is impossible both practically (e.g. black-body radiation, gravity, ...) and conceptually (e.g. if the system is perfectly isolated, we cannot interact with it: how can it be physically distinguished?). It's a simplifying assumption that can only be taken if the environment and the internal dynamics of the system interact in such a way that they little affect and are little affected by the aspect we are studying. As we saw before, for example, assuming that the state consists of the position and momentum of the center of mass requires assuming that the Brownian motion of the body is negligible.

Yet, this is a fundamental assumption in the sense that it is needed. If a particular set of states does not satisfy deterministic and reversible evolution under certain conditions, what we do is to keep at it until we find a set that does. That is, we work to restore the assumption. Finding new sets of states with new laws of evolution is, in fact, what leads to new physics. We therefore call \emph{fundamental model of physics} the triad of state, unstated part and environment, with the assumption that the state may undergo deterministic and reversible evolution, and the unstated part undergoes non-deterministic non-reversible evolution that does not alter the set of states.
\end{rationale}

\section{States and state space}
\label{sec:dynamical_systems}

% Status: First subsection. Conceptual work done. Text ready. Incorporated Christine feedback
% Second subsection. Conceptual work almost done. Need to finish the text

%Open questions: smoothness (obviously required to define densities later, is it required in general? What does it means that physical processes/measurements don't give a smooth topology?)

We now proceed to characterize states and physical distinguishability in more precise terms so that we can capture their description mathematically. What we'll see is that the outcomes of all physical processes that can be used to gain information about the system, that can be used to perform a measurement, induce a topology on the set of states that is at least Hausdorff. That is, physical distinguishability is mathematically captured by topological distinguishability. Deterministic and reversible evolution will then preserve the topology, and they will be mathematically captured by self-homeomorphisms.

We'll focus on state spaces that can be described by a set of independent state variables, either discrete or continuous, and see what can be said in general on the evolution of state variables.

\subsection{States and topology}

As the term ``measurement" has become particularly loaded, let's first characterize what we mean by physical distinguishability in our context.

Consider the motion of a cannonball under inertial and gravitational forces. Light will scatter off of it; as it lands, the ground will be deformed and the impact will make the temperature rise slightly. Those external physical processes, which happen no matter what we do, can be used to distinguish the motion of the cannonball as their outcomes are correlated. Therefore we can learn the position by looking at the reflected light, and learn the final kinetic energy by looking at the deformation of the ground. In principle, any external process that has a correlation with the states under study can be used to perform a measurement, and any measurement is based on such a process. That is, for us a measurement is simply a physical process that we can use to distinguish states. Setting up an experiment means choosing a particular process with desired outcomes and forcing the system under study to interact with it one or more times. After that, there is no special role played by the ``observer" in making outcomes come about.

This external process may be quite complicated: when a particle enters a calorimeter, a shower of particles is produced, photons are captured and are directed to photomultipliers, a current is read out, the current is then digitized, and so on. Some processes interfere with the system, they affect its dynamics, and others are destructive, the system no longer can be described by the original set of states. A tracking chamber is an example of the first (the magnetic field curves the motion of a charged particle); burning a substance to determine its caloric content is an example of the latter. Therefore intimate knowledge of the process is always needed to ensure that one makes the proper link between outcomes and the \emph{original} states, and properly accounts for systematic uncertainties that would skew that link.

Repeatability is also fundamental. First, to make sure the process is indeed correlated to the states. Second, because ``a single take does not a measurement make." One has to gather enough statistics. Note that the number of takes influences the outcomes: with greater statistics the precision and number of distinguishable cases increases. Therefore the processes, as we defined them, may require repeated interactions with similarly prepared states. They may even be a combination of different kinds of interactions that, taken all together, create a set of distinguishable outcomes. But in the end, however complicated it is, the conceptual model remains the same: each process has a set of possible outcomes, and each possible outcome will be associated with a set of states consistent with that outcome. For example, if the cannonball deformed the ground by this amount then its kinetic energy at impact was within this range; if the electron follows a certain path, then its state is among the ones that have spin $+1/2$.

Note that the outcome of the process (e.g. the electron followed a particular path) is conceptually distinct from the set of states compatible with it (the states associated with the path) and the way we label them (spin=$+1/2$). Yet,  we'll conflate the concepts as this makes the language closer to experimental physics (``the outcome of the measurement is spin=$+1/2$"). Context will allow us to recover the distinction.

However precise our measurements are, we can only gather a finite amount of statistics and each outcome is expressed by a finite set of digits; therefore, the set of outcomes is countable.\footnote{The information provided by the process as measured by Shannon's entropy~\cite{Shannon,Jaynes} is finite.} Note that different outcomes for the same process can overlap. For example, $4.12 \pm 0.05$ cm and $4.13 \pm 0.05$ cm are both legitimate possible outputs of the same measurement device. But since all states must be distinguishable, given two arbitrary states there must be a process precise enough to tell them apart. That is, the potential outcomes associated with the two states do not overlap. For example, $4.12 \pm 0.0005$ cm and $4.13 \pm 0.0005$ cm do not overlap anymore.

We can also conceptually combine two different processes into a single one.\footnote{That is, take the logical \textsc{and} between outcomes of different processes.} That is, having a way to measure quantity $x$ and a way to measure quantity $y$ gives us a way to measure the combination $(x,y)$. For ensembles, if we can measure the marginal distribution $\rho_x(x)$ and the marginal distribution $\rho_y(y)$, we know the joint distribution $\rho(x,y)$ has to be compatible with both. Note, though, that this does not provide a way to fully measure $\rho(x,y)$ as we know nothing about the correlation between the two variables.\footnote{The quantum case is similar: we measure marginal distributions and rule out states that are incompatible with them.} Formally, the states compatible with the outcomes of the combined process will be the intersections of the states compatible with each pair of outcomes of the original processes.

We can also coarsen a single process.\footnote{That is, take the logical \textsc{or} between outcomes of the same process.} That is, having a way to measure $(x,y)$ gives us a way to measure $x$ alone (or any $f(x,y)$). Formally, the outcomes of the coarsened process are given by performing the union of some outcomes of the original process.

This model maps very naturally to a topological space. The states are elements of a set and the physical outcomes provide a topology on that set. The fact that two elements of the set can be distinguished requires the space to be Hausdorff.

\begin{prop}\label{prop:state_topology}
The state space $\mathcal{S}$ of a physical system is a Hausdorff topological space.
\end{prop}

\begin{justification}
We claim $\mathcal{S}$ is a set. Each state is well defined as it is physically distinguishable. The collection of all possible states forms a set.

We claim $\mathcal{S}$ has a topology $\mathsf{T}$. Consider the set of all possible physical outcomes associated with all physical processes. Each possible outcome is associated with a set of states that are compatible with that outcome. Let $\mathsf{T}$ be the set of all sets associated with all physical outcomes. $\mathcal{S} \in \mathsf{T}$ and is associated with the outcome ``the system exists." $\varnothing \in \mathsf{T}$ and is associated with the outcome ``the system doesn't exist." Let $V_1, V_2 \in \mathsf{T}$. Then, by definition, there exists a process $P_1$ that admits $V_1$ as an outcome and a process $P_2$ that admits $V_2$ as an outcome. Consider the process $P$ that combines the outcomes of $P_1$ and $P_2$ with a logical \textsc{and}. This always exists physically as we can prepare the same state multiple times and let it interact with each process separately. $P$ will have a possible outcome $V$ corresponding to the case where $P_1$ gave outcome $V_1$ and $P_2$ gave outcome $V_2$. The states compatible with $V$ must be in both $V_1$ and $V_2$, that is $V = V_1 \cap V_2$. Therefore $\mathsf{T}$ is closed under intersection. Let $V_1, V_2 \in \mathsf{T}$. If they are physically distinguishable, then there exists a physical process $P$ that admits both as outcomes. Given $P$, we can always construct the process $P_0$ that combines $V_1, V_2$ with a logical \textsc{or} into a single outcome $V$, by ``forgetting" which of the two was given. The states compatible with $V$ must be in either $V_1$ or $V_2$, that is $V = V_1 \cup V_2$. Therefore $\mathsf{T}$ is closed under union. $\mathsf{T}$ is a topology by definition.

We claim that $\mathsf{T}$ is Hausdorff. Let $\mathcal{s_1}, \mathcal{s_2} \in \mathcal{S}$. As states are physically distinguishable, there must exist a physical process with two possible outcomes $V_1, V_2 \in \mathsf{T}$ for which $s_1 \in V_1, s_2 \in V_2, V_1 \cap V_2 = \varnothing$. $\mathcal{S}$ is Hausdorff by definition.
\end{justification}

Note that the arguments that led to the topological space had nothing to do with states per se, just that they are physically distinguishable. States, though, are not the only objects with that property. In fact, any element of a set of physical objects (e.g. forces, physical properties such as mass or charge, time) needs to be physically distinguishable to be well defined. We can generalize the above justification: any set of physically distinguishable elements is a topological space.

\begin{prop}\label{prop:topology}
	Any set $\mathcal{S}$ of physically distinguishable elements is a Hausdorff topological space.
\end{prop}

\begin{justification}
	Same justification as in \ref{prop:state_topology} with ``state" replaced by ``element" of the set $\mathcal{S}$.
\end{justification}

With our state space defined, deterministic and reversible evolution corresponds to a bijective map between initial and final states. But that's not enough. A physical process that can distinguish final states can be used, together with the evolution, to distinguish initial states. We prepare the initial state, let it evolve deterministically, measure the final state and use reversibility to convert the measure to one on the initial state. Therefore, not only the evolution is a bijection, but it also maps outcomes to outcomes. That is: the topology is mapped and preserved by the deterministic and reversible evolution because physical distinguishability must remain unchanged. The evolution is then a self-homeomorphism on the state space.

\begin{prop}\label{prop:homeomorphism}
A deterministic and reversible evolution map is a self-\journal{\break}homeomorphism on the state space $\mathcal{T}_{\Delta t}:\mathcal{S} \rightarrow \mathcal{S}$.
\end{prop}

\begin{justification}
We claim that $\mathcal{T}_{\Delta t}$ exists. The system is deterministic: given an initial state $\mathcal{s} \in \mathcal{S}$ there exists a well defined final state $\mathcal{T}_{\Delta t}(\mathcal{s})=\hat{\mathcal{s}} \in \mathcal{S}$.

We claim $\mathcal{T}_{\Delta t}$ is continuous. Let $U \subseteq \mathcal{S}$ represent an outcome of a process $P$ on the final states. $U$ is an open set in the (final) state space topology by definition. Consider the process $P_0$ that first evolves the initial states and then distinguishes the final state with $P$. $P_0$ is a process that distinguishes initial states. The set of initial states compatible with $U$ are $\mathcal{T}_{\Delta t}^{-1}(U)$. $\mathcal{T}_{\Delta t}^{-1}(U)$ is an open set in the (initial) state space topology by definition. $\mathcal{T}_{\Delta t}$ is a continuous map.

We claim that $\mathcal{T}_{\Delta t}$ is a bijection. The system is reversible: there exists a map $\mathcal{T}_{-\Delta t}:\mathcal{S} \rightarrow \mathcal{S}$ that returns the initial state given the final state. $\mathcal{T}_{-\Delta t} \circ \mathcal{T}_{\Delta t} = \mathcal{T}_{\Delta t} \circ \mathcal{T}_{-\Delta t} = \id_{\mathcal{S}}$ as mapping forward and then backward or backwards and then forward must return the original element. $\mathcal{T}_{\Delta t}$ admits $\mathcal{T}_{-\Delta t}$ as an inverse. $\mathcal{T}_{\Delta t}$ is a bijection.

We claim that $\mathcal{T}_{\Delta t}$ is a self-homeomorphism as it is a continuous bijection.
\end{justification}

Again, we note that the arguments that lead to continuity had nothing to do with states per se, just that the relationship is between physically distinguishable objects. We can then generalize the above justification.

\begin{prop}\label{prop:continuity}
	A map $f:\mathcal{S_1} \rightarrow \mathcal{S_2}$ between two sets of physically distinguishable elements $\mathcal{S_1}$ and $\mathcal{S_2}$ is a continuous map.
\end{prop}

\begin{justification}
	Same justification for continuity as in \ref{prop:homeomorphism} with ``initial states" and ``final states" replaced by ``elements" of $\mathcal{S_1}$ and $\mathcal{S_2}$ respectively.
\end{justification}

The generality of this result explains why in physics one always assumes functions to be ``well behaved." As the result was derived from our notion of physical distinguishability, this is not a matter of practical convenience. Suppose we were able to prepare a force field that was zero everywhere in space except at a single point. This gives us a way to tag a specific point. But it also allows us to create an outcome compatible with only a single point. A finite precision measurement of the force provides us infinite precision of space, which we ruled out.\footnote{This assumes that we are able to position the probe perfectly, which we could do if we were able to manipulate forces at that precision.} This is what the math is telling us, that if we claim that outcomes can't distinguish isolated points (i.e. standard topology on $\mathbb{R}^n$) then neither can maps. It is physical consistency that limits us to continuous functions.

\begin{figure}
		\begin{tikzpicture}[scale=.4]

		% INPUTS
				
		\newcommand{\displace}{23};			 % distance between the two plots, origin to origin		

		\newcommand{\axislength}{10};	   	 % axis lengeth
		\newcommand{\axisheight}{8};		 % axis height	

		\newcommand{\horlinelength}{7};		% length of all horizontal lines
		\newcommand{\vertlinelength}{7};	 % length of all vertical lines

		\newcommand{\lineheighta}{4};		  % y-coordinate of first horizontal line in both plots
		\newcommand{\lineheightb}{5.5};		 % y-coordinate of second horizontal line in both plots
		\newcommand{\functionheight}{2};  % y-coordinate of function in right plot
		
		\newcommand{\linea}{3};					  % x-coordinate of first vertical line in left plot
		\newcommand{\lineb}{5};				   	  % x-coordinate of second vertical line in left plot
		\newcommand{\linec}{4};					  % x-coordinate of vertical line in right plot	
		
		\newcommand{\hmargin}{2};
		\newcommand{\vmargin}{1.5};
		
		% styles		
	    \tikzstyle{axis}=[color=black];				
	    \tikzstyle{line}=[color=black,dashed];  % dashed better than dotted?
	    \tikzstyle{function}=[ultra thick,color=blue];			    

		% outer rectangle and name of diagram
		\draw ({-\hmargin - \axislength},{- \axisheight - \vmargin}) rectangle ({\displace + \axislength + \hmargin},{\axisheight + \vmargin});

		% axes for left plot
		\draw [style=axis] (0,-\axisheight) -- (0,\axisheight);		
		\draw [style=axis] (-\axislength,0) -- (\axislength,0);		
		
		% lines in left plot
		\draw [style=line] (\lineb,-\vertlinelength) -- (\lineb,\vertlinelength);
		\draw [style=line] (\linea,-\vertlinelength) -- (\linea,\vertlinelength);		
		
		\draw [style=line] (-\horlinelength,\lineheighta) -- (\horlinelength,\lineheighta);		
		\draw [style=line] (-\horlinelength,\lineheightb) -- (\horlinelength,\lineheightb);		

		% function in left plot		
		\draw [style=function] plot [smooth] coordinates {(-8,-3) (-4, -2) (0,0) (\linea,\lineheighta) (\lineb,\lineheightb) (8,6)};

		% axes for right plot
		\draw [style=axis] (\displace + 0,-\axisheight) -- (\displace + 0,\axisheight);		
		\draw [style=axis] (\displace + -\axislength,0) -- (\displace + \axislength,0);		
		
		% lines in right plot
		\draw [style=line] (\displace + \linec,-\vertlinelength) -- (\displace + \linec,\vertlinelength);

		\draw [style=line] (\displace + -\horlinelength,\lineheighta) -- (\displace + \horlinelength,\lineheighta);		
		\draw [style=line] (\displace + -\horlinelength,\lineheightb) -- (\displace + \horlinelength,\lineheightb);		

		% function including dot in right plot		
		\draw [style=function]	(\displace + -8,\functionheight) -- (\displace + 8,\functionheight);
		\fill [style=function] (\displace + \linec,4.75) circle (.25);
		\filldraw [draw=blue,fill=white,thick] (\displace + \linec,\functionheight) circle (.25);

	\end{tikzpicture}
	
	\caption{A continuous function always maps finite precision knowledge in the vertical axis (i.e. an open set) to finite precision knowledge in the horizontal axis. A discontinuous function may map to infinite precision (i.e. a closed set). If infinite precision identification is ruled out (i.e. use of standard topology) then discontinuous functions must also be ruled out.}
	\label{fig:continuity}
\end{figure}
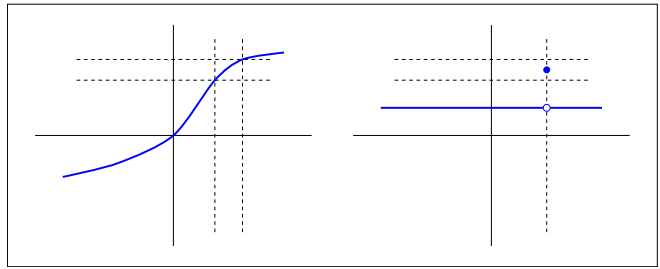

While functions with few discontinuities are useful and used in physics and engineering, they are employed for idealized cases  (e.g.~a signal change is fast enough, a charge distribution is small enough) that are often treated as a special case (e.g.~propagation across material discontinuities). While this may be obvious and intuitive to the physicist, it may be troubling to the mathematician as the proper use of many mathematical techniques requires the inclusion of discontinuous functions. This is less of a problem than it would seem at first. Once we made sure that the objects and their relationships are physically meaningful, we can extend our mathematical spaces for the purpose of math computations. Our physically meaningful continuous function can be expanded into a sum of discontinuous functions. One just has to be mindful of the extension and be wary that mathematical results that depend on such extension may or may not be physically meaningful.

\subsection{Manifolds and labeling states}

To identify and name states one uses a set of quantities, typically numbers. For example, the orbital of an electron in a hydrogen atom is identified by the quantum numbers $n$, $l$, $m$ and $s$. We call each of these quantities \emph{state variables}. We call a \emph{possibility} a possible value that can be taken by a state variable. We call a state variable \emph{discrete} or \emph{continuous} if the possibilities are integer or real numbers respectively. The topology used reflects the notion that integers can be measured perfectly (i.e. discrete topology) while real numbers can be measured up to a finite precision (i.e. standard topology on $\mathbb{R}$). From now on, we are going to consider state spaces whose states can be identified, at least within a region, by a finite set of discrete and continuous state variables (i.e. the state space is locally isomorphic to $\bigcup\limits_{1 \leq i \leq n} \mathbb{R}^{m_i}$).

We purposely use the term state variables instead of coordinates (even though that's what they are mathematically) as it would create confusion with space-time coordinates. We also avoid the term observable or measurable, as not all state variables may be directly physically tangible (e.g. conjugate momentum in a gauge theory). The only requirement for state variables is that they identify states. This means the possible values for state variables are physically distinguishable and are defined at equal time (since states are defined and mapped at a particular time).\footnote{This is the main reason that a quantity like velocity is not a suitable state variable, as it is  defined over an interval of time. Therefore velocity is always physically well defined but is not a state variable in general, while conjugate momentum is always a state variable but is not physically well defined by itself (it requires the vector potential to be specified as well). As we'll see later, when there exists a one to one map between velocity and conjugate momentum they can both be physically well defined state variables.}

\begin{prop}\label{prop:state_variable}
	A \emph{state variable} is a continuous map $q : U \rightarrow \mathbb{L}$ where $U \subseteq \mathcal{S}$ and $\mathbb{L}$ is the space for the possible values. If $\mathbb{L}\cong \mathbb{Z}$ the variable is said \emph{discrete}. If $\mathbb{L}\cong \mathbb{R}$ the variable is said \emph{continuous}.
\end{prop}

\begin{justification}
	We claim $\mathbb{L}$ is a topological space. $\mathbb{L}$ is a set of physically distinguishable possibilities. $\mathbb{L}$ is a topological space because of \ref{prop:topology}.
	
	We claim $q$ is a continuous map. $q$ is a map between two sets of physically distinguishable elements. $q$ is continuous because of \ref{prop:continuity}.
\end{justification}

When combining multiple state variables, it is important to understand how they relate to each other. Consider the orbital of an electron in a hydrogen atom, which is identified by the quantum numbers $n$, $l$, $m$ and $s$. For each combination of $n$, $l$ and $m$, the spin $s$ can have two values. The choice of spin is independent from the rest. The choice of $l$, though, depends on the choice of $n$: for $n=1$ only $l=0$ is available; for $n=2$ we can choose $l=0$ or $l=1$. The choices are not independent. That is: two or more variables are independent if there always exists a state for any possible combination, if the total number of states is the product of the possibilities of each variable. We call \emph{state vector} a collection of independent state variables that fully identify a state.

\begin{prop}\label{prop:independent_state_variables}
	Two state variables $q_1 : U \rightarrow \mathbb{L}_1$ and $q_2 : U \rightarrow \mathbb{L}_2$ are said \emph{independent} if $\exists \mathcal{s}$ such that $q_1(\mathcal{s})=l_1, q_2(\mathcal{s})=l_2 \forall l_1 \in q_1(U), l_2 \in q_2(U)$.
\end{prop}

In general, the entire state space may not be identified by a predetermined set of independent state variables. Consider the state of a pool table, determined by the number of balls together with position and momentum of the center of mass: the number of state variables is not fixed as it depends on the number of balls. But deterministic and reversible evolution cannot take us from a different number of continuous independent state variables (i.e. there is no homeomorphism between $\mathbb{R}^n$ and $\mathbb{R}^m$). That is the number of balls cannot change under deterministic and reversible evolution. Therefore we can restrict ourselves to the case where the number of continuous independent state variables is constant without loss of generality. This means that, at least locally, the state space is always homeomorphic to $\mathbb{R}^n$, and is therefore a manifold.

Discrete variables do not present such problems, as any number of them can be flattened out in a single one (i.e. $\mathbb{Z}^n$ is homeomorphic to $\mathbb{Z}$). They will determine the number of connected components of the state space. Once we introduce a continuous parameter for time evolution, though, these become irrelevant as continuous time evolution requires continuous trajectories that cannot move states across disconnected components. This justifies the special interest in path connected manifolds, as this is where trajectories for deterministic and reversible continuous evolution live.

\begin{prop}\label{prop:manifold}
	Let $\mathcal{s} \in \mathcal{S}$ a state within a state space. The set of states $\mathcal{S}'\subseteq\mathcal{S}$ potentially reachable from $\mathcal{s}$ by deterministic and reversible continuous evolution is a path connected manifold of dimension equal to the number of independent continuous state variables necessary to identify it.
\end{prop}

\begin{justification}
	We claim $\mathcal{S}'$ is a manifold. Let $\mathcal{s} \in \mathcal{S}$. Let $n$ be the number of independent continuous variables needed to identify $\mathcal{s}$. There exists a neighborhood $U$ around $\mathcal{s}$ where we have $(q_1,...,q_n):U\rightarrow \mathbb{R}^n$. This map is a bijection as $\mathcal{s}$ is identified by those variables. $\mathcal{S}$ is homeomorphic to $\mathbb{R}^n$ around  $\mathcal{s}$. Let $\mathcal{S}'$ be the set of all states potentially reachable by deterministic and reversible evolution from $\mathcal{s}$. Deterministic and reversible evolution is a homeomorphism between initial and final states. $\forall \hat{\mathcal{s}} \in \mathcal{S}'$ there exist a map $\mathcal{T}_{\Delta t}:\mathcal{S} \rightarrow \mathcal{S}$ such that $\hat{\mathcal{s}} =\mathcal{T}_{\Delta t}(\mathcal{s})$. $\hat{\mathcal{s}} \in \mathcal{T}_{\Delta t}(U)$ and $\mathcal{T}_{\Delta t}(U)$ is isomorphic to $\mathbb{R}^n$. $\mathcal{S}'$ equipped with the subspace topology is a topological space everywhere homeomorphic to $\mathbb{R}^n$. $\mathcal{S}'$ is a manifold of dimension $n$.
	
	We claim $\mathcal{S}'$ is path connected. Let $\hat{\mathcal{s}} \in \mathcal{S}'$. There exist map $\lambda : [t_0,t_1] \rightarrow \mathcal{S}'$ where $\lambda(t_0)=\mathcal{s}$, $\lambda(t_1)=\hat{\mathcal{s}}$ and $t_0$ and $t_1$ are the initial and final time respectively. $\lambda$ is a map between two physically distinguishable quantities, time and states. $\lambda$ is continuous by \ref{prop:continuity}. All $\hat{\mathcal{s}} \in \mathcal{S}'$ are path connected to $\mathcal{s}$. $\mathcal{S}'$ is path connected.
\end{justification}

\subsection{Evolution of state variables}

To study time evolution, we need to describe how state variables change under deterministic and reversible evolution. There are two ways to do it, and we'll need both. The first approach is to \emph{evolve} the state variables from the initial value $q^i(\mathcal{s})$ to the final value $q^i(\hat{\mathcal{s}})$. The result is a trajectory $q^i(t)=q^i(\lambda(t))$ which is especially useful when state variables correspond to physically meaningful quantities. For example, we track how the temperature or the pressure of an ideal gas changes. Evolved state variables, however, make it hard to find relationships that are invariant and common to all deterministic and reversible processes.

The second approach is to transport the state variables. The idea is to keep the connection to the initial state by labeling the future state by the original value of $q^i$, instead of by the future values. For example, the evolved state variable will not tell us the current value of pressure, but the one for the initial conditions. So we introduce a new set of variables for which $\hat{q}^i(\hat{\mathcal{s}})=q^i(\mathcal{s})$, which we can always do as the evolution is deterministic and reversible. The hypersurfaces at constant $\hat{q}^i$ are the evolved hypersurfaces at constant of $q^i$ (e.g. all the states that started with a particular value for pressure) which allows us to study how groups of states evolve in time. Given that the value of the transported state variables does not change during evolution, and that it is unique for each initial state, transported state variables also provide a way to label the trajectories themselves.

Evolved state variables are useful to write equations of motion, study how physical quantities change, form a physical picture of what happens. Transported state variables are useful to write invariants, study state space trajectories, form a geometric picture for the state space.\footnote{One should not confuse evolved/transported state variable with active/passive transformations or with Schroedinger/Heisenberg pictures. In those cases, the choice is between changing the state or the coordinates/observables. In our case, the state is always changing. The choice is between tracking the change with different values of the same state variable or different state variables that give the same value.}

\begin{prop}\label{prop:evolved_transported_variable}
Let $q^i$ be a set of state variables and $\mathcal{T}_{\Delta t}$ a deterministic and reversible evolution map. The evolved state variables are given by $q^i \circ \mathcal{T}_{\Delta t}$. The transported state variables are given by $q^i \circ \mathcal{T}_{-\Delta t}$.
\end{prop}

One reason that transported variables will be useful is that, during deterministic and reversible evolution, we need to make sure that the number of initial and final possibilities across independent state variables remains the same. This is more easily done using transported variables as the actual values do not change. In the discrete case, this is easy to show.

\begin{prop}\label{prop:discrete_measure}
	Let $U \subseteq \mathcal{S}$ a set of states fully identified by a set of $n$ independent discrete state variables $q^i$. Let $\Delta q^i \equiv q^i(U)$ the range of possibilities of each variable. Then $\#(\Delta q^i)=\#(\Delta \hat{q}^i) \; \forall i$ and $\#(U)=\prod\limits_{i=1}^{n}\#(\Delta q^i)=\prod\limits_{i=1}^{n}\#(\Delta \hat{q}^i)=\#(\hat{U})$ where $\hat{U}=\mathcal{T}_{\Delta t}(U)$ and $\#$ denotes the number of elements in the given set.
\end{prop}

\begin{justification}
	We claim $\#(\Delta q^i)=\#(\Delta \hat{q}^i)$. For each $\mathcal{s} \in U$, let $\hat{\mathcal{s}}=\mathcal{T}_{\Delta t}(\mathcal{s})$. We have $\hat{q}^i(\hat{\mathcal{s}}) = q^i(\mathcal{s})$. Therefore $\Delta q^i = q^i(U) = \hat{q}^i(\hat{U})=\Delta \hat{q}^i$. $\#(\Delta q^i)=\#(\Delta \hat{q}^i)$.
	
	We claim $\#(U)=\prod\limits_{i=1}^{n}\#(\Delta q^i)$. As $q^i$ are independent variables, the states are the Cartesian product of the possibilities of each variable.
\end{justification}

What happens is that, because the variables are independent, the total number of states is the product of the number of possibilities for each variable. If $\Delta q^1$ and $\Delta q^2$ are ranges of possibilities for two independent variables, the total number of possibilities is $\#(\Delta q^1) \#(\Delta q^2)$. By construction, the set $\Delta q^i$ of possibilities for each independent variable is the same as the set $\Delta \hat{q}^i$ of possibilities for the transported variable. Also, the transported variables remain independent. Therefore the relationship $\#(\Delta \hat{q}^1) \#(\Delta \hat{q}^2) = \#(\Delta q^1) \#(\Delta q^2)$ is valid throughout the evolution.

We'll see that very similar relationships are what define Hamiltonian and Lagrangian mechanics. But recovering them in the continuous case is not as straightforward. Consider the map $q'=aq$ with $0<a<1$. At first glance, it's a bijective continuous map so we may think it can represent a deterministic and reversible evolution. Yet, $\Delta q \equiv [-b, b] \supset \Delta q'$: a set of states is mapped to a proper subset (i.e. to fewer states) which does not make sense for a reversible process. In the limit where we apply the map an infinite amount of times, any value of $q$ will be brought infinitely close to $0$, which also does not sound like a reversible process.

The issue is that while deterministic evolution will map states to states, not all bijective state to state maps can be considered deterministic and reversible.\footnote{In fact, we have already seen that the map must at least be continuous.} This is a crucial point where this work is different from more usual treatments of dynamical systems. If one focuses only on the states themselves and counts them, it reaches the conclusion that a range of $1$m contains the same number of initial conditions of $1$Km: they have the same infinite number of elements. This conclusion has a number of problems: it does not match our physical intuition; it's not consistent with all types of topology; it does not explain why the laws of evolution are very often, not only continuous, but differentiable. In short: a bijection does not capture the entirety of the physical description.

If we work with a discrete topology then, by all means, the way to measure the size of a set is to count the elements in a set. But if we use the standard topology on $\mathbb{R}$, isolated points have measure zero (i.e. we can't physically distinguish them): what defines the number of initial conditions is the range of the state variables. As we'll see, the geometrical structures in Hamiltonian/Lagrangian classical/quantum mechanics are there to allow us to properly count states and, more specifically, possibilities along independent state variables. If we think that counting states is equivalent to counting elements in the set, we miss the whole point of that geometrical structure.

What makes the matter even more confusing is that the single physical idea of counting states and possibilities along independent state variables does not map to a single mathematical concept. Sometimes we need a measure, sometimes we need a symplectic form and sometimes we need a metric tensor. It will depend, case by case, on the state space of the system under study. But, however defined mathematically, deterministic and reversible motion must preserve that count. This is the crucial insight that gives physical meaning to those structures, that justifies the use of differential equations for state evolution, that tells us why symplectic and Riemannian geometry are so important in physics. Properly characterizing this count will be one of the key goals of the following sections.

\section{Composite systems and reducibility}
\label{sec:reducibility}

% Review started

Given that scientific reductionism (i.e. the idea of reducing physical systems and interactions to the sum of their constituent parts in order to make them easier to study) is at the heart of fundamental physics~\cite{Ruse}, we now explore how to characterize a system in terms of its components. That is, we want to study the relationship between the state space of a composite system and the state space of its parts. In general, this is quite a complicated thing to do, which requires intimate knowledge of the system at hand. So we simplify our problem and study a material made of infinitesimal homogeneous parts. What we'll find is that under the additional assumption that the material is infinitesimally reducible (i.e. its state is equivalent to the states of its infinitesimal parts) and that each part undergoes deterministic and reversible evolution, the motion is suitably described by the standard framework of classical Hamiltonian mechanics.

\subsection{Homogeneous decomposable systems and vector spaces}
The notion that a system is decomposable means the states are equipped with a rule of composition that allows one to write $\mathcal{c}=\mathcal{c}_1+\mathcal{c}_2$: the composite system is the sum of its components. For example, the state of a ball is equal to the state of its top and bottom parts.

The notion that the system is homogeneous means that the states of the composite and of each part are not unrelated: they are all made of the same material.\footnote{Whether a system is homogeneous depends on context (e.g. air can be thought as homogeneous if the mixture of gases does not change in space or in time due to phase transitions or chemical processes) and one must check that such a property is maintained by time evolution.} In fact, the state space $\mathcal{C}$ of all systems composed of such material will include the state of the system as well as the states of  its parts (i.e. $\mathcal{c}, \mathcal{c}_1, \mathcal{c}_2 \in \mathcal{C}$). Since combining any two systems made of a homogeneous material will always give us a system made of the same homogeneous material (e.g. combining elements made of water gives us another element made of water), the state space $\mathcal{C}$ is closed under composition.

As we study the system under evolution, we'll also want to study state changes. For example, if a stable mixture of gas expands, we'll want to know the difference between the initial and final distributions. That is: $\delta\mathcal{c}=\hat{\mathcal{c}}-\mathcal{c}$. Note that such a state difference may not describe a physical state: it may remove some material from one location to add it somewhere else. Yet, these state changes are still physically distinguishable objects that provide a configuration for the same homogeneous material, therefore we extend the state space $\mathcal{C}$ to include state differences.\footnote{Note that one can characterize a state difference without knowing the states themselves. If we take some amount of material out of a container we know how the state changed, yet we may not know the amount of material before or after. In certain cases one may be interested in measuring the asymmetry of a particular state without measuring (or without being able to measure) the full state itself.} All combined, this gives us the structure of an abelian group.

\begin{prop}\label{prop:abelian_group}
The state space $\mathcal{C}$ for a decomposable homogeneous material is an additive abelian (i.e. commutative) group.
\end{prop}

\begin{justification}
We claim $\mathcal{C}$ is an additive monoid. There exists a law of composition $+ : \mathcal{C} \times \mathcal{C} \rightarrow \mathcal{C}$ that takes two states and returns one that is the physical composition of the two. The domain and codomain match because the material is homogeneous. The law is commutative $\mathcal{c}_1 +\mathcal{c}_2 = \mathcal{c}_2+\mathcal{c}_1$ and associative $(\mathcal{c}_1 + \mathcal{c}_2) + \mathcal{c}_3 = \mathcal{c}_1 + (\mathcal{c}_2 + \mathcal{c}_3)$, as it does not matter in what order we physically compose the components. There exists a unique zero element $\mathcal{c} + 0 = \mathcal{c}$ and it represents the physically empty state (i.e. no amount of material). $\mathcal{C}$ is an additive monoid by definition.

We claim $\mathcal{C}$ is an additive group. We require $\mathcal{C}$ to include state changes. A change of a physically distinguishable object is itself physically distinguishable therefore $\mathcal{C}$ is still a Hausdorff topological manifold as per \ref{prop:topology}. There exists an inverse $- : \mathcal{C} \rightarrow \mathcal{C}$ such that $\mathcal{c} + ( - \mathcal{c}) = 0 \forall \mathcal{c} \in \mathcal{C}$. $\mathcal{C}$ is an additive group by definition. 
\end{justification}

\subsection{Classical material and distributions}
We also want to be able to express the state of the composite system in terms of the state of each part. For example, given the state of a fluid we'll want to know how the material is distributed in space. This in general will depend on how much the state of the composite system ``knows" about its parts. For example, the position and orientation of an ideal rigid body is enough to define where all its constituents are, while the volume/pressure/temperature of an ideal gas is not enough to determine the position and momentum of all its constituents. Therefore we need to characterize the system further.

We will call a \emph{particle} the smallest amount of a material to which we can assign an independent state. For example, a photon is a particle of light (it is the smallest amount of light we can describe), an infinitesimal amount of water is treated classically as a point particle (the smallest amount for a continuous fluid).\footnote{We use this definition as it allows us to retroactively talk on somewhat equal grounds about classical particles and quantum particles, and underscore how physically (and mathematically) they are very similar yet very different objects.}

We will call a \emph{classical material} one that is homogeneous and infinitesimally reducible. That is, we can keep decomposing the system indefinitely into smaller and smaller parts, each with a well defined state. In this case, the particles are the infinitesimal parts given by the limit of this process of recursive reduction. Giving the state of the whole system is equivalent to giving the states of all particles. Given the state space $\mathcal{S}$ for the infinitesimal parts, which we assume consistently with \ref{prop:manifold} to be a manifold, a composite state $\mathcal{c} \in \mathcal{C}$ will tell us the amount of material for each possible particle state. That is, each state is fully identified by a function $\rho_\mathcal{c}: \mathcal{S} \to \mathbb{R}$ that returns the amount of material in the composite state $\mathcal{c}$ that is prepared in the particle state $\mathcal{s}$. For example, given a certain configuration of gas, we can tell the density of material that is at a specific point in space with a specific value of momentum.

The notion that the parts are infinitesimal requires the co-domain of $\rho_\mathcal{c}$ to be a real number, as opposed to an integer. The state space $\mathcal{S}$ could, instead, be a discrete set. For example, for a system composed of different tanks connected by pipes, the state of the composite system could be the overall distribution of water among the tanks. The amount in each tank is continuous (as we assume the water to be infinitesimally divisible) yet there are only a finite number of tanks the water can be placed into. As it links two physically distinguishable objects, $\rho_\mathcal{c}$ is a continuous function as discussed in \ref{prop:continuity}.

As we combine the states of different parts, we sum the distributions over particle states $\rho_{\mathcal{c}}(\mathcal{s})=\rho_{\mathcal{c}_1}(\mathcal{s})+\rho_{\mathcal{c}_2}(\mathcal{s})$: the amount of material in the composite state is the sum of the amount of material of the parts. We can also increase or decrease the amount of material by a constant factor $\rho_{\mathcal{c}_1}(\mathcal{s})=a\rho_{\mathcal{c}_2}(\mathcal{s})$. This gives us the structure of a real vector space that is isomorphic to a subspace of continuous functions.

\begin{prop}\label{prop:real_vector_space}
The state space $\mathcal{C}$ for a classical (i.e. homogeneous infinitesimally reducible) material is a vector space over $\mathbb{R}$ isomorphic to a subspace of the space of continuous functions. That is $\mathcal{C} \cong G \subseteq C(\mathcal{S}) \equiv \{\rho:\mathcal{S} \rightarrow \mathbb{R} \; | \; \rho$ is continuous$\}$, where $\mathcal{S}$ is the state space of an infinitesimal amount of material.
\end{prop}

\begin{justification}
We claim $\mathcal{C}$ is an abelian group. $\mathcal{C}$ is the state space for a decomposable homogeneous material and is therefore an abelian group  by \ref{prop:abelian_group}.

We claim $T$, the set of transformations that increase or decrease the amount of material in the system by a constant factor, is a field\footnote{Here field is intended in the abstract algebraic sense (a nonzero commutative division ring) which has no relationship to the field in the physics or differential geometry sense (a physical quantity/tensor with a value for each point in space).} isomorphic to $\mathbb{R}$. Consider $\tau: \mathbb{R} \rightarrow T$ the mapping between a number and the transformation that increases or decreases the amount of material by that factor. This transformation exists: the system is infinitesimally decomposable and the amount can be changed continuously. Define on $T$ an addition $+: T \times T \rightarrow T$ and a multiplication $*: T \times T \rightarrow T$ such that $\tau(a) + \tau(b) = \tau(a+b)$ and $\tau(a) * \tau(b) = \tau(a*b)$, $a,b \in \mathbb{R}$, so that the sum and product of the transformation is equal to the sum and product of their respective factors. $\tau$ is an isomorphism between $T$ and $\mathbb{R}$ as fields.

We claim $\mathcal{C}$ is a vector space over $\mathbb{R}$. The abelian group $\mathcal{C}$ can be extended with the operations defined by $T$, as each element $\tau \in T$ is a map $\tau : \mathcal{C} \rightarrow \mathcal{C}$. The map has the following properties: $(\tau_1 + \tau_2) \, \mathcal{c} = \tau_1 \mathcal{c} + \tau_2 \mathcal{c} \; \forall \tau_1, \tau_2 \in T$ and $\mathcal{c} \in \mathcal{C}$, increasing the amount of material by the sum of two constant factors is the same as combining the separate increases, and $\tau \, (\mathcal{c}_1 + \mathcal{c}_2) = \tau \mathcal{c}_1 + \tau \mathcal{c}_2\; \forall \tau \in T$ and $\mathcal{c}_1, \mathcal{c}_2 \in \mathcal{C}$, increasing the amount of the total system is the same as the combination of the increased parts. $\mathcal{C}$ is a module over $T$, which is a field and isomorphic to $\mathbb{R}$. $\mathcal{C}$ is (isomorphic to) a real vector space.

We claim $\mathcal{C}$ is isomorphic to a subspace of $C(\mathcal{S})$ as a vector space over $\mathbb{R}$. $\forall \mathcal{c} \in \mathcal{C} \; \exists ! \rho_{\mathcal{c}}:\mathcal{S} \rightarrow \mathbb{R}$ returning the amount of material for each state $\mathcal{s} \in \mathcal{S}$. $\rho_{\mathcal{c}} \in C(\mathcal{S})$ by \ref{prop:continuity}. Let $\varrho : \mathcal{C} \rightarrow C(\mathcal{S})$ such that $\varrho(\mathcal{c}) \mapsto \rho_\mathcal{c}$. 
As the system is reducible, two distinct composite states must represent different distributions. $\forall \mathcal{c_1}, \mathcal{c_2} \in \mathcal{C}, \mathcal{c_1} \neq \mathcal{c_2} \implies \varrho(\mathcal{c_1}) \neq \varrho(\mathcal{c_2})$. $\varrho$ is injective. $\varrho$ is a bijection between $\mathcal{C}$ and $\varrho(\mathcal{C})$. Let $\mathcal{c}=\mathcal{c}_1+\mathcal{c}_2$, then $\varrho(\mathcal{c})=\varrho(\mathcal{c}_1)+\varrho(\mathcal{c}_2)$ as the amount of material of the composition of two states is the sum of the individual amounts. Let $\mathcal{c}_1=\tau(a)\mathcal{c}_2$, then $\varrho(\mathcal{c}_1)=a \varrho(\mathcal{c}_2)$ as $\tau(a)$ increases the amount of material by the factor $a$. $\varrho$ is a homomorphism as a vector space. $\mathcal{C}$ is isomorphic to $\varrho(\mathcal{C}) \subseteq C(\mathcal{S})$.
\end{justification}

\subsection{Integration and measure}

When the topology of $\mathcal{S}$ is not discrete, the distribution $\rho_\mathcal{c}$ is a density, which is not something we directly measure. What we do measure experimentally are finite amounts of material. For example, the amount of fluid in a particular volume within a particular range of momentum specified in some units (e.g. moles, kg, ...). This means that given a composite state $\mathcal{c}$ and a set $U \subseteq \mathcal{S}$ of particle states compatible with an outcome of a process, we must be able to tell the amount of material that we will find associated with that outcome. That is for each open set $U \subseteq \mathcal{S}$ there will be a functional $\Lambda_U : \mathcal{C} \rightarrow \mathbb{R}$.

As we combine parts, or increase the amount of material in each part, the total amount of material found will have to be consistent with those operations. That is: $\Lambda_U(a_1 \mathcal{c}_1 + a_2 \mathcal{c}_2) = a_1 \Lambda_U(\mathcal{c}_1) + a_2 \Lambda_U(\mathcal{c}_2)$. For a proper state, one that is not the difference of two states, we will expect a positive amount of material. So, if a density $\rho_\mathcal{c}$ is positive everywhere, then the total amount is also positive. This makes $\Lambda_U$ a positive linear functional. As the amount of material we measure is always finite, the functional applied to any composite state $\mathcal{c}$ over any set $U$ will be finite.

The ability to associate finite amounts of material with sets of states allows us to define finite regions of $\mathcal{S}$ and compare them. Consider a region of position and momentum in phase space. If we are able to spread a finite amount of material into a non-infinitesimal uniform distribution, then we know we have a finite region. And the density will give us an indication of how big the region is: if we double the region, the density will halve. In other words: because we are describing densities over $\mathcal{S}$, we are able to give a unique measure\footnote{Here measure is intended in the mathematical sense (a real valued function of a $\sigma$-algebra) which is distinct from other connotations in physics.} $\mu$ that allows us to count the number of states in each set of $\mathcal{S}$. In terms of said measure, positive linear functionals become integrals, $\Lambda_U (\mathcal{c}) = \int_U \rho_{\mathcal{c}} d \mu$.

This intuitive picture is formalized mathematically by the Riesz representation theorem for linear functionals~\cite{Halmos}, which gives $\mathcal{S}$ the structure of a measure space, with a suitable Borel $\sigma$-algebra and measure $\mu$.

\begin{prop}\label{prop:integration}
	The state space $\mathcal{S}$ for the particles of a classical material is endowed with a Borel measure $\mu$. The state space $\mathcal{C}$ for a classical material is isomorphic to a subspace of the space of Lebesgue integrable functions. That is $\mathcal{C} \cong G \subseteq L^1(\mathcal{S}, \mu) = \{ \rho : \mathcal{S} \rightarrow \mathbb{R} \; | \; \int_{\mathcal{S}} |\rho| d\mu < \infty \}$.
\end{prop}

\begin{justification}
	We claim there exists a positive linear functional $\Lambda_U : \mathcal{C} \rightarrow \mathbb{R}$ for each $U \subseteq \mathcal{S}$ such that $\Lambda_U = \Lambda_{\mathrm{int}(U)}$ where $\mathrm{int}(U)$ is the interior of $U$. Let $\Lambda_{U} : \mathcal{C} \rightarrow \mathbb{R}$ be the functional that returns the amount of material in $\mathcal{c}$ compatible with the outcome associated with the open set $U \subseteq \mathcal{S}$. $\Lambda_U$ is well-defined: $U$ is associated with a physically distinguishable outcome. $\Lambda_U$ is linear: $\Lambda_U(a_1 \mathcal{c}_1 + a_2 \mathcal{c}_2) = a_1 \Lambda_U(\mathcal{c}_1) + a_2 \Lambda_U(\mathcal{c}_2)$ as it has to be consistent with the operations of composing states and increasing/decreasing the amount of material by a factor. $\Lambda_U$ is positive: if the value of the distribution for each particle state is positive then the total amount of material is positive. Let $U \subseteq \mathcal{S}$ not necessarily open. Define $\Lambda_U$ as $\Lambda_{\mathrm{int}(U)}$.
	
	We claim that $|\Lambda_{U}(\mathcal{c})| < \infty \; \forall \mathcal{c} \in \mathcal{C} \; \forall U \subseteq \mathcal{S}$. Let $U \subseteq \mathcal{S}$ be an open set of particle states associated with an outcome. Let $\mathcal{c} \in \mathcal{C}$ be a composite state. The amount of material of $\mathcal{c}$ associated with the outcome $U$ must be finite, as physically we always work with finite quantities. $|\Lambda_{U}(\mathcal{c})| < \infty$. Let $U \subseteq \mathcal{S}$ not necessarily open. $|\Lambda_{U}(\mathcal{c})| = |\Lambda_{\mathrm{int}(U)}(\mathcal{c})| < \infty$.
	
	We claim that $\Lambda_{U_1} + \Lambda_{U_2} = \Lambda_{U_1 \cup U_2} + \Lambda_{U_1 \cap U_2}$ $\forall $ $U_1, U_2 \in \mathcal{S}$. Let $U_1, U_2 \in \mathcal{S}$ be open sets of particle states associated with two outcomes. Suppose $U_1 \cap U_2 = 0$, $\Lambda_{U_1 \cup U_2} = \Lambda_{U_1} + \Lambda_{U_2}$ as the material found in $U_1 \cup U_2$ must be either in $U_1$ or $U_2$. Suppose $U_1 \cap U_2 \neq 0$, $\Lambda_{U_1 \cup U_2} = \Lambda_{U_1} + \Lambda_{U_2} - \Lambda_{U_1 \cap U_2}$ as the sum of the material associated with each outcome will double count the intersection. Let $U_1, U_2 \in \mathcal{S}$ not necessarily open. $\Lambda_{U_1} + \Lambda_{U_2} = \Lambda_{\mathrm{int}(U_1)} + \Lambda_{\mathrm{int}(U_2)} = \Lambda_{\mathrm{int}(U_1) \cup \mathrm{int}(U_2)} + \Lambda_{\mathrm{int}(U_1) \cap \mathrm{int}(U_2)} = \Lambda_{U_1 \cup U_2} + \Lambda_{U_1 \cap U_2}$
	
	We claim that $\mathcal{S}$ is endowed with a unique Borel measure $\mu$ such that $\Lambda_U (\mathcal{c}) = \int_U \rho_{\mathcal{c}} d \mu$.  $\mathcal{S}$ is locally compact as it is a manifold. $\mathcal{C} \cong G \subseteq C(\mathcal{S})$ therefore $\Lambda_U(\mathcal{c}) \cong \Lambda_U(\rho_\mathcal{c})$. $\Lambda = \{\Lambda_U : C(\mathcal{S}) \rightarrow \mathbb{R}\}_{U \subseteq \mathcal{S}}$ is a family of positive linear functionals such that $\forall U \subseteq \mathcal{S} \; \Lambda_U = \Lambda_{\mathrm{int}(U)}$ and $\Lambda_{U_1} + \Lambda_{U_2} = \Lambda_{U_1 \cup U_2} + \Lambda_{U_1 \cap U_2} \; \forall U_1, U_2 \subseteq \mathcal{S}$. By the extension \ref{extended_riesz_theorem} of the Riesz representation theorem for linear functionals 
	there exists a unique Borel measure $\mu$ such that $\Lambda_U (\mathcal{c}) = \int_{U} \rho_\mathcal{c} d\mu$.

	We claim $\mathcal{C}$ is isomorphic to a subspace of $L^1(\mathcal{S}, \mu)$ as a vector space over $\mathbb{R}$. $\forall \mathcal{c} \in \mathcal{C} \exists ! \rho_{\mathcal{c}}:\mathcal{S} \rightarrow \mathbb{R}$ returning the amount of material for each state $\mathcal{s} \in \mathcal{S}$. As $|\Lambda_{U}(\mathcal{c})| < \infty \; \forall \mathcal{c} \in \mathcal{C} \; \forall U \subseteq \mathcal{S}$, then $\rho_{\mathcal{c}} \in \mathcal{L}(S,\mu) \equiv \{ \rho : S \rightarrow \mathbb{R} \; | \;\; |\int_{U} \rho d\mu| < \infty \; \forall U \subseteq S\}$. $\mathcal{L}(S,\mu) = L^1(\mathcal{S}, \mu)$ by \ref{everywhere_integrable_is_lebesgue_integrable}. Let $\varrho : \mathcal{C} \rightarrow L^1(\mathcal{S}, \mu)$ such that $\varrho(\mathcal{c}) \mapsto \rho_\mathcal{c}$. $\varrho$ is a homomorphism, as justified in \ref{prop:real_vector_space}. $\mathcal{C}$ is isomorphic to $\varrho(\mathcal{C}) \subseteq L^1(\mathcal{S}, \mu)$ as a vector space.
\end{justification}

\subsection{Invariant densities and differentiability}

On a state space with countable elements, using a single state variable, we have $\Lambda_U (\mathcal{c}) = \sum \limits_{q=a}^b \rho_\mathcal{c}(q)$. We would expect the expression to become $\Lambda_U (\mathcal{c}) = \int_a^b \rho_\mathcal{c} (q) dq$ for states identified by a continuous state variable. By changing state variable, though, we see that this does not work in general: $\rho_\mathcal{c}(\hat{q})= \rho_\mathcal{c}(q) dq/d\hat{q}$. For example, if we changed units from $g/m$ to $g/Km$, the value of the density would increase by a factor of $1000$. This makes $\rho_\mathcal{c}$ a function not just of the state $\mathcal{s}$, but also of the state variables we are using. This is inconsistent with our previous definition. Moreover if the transformation is not differentiable, the density is not even well-defined. This tells us that $\mathcal{S}$ cannot be any manifold: it has to be one that allows state-variable-independent densities.

The first order of business is to guarantee that the density remains defined under an arbitrary change of state variables. This means the Jacobian of the transformation must exist. For the Jacobian to exist the transformation between state variables has to be differentiable. Therefore a manifold that guarantees densities to be well defined is one that guarantees that state variable transformations are differentiable: a differentiable manifold. $\mathcal{S}$ has a differential structure in the sense that the distributions and the volume element $d\mu$ are only defined on a set of state variables that are linked by differentiable transformations. In the same way that discontinuous changes of state variable can be detected because they tamper with physical distinguishability (i.e. the topology), non differentiable changes of state variable can be detected because they tamper with our ability to define state-variable-independent densities and count states.

Moreover, the density itself has to be differentiable. The distribution $\rho_\mathcal{c}$ is a real valued function of the state. As such, at least locally, we can use it as a state variable $q^1=\rho_\mathcal{c}$: we can identify particle states by the density of the material associated with them. This is more physically meaningful than one may first suspect, since placing physical markers (i.e. placing material at particular states) is a common way to define the references for a coordinate system. In a coordinate system where $q^1=\rho_\mathcal{c}$, the function is clearly differentiable as $\partial_1 \rho_\mathcal{c} = 1$ and $\partial_i \rho_\mathcal{c} = 0$ for $i \neq 1$. As we change state variables, the transformation is differentiable, therefore all the derivatives $\partial_{i'} q^1 = \partial_{i'} \rho_\mathcal{c}$ exist. The density itself is differentiable.

\begin{prop}\label{prop:differentiable_manifold}
	The state space $\mathcal{S}$ for the particles of a  classical material is a differentiable manifold. The state space $\mathcal{C}$ for a classical material is isomorphic to the space of Lebesgue integrable differentiable functions. That is $\mathcal{C} \cong C^1(\mathcal{S}) \cap L^1(\mathcal{S}, \mu)$, where $C^1(\mathcal{S}) \equiv \{\rho:\mathcal{S} \rightarrow \mathbb{R} \; | \; \rho$ is differentiable$\}$.
\end{prop}
\begin{justification}
	We claim $\mathcal{S}$ is a differentiable manifold. Let $\rho_\mathcal{c} : \mathcal{S} \rightarrow \mathbb{R}$ be the distribution associated with a composite state $\mathcal{c} \in \mathcal{C}$. Let $q^i$ and $\hat{q}^j$ be two sets of independent state variables that fully identify states within $U\subseteq \mathcal{S}$. Let $\rho_\mathcal{c}(q^i)$ and $\rho_\mathcal{c}(\hat{q}^j)$ be the expression in local coordinates of the distribution.  $\rho_\mathcal{c}(\hat{q}^j)=\rho_\mathcal{c}(q^i) | \partial_j q^i |$ as $\rho_\mathcal{c}$ transforms as a density. $\rho_\mathcal{c}(\mathcal{s}) = \rho_\mathcal{c}(q^i(\mathcal{s})) = \rho_\mathcal{c}(\hat{q}^j(\mathcal{s}))$. The Jacobian $| d q' / d q |$ exists and is non-zero. The map between any two charts is differentiable. $\mathcal{S}$ is a differentiable manifold.
	
	We claim that the distributions associated with composite states are differentiable in the regions where they are strictly monotonic. Let $\rho_\mathcal{c}$ be the distribution associated with state $\mathcal{c}$. Let $\rho_\mathcal{c}$ be strictly monotonic in a region $U\subseteq \mathcal{S}$. Then its fibers in $U$, the inverse images of the values, are connected hypersurfaces. We can construct a chart such that $q^1=\rho_\mathcal{c}$. The distribution is differentiable over $U$ for that variable: $\partial_{i} \rho_\mathcal{c} = \partial_{i} q^1 = \delta_i^1$ where $\delta_i^j$ is the Kronecker delta. Let $q^{j}$ be another set of state variables. $q^1(q^{j})$ is differentiable as the manifold is differentiable. $\rho_\mathcal{c}(q^{j})=q^1(q^{j})$ is differentiable over $U$.
	
	We claim that the distributions associated with composite states are differentiable. Let $\rho_\mathcal{c}$ be the distribution associated with state $\mathcal{c}$. Let $\rho_{\mathcal{c}_0}$ be the distribution associated with state $\mathcal{c}_0$. Let $\rho_{\mathcal{c}_0}$ be strictly monotonic in a region $U\subseteq \mathcal{S}$. Let $\delta \rho = \rho_{\mathcal{c}_0} - k \rho_{\mathcal{c}}$ where $k \in \mathbb{R}$ and $k>0$. Let $k$ be sufficiently small so that $\delta \rho$ is strictly monotonic over $U$. $\rho_{\mathcal{c}_0}$ and $\delta \rho$ are differentiable over $U$ as they are strictly monotonic over $U$. $\rho_{\mathcal{c}}$ is a linear combination of differentiable functions over any arbitrary region $U$. $\rho_{\mathcal{c}}$ is differentiable.
	
	We claim $\mathcal{C} \cong C^1(\mathcal{S}) \cap L^1(\mathcal{S}, \mu)$. Let $\rho_\mathcal{c}$ be the distribution associated with a state $\mathcal{c}$. Let $\rho \in C(\mathcal{S}) \cap  L^1(\mathcal{S}, \mu)$ be a Lebesgue integrable continuous function. If $\rho_\mathcal{c} \neq \rho$ then $\exists U \in \mathcal{S}$ such that $\int_{U} \rho_\mathcal{c} d \mu \neq \int_{U} \rho d \mu$. Both integrals are finite as $\rho_\mathcal{c}, \rho \in L^1(\mathcal{S}, \mu)$. Both integrals are expressible with the same state variables as all elements are defined on the same differential structure. Therefore there exists an outcome for $\mathcal{S}$ that can physically distinguish the two distributions. There must be a state $\mathcal{c}_1 \in \mathcal{C} \; | \; \varrho(\mathcal{c}_1)=\rho$. $\varrho : \mathcal{C} \rightarrow C(\mathcal{S}) \cap  L^1(\mathcal{S}, \mu)$ is surjective. $\varrho$ is a homomorphism, as justified in \ref{prop:real_vector_space}. $\varrho$ is an isomorphism between $\mathcal{C}$ and $C(\mathcal{S}) \cap  L^1(\mathcal{S}, \mu)$ as vector spaces.
\end{justification}

As we identified the space of Lebesgue integrable differentiable functions $C^1(\mathcal{S}) \cap  L^1(\mathcal{S}, \mu)$ as the space of distributions that are physically meaningful, we can better understand why other commonly used function spaces do not fit the bill. Some are not restrictive enough. The set of Lebesgue integrable functions $L^1(\mathcal{S})$ includes discontinuous functions that for \ref{prop:continuity} are unphysical. The space of continuous functions $C(\mathcal{S})$, the space of continuous functions that vanish at infinity $C_0(\mathcal{S})$ and the space of differentiable functions $C^1(\mathcal{S})$ include functions whose integral is infinite, which would represent an infinite amount of material. Some definitions are too restrictive. Requiring compact support (i.e. the function is different from zero only on a finite region) would exclude distributions, such as Gaussians, that span over the whole range of states. Such distributions are physical: if we take a finite volume of an ideal gas at equilibrium, the momentum distribution spans over the whole range. Schwartz space $S(\mathcal{S})$ excludes functions that are not infinitely smooth which we don't have a general physical justification for.\footnote{One can, though, make the argument that given any distribution $\rho$ one can find an infinitely smooth $\rho_{sm}$ such that the difference in description is small compared to the error already introduced by assuming the system to be homogeneous and infinitely reducible. The Whitney approximation theorem~\cite{Lee} makes this mathematically well defined.}

Another consideration is that while the norm associated to $L^1(\mathcal{S})$ is\journal{\break} $\int_{\mathcal{S}} |\rho| d\mu$, the expression $\int_{\mathcal{S}} \rho d\mu$ may be more physically meaningful in some cases. For proper states, the two are the same and they represent the total amount of material. For state changes they differ. The second represents the amount of material added (or taken away if negative) by the state change. For example, if $\delta \rho$ is the change due to deterministic and reversible evolution, the amount of material does not change and therefore $\int_{\mathcal{S}} \delta \rho \, d\mu = 0$. The first norm would represent the total amount of material that is changing (i.e. being added and being removed), therefore $\int_{\mathcal{S}} | \delta \rho | \, d\mu = 0$ means no material is moving, no change is occurring.

Note that neither of those expressions matches the norm given by the inner product defined as $\langle f, g \rangle = \int_{\mathcal{S}} fg \, d\mu$. Moreover $\mathcal{C}$, since it excludes discontinuous functions, can never be a complete normed space: it will not include the limit for all Cauchy sequences; it cannot be a Banach or a Hilbert space. Such a construction, though, is still useful. Consider the following expressions:
\begin{align*}
\Lambda_{\mathcal{S}} (\mathcal{c}) &= \int_\mathcal{S} \rho_{\mathcal{c}} d \mu = \langle 1 , \rho_{\mathcal{c}} \rangle \\
\Lambda_U (\mathcal{c}) &= \int_U \rho_{\mathcal{c}} d \mu = \langle 1_U , \rho_{\mathcal{c}} \rangle \\
\rho_\mathcal{c}(q_0,p_0) &= \int_\mathcal{S} \delta_{q_0,p_0} \cdot \rho_{\mathcal{c}} d \mu = \langle \delta_{q_0,p_0} , \rho_{\mathcal{c}} \rangle \\
p_{total}(\mathcal{c}) &= \int_\mathcal{S} p \cdot \rho_{\mathcal{c}} d \mu = \langle p , \rho_{\mathcal{c}} \rangle \\
q_{avg}(\mathcal{c}) &= \frac{\int_\mathcal{S} q \cdot \rho_{\mathcal{c}} d \mu}{\int_\mathcal{S} \rho_{\mathcal{c}} d \mu} = \frac{\langle q , \rho_{\mathcal{c}} \rangle}{\langle 1 , \rho_{\mathcal{c}} \rangle}
\end{align*}
For each state, they represent respectively the total amount of material, the material within $U$, the density at $(q_0, p_0)$, the total momentum and the average of a state variable. As they are linear functionals of $\mathcal{C}$,  they can be expressed using their dual vector. But these vectors ($1_U$, $p$, ...) are \emph{not} elements of $\mathcal{C}$. They do not represent states, they are not continuous and integrable, and there is no single general physical meaning for all of them. Therefore, as we mentioned before, while we can extend the function space for convenience, computation purposes or to study a limit case, we need to be mindful of the extension and carefully check that the mathematical results are physically meaningful case by case.

\subsection{Invariant densities and cotangent bundle}

Now that we are guaranteed that the density is properly defined for all state variables, we have to make sure it is invariant. The distribution should only depend on the states, and not the particular choice of state variables to label them.

If $\rho_{\mathcal{c}}$ is to remain invariant under state variable changes, not only does the Jacobian have to exist but it must be unitary. This means we cannot simply change one state variable as we please. If we change one at least another has to change in some coordinated way such that the Jacobian of the total transformation is unitary. This  means that we cannot change physical units of all the variables as we like: they are part of a unit system.

The simplest case is when a single variable is enough to define our units, and therefore tell how all the others must change. Suppose that our state space $\mathcal{S}$ is identified by $n$ independent state variables $\{q, k^i\}$ and a change of units for the first, that is $\hat{q} = \hat{q}(q)$, determines the change for all others. The Jacobian matrix is the block matrix:
\begin{align*}
J =  \left[
\begin{array}{cc}
d_q\hat{q} & 0 \\
\partial_q \hat{k}^j & \partial_i \hat{k}^j \\
\end{array}
\right] 
\end{align*}
The Jacobian has to be unitary, therefore all elements $\partial_i \hat{k}^j$ must be well defined and $|\partial_i \hat{k}^j| = 1/d_q\hat{q}$. As this can only happen if $\partial_i \hat{k}^j$ is a $1\times1$ matrix, there can only be one $k$. This gives $\partial_{k} \hat{k} = d_{\hat{q}} q$ and $\hat{k} = d_{\hat{q}} q \, k + f(q)$. As this is a change of units, we expect the zero value of $k$ to remain constant: $\hat{k}(q, 0) = 0 = f(q)$. Therefore $\hat{k} = d_{\hat{q}} q \, k$ transforms as the component of a covariant vector. $\mathcal{S}$ is isomorphic to $T^*\mathcal{Q}$, the space of co-vectors at a point, where $\mathcal{Q}$ is the manifold that defines our unit.

Physically, this means that the state variable $k$, which is the classical analogue of the wave number, is expressed with the inverse of the unit used for $q$. If $q$ is in meters, $k$ is in inverse meters. Consider now the area $dq \wedge dk$ of an infinitesimal rectangular region. This quantity is dimensionless and therefore invariant. The total number of states will be proportional to it, as  doubling the range of $dq$ or $dk$ will give us double the number of possible states. We have $d\mu = \hbar dq \wedge dk$, where $\hbar$ is the proportionality constant that will depend on the unit chosen to count the possibilities of the pair $(q, k)$. That is, a unitary range of possibilities for $q$ and $k$ will give us $\hbar$ possibilities.\footnote{The actual value and physical dimensions of $\hbar$ are determined by the system of units, and should not be taken to describe some intrinsic physical property. Only dimensionless relationships to other physical constants, such as the fine structure constant $\alpha = e^2/\hbar c$, can possess that trait. That is why one can choose ``natural units" for which $\hbar=1$. In SI units the relationship $\hbar dk = dp = m du$ between kinetic and conjugate momentum, derived later, sets the relationship between units.} The state count in a finite (i.e. compact) region $U \subseteq \mathcal{S}$ will be given by $\mu(U) = \int_U \hbar dq \wedge dk$: we are finally able to write the volume of integration in terms of state variables.

We can generalize to the case where $\mathcal{Q}$ has more than one dimension. As we must be able to change one unit at a time, independently from the other, to each $q^i$ will correspond a $k_i$ that uses the inverse of the corresponding units. We have $\mathcal{S}\cong T^*\mathcal{Q}$ and $d\mu = \hbar^n dq^n \wedge dk_n$. But it's not just the volume, the state count, that is preserved as we change state variables. Each pair $(q^i, k_i)$ is expressed by an independent unit with the count of possibilities given by $\hbar dq^i \wedge dk_i$. We call such a pair a degree of freedom. As they are independent the total state count is the product of the possibilities of each d.o.f.: $d\mu = \hbar^n dq^n \wedge dk_n = \prod \limits_{i=1}^n \hbar dq^i \wedge dk_i$. In other words: independent d.o.f. are orthogonal surfaces in $\mathcal{S}$. The possibility count of an arbitrary degree of freedom, then, will be $\sum \limits_{i=1}^n \hbar dq^i \wedge dk_i$, the sum of the projections over the $n$ orthogonal and independent d.o.f. defined by $(q^i, k_i)$.

The possibility count for each d.o.f. (i.e. the wedge product within a d.o.f.) and the orthogonality of different d.o.f. (i.e. the scalar product across d.o.f.) must be the same regardless of the choice of state variables. We can express both requirements mathematically in a compact way. We first define conjugate momentum as $p_i=\hbar k_i$ and unified state variables as $\xi^a\equiv \{q^i, p_i\}$. Then we consider the canonical symplectic two-form $\omega =\Sigma \, dq^i \wedge dp_i$ given by the following components:
\begin{align*}
\omega_{ab} =  \left[
\begin{array}{cc}
0 & 1 \\
-1 & 0 \\
\end{array}
\right] \otimes I_n =
\left[
\begin{array}{cc}
0 & I_n \\
-I_n & 0 \\
\end{array}
\right] \\
\end{align*}
It returns the wedge product within a d.o.f. and the scalar product across. Requiring the invariance of this metric under state variable changes assures us the count of states and possibilities is well defined.

\begin{figure}
		\begin{tikzpicture}[scale=.4]
	
	% INPUTS
	
	\newcommand{\xvar}{5};					% x component of omega
	\newcommand{\yvar}{6};					% y component of omega	
	\newcommand{\xprimevar}{6.4};	    % length of the bottom side of the second triangle, called x prime
	% Note: must have  x' < omega = sqrt(x^2 + y^2) 
	
	\newcommand{\dvar}{20}; 				% space between the two drawings
	
	% outer rectangle and diagram name
	\draw (-8,-8) rectangle (34,12.5);
	
	% outer rectangle and diagram name, alternate method, needs more calculations...
	%		\draw (-\yvar - 2,-12) rectangle (\dvar + \dvar, \yvar + \xvar + 2);
	%		\node at (13,-8) {\huge Independent d.o.f.};
	
	% calculate omega from x and y components
	\newcommand{\omegalength}{sqrt(\xvar*\xvar+\yvar*\yvar)}
	
	% calculate the angle that omega creates w.r.t. the horizontal axis
	\newcommand{\thetavar}{atan2(\yvar,\xvar)}
	
	\newcommand{\yprimevar}{sqrt(\omegalength*\omegalength - \xprimevar*\xprimevar)}
	
	\newcommand{\thetaprimevar}{atan2(\yprimevar,\xprimevar)}

	% calculations
	\newcommand{\phivar}{90-\thetavar};							  % second angle 
	\newcommand{\phiprimevar}{90-\thetaprimevar};	 	   % second prime angle
	
	% styles
	\tikzstyle{triangle}=[ultra thick,fill=white,opacity=0.5];
	\tikzstyle{projection} = [fill=yellow,fill opacity=0.1,text opacity =1,text=black];
	
	% drawing the right basic triangle 
	\filldraw[style=triangle] (0,0) -- (\xvar,0) -- (\xvar,\yvar) -- cycle;
	
	\filldraw[style=projection] (0,0) rectangle (\xvar,-\xvar) node[pos=.5] {\large $ {dq^1}\wedge {dp_1}$};
	
	\filldraw[style=projection] (\xvar,0) rectangle (\xvar + \yvar, \yvar) node[pos=.5] {\large ${dq^2}\wedge {dp_2}$};
	
	\filldraw[style=projection,rotate around={\thetavar:(0,0)}]  (0,0) rectangle ({\omegalength}, {\omegalength}) node[pos=.5, rotate={\thetavar}] {\large $\omega$};
	
	% drawing the complex right triangle  
	\filldraw[style=triangle,rotate around={\thetavar:(\dvar,0)}] ({\dvar + \omegalength},0) -- ({\dvar},0) -- ++({-\thetaprimevar}: {\xprimevar}) -- ++({\phiprimevar}:{\yprimevar});
	
	\filldraw[style=projection,rotate around={\thetavar - \thetaprimevar:(\dvar,0)}]  (\dvar,0) rectangle (\dvar + \xprimevar,-\xprimevar) node[pos=.5, rotate=\thetavar-\thetaprimevar] {\large $d\hat{q}^1\wedge d\hat{p}_1$};
	
	\filldraw[style=projection,rotate around={\thetavar - \thetaprimevar:(\dvar,0)}]  ({\dvar + \xprimevar},0) rectangle ({\dvar + \xprimevar + \yprimevar}, {\yprimevar}) node[pos=.5, rotate={\thetavar-\thetaprimevar}] {\large $d\hat{q}^2\wedge d\hat{p}_2$};
	
	\filldraw[style=projection,rotate around={\thetavar:(\dvar,0)}]  (\dvar,0) rectangle ({\dvar + \omegalength}, {\omegalength}) node[pos=.5, rotate=\thetavar] {\large $\omega$};

	\end{tikzpicture}
	\caption{The symplectic form $\omega$ and its components along two independent degrees of freedom. As we change state variables, the new d.o.f. are also independent and therefore orthogonal to each other. Note: this is a 2D conceptualization of 4D space.}
	\label{fig:independent_dof}
\end{figure}
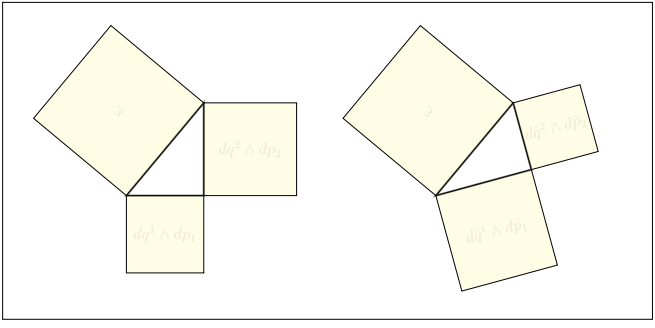

This gives us insight on the physical meaning of the geometrical structure $T^*\mathcal{Q}$. The canonical one-form $\theta=k dq$ represents the geometrical object we associate with each particle state. The symplectic form $\omega$ is the bilinear that quantifies the possibilities described by two given state variables. Note, instead, that the relationship $\omega = - \hbar d\theta$, while mathematically true, has no clear physical meaning as we have not defined what the exterior derivative on a state actually represents.

We can capture the above discussion by stating that the state space $\mathcal{S}$ is the symplectic manifold $(T^*\mathcal{Q}, \omega)$, where $\mathcal{Q}$ is the manifold that defines the unit system. This is the only manifold that allows us to define state-variable-invariant densities. The symplectic form allows us to count possibilities on an arbitrary d.o.f. 

\begin{prop}\label{prop:symplectic_manifold}
The state space $\mathcal{S}$ for the particles of a classical material is a symplectic manifold formed by a cotangent bundle $T^*\mathcal{Q}$ equipped with the canonical two-form $\omega$.
\end{prop}

\begin{justification}
	We claim the simplest (i.e. lowest dimension) state space $\mathcal{S}$ that allows state-variable-invariant densities is the cotangent bundle $T^*\mathcal{Q}$ of a single dimensional manifold $\mathcal{Q}$. Consider a set of state variables for $U \in \mathcal{S}$: the simplest unit system is one defined by a single state variable. Let $q$ be the state variable that defines the unit system. Let $\mathcal{Q}$ be the manifold charted by $q$. A diffeomorphism $\hat{q}=\hat{q}(q)$ fully defines a change of units. Let $(q,k^i)$ with $i=1...n-1$ be a set of state variables. The Jacobian determinant $|J| = d_q\hat{q} |\partial_i\hat{k}^j|=1$ as densities must be invariant. $|\partial_i \hat{k}^j| = d_{\hat{q}}q$. $\partial_i \hat{k}^j$ is one dimensional as its components must be fully specified by the previous relationship. Let $k=k^1$. $\partial_k \hat{k} = d_{\hat{q}}q$. $\hat{k}(q,0) = 0$ as a change in units does not change the $0$ value. $\hat{k} = d_{\hat{q}}q k$. $k$ changes as the component of a co-vector. Each state is identified by a point and a co-vector in $\mathcal{Q}$. The simplest state space $\mathcal{S}$ is isomorphic to $T^*\mathcal{Q}$ where $Q$ is the one dimensional manifold that defines the units.
	
	We claim the simplest state space $\mathcal{S}$ that allows state-variable-invariant densities is the symplectic manifold $(T^*\mathcal{Q}, \omega)$ where $T^*\mathcal{Q}$ is the cotangent bundle of a one dimensional differentiable manifold $\mathcal{Q}$ and $\omega$ is the canonical two-form. Express integration in coordinates. $\int_{\mathcal{S}} \rho_\mathcal{c} d\mu \propto \int_{\mathcal{S}} \rho_\mathcal{c}(q, k) dq \wedge dk$. $d\mu = \hbar dq \wedge dk = dq \wedge dp = \omega$ where $\hbar$ is a constant, $p\equiv \hbar k$  and $\omega$ is the canonical two form. $\omega$ is invariant under state variable change. $\omega$ is the symplectic form for $\mathcal{S}=T^*\mathcal{Q}$. $\mathcal{S} = (T^*\mathcal{Q}, \omega)$ is a symplectic manifold.
	
	We claim the state space $\mathcal{S}$ for the particles of a homogeneous classical material is a symplectic manifold $(T^*\mathcal{Q}, \omega)$ where $T^*\mathcal{Q}$ is the cotangent bundle of an $n$-dimensional differentiable manifold $\mathcal{Q}$ and $\omega$ is the canonical two-form. Let $q^i$ be a set of n continuous independent state variables that define the units necessary to describe a state in $U \subseteq \mathcal{S}$. Let $\mathcal{Q}$ be the manifold charted by $q^i$. Locally $\mathcal{Q} \cong \prod \mathcal{Q}^i \cong \mathbb{R}^n$. Changing units of one state variable must not change the units of the other as they are independent. State-variable-distribution can  be defined separately on each degree of freedom. For each $\mathcal{Q}^i$ we have a covector $k_i(q) dq^i \in T^*\mathcal{Q}^i$. Locally $\mathcal{S} \cong \prod T^*\mathcal{Q}^i \cong T^* \prod \mathcal{Q}^i \cong T^* \mathcal{Q}$. Integration must also be defined on an independent d.o.f. There must exist a non-degenerate two-form $\omega$ such that $\int_{U \subset \mathcal{S}} \rho_\mathcal{c} \omega$ where $U$ is any two dimensional surface in $\mathcal{S}$. $\omega$ has to be form invariant under state variable change. The canonical two-form is the only such form. $\omega = \sum dq^i \wedge dp_i$. $\mathcal{S} = (T^*\mathcal{Q}, \omega)$ is a symplectic manifold.
\end{justification}

It should be clear by now that the case of discrete topology is qualitatively different from the standard topology for $\mathbb{R}^n$. The notion that the continuous case is a limit of the discrete case, that it's ``like the discrete but with more points", leads in this case to erroneous intuition. The key question is: can we physically distinguish an isolated state? Can we have an outcome associated with only one element? The answer is no with the standard topology on $\mathbb{R}^n$. The consequence is that when we define the measure $\mu$ for the count of continuous states, we assign a finite value to a finite region, and we assign zero measure to a single state. That is what gave us integration, densities and, ultimately, conjugate variables. If we were to use the discrete topology on $\mathbb{R}^n$, if we were to assume we can identify single states as we do for countable states, then we would assign measure one to each state and infinity to a finite region. We would not have integration, just a simple sum. Our distribution would not be a density and there would be no justification for conjugate variables. A finite distribution could only distribute a finite amount of material in a finite number of states.

In other words: this is not a case where a difference at small scales leads to a small difference at large scales. The two cases are radically different. We can tell them apart. We should take the use of densities and conjugate quantities as evidence that quantities like space and time are not discrete in the sense that the processes we use to distinguish those quantities cannot identify single instances. The quantum case does not change this, as the use of densities (in the form of the wave function) and conjugate quantities is even more prominent.\footnote{If one wants to posit that space is fundamentally discrete at scales where both classical and quantum mechanics cease to be valid, the better strategy would be to work with a discrete version of $T^*\mathcal{Q}$ and not $\mathcal{Q}$ by itself. Though it is not clear how a discrete topology would become non-discrete in the limit, at least the measure may be workable.}

It should also be noted that classical particles, under this light, cannot be considered point-like. As they are the limit of infinitesimal subdivision, their spatial extent becomes infinitesimal but not zero.\footnote{This picture is also compatible with general relativity, unlike point-like particles. These would not follow geodesics as their infinite mass density would significantly affect the gravitational field.} But suppose particles were truly point-like. Then distinguishing particles would mean distinguishing points. We are back to the idea of a discrete topology on $\mathbb{R}^n$. In that case we would not have conjugate momentum, no $T^*\mathcal{Q}$, just the coordinates of the point in $\mathcal{Q}$. This would actually be more self-consistent: why wouldn't a point be enough to define the state of a point-like particle? As before, we should take the use of conjugate quantities as evidence that classical particles are really the limit of an infinitesimal subdivision.

Finally, we should note that we are in a position similar to the one discussed in \ref{prop:discrete_measure} for discrete states. We have a topological space plus a measure that allows to count states and a symplectic form that allows to count possibilities, both expressible in terms of state variables. We recovered this structure starting from the idea of an infinitesimally reducible system. That led to a state space $\mathcal{S}$ for the infinitesimal parts, on which we must be able to define state-variable-independent distributions, which in turn gave us degrees of freedom made by pairs of conjugate variables and the symplectic form typical of classical phase space. In short: being able to measure the amount of material is what allows us to count states.

\subsection{Infinitesimal reducibility}
\label{subsec:infinitesimal_reducibility}

Now that we have fully characterized what we mean by a classical material, we can stipulate the following:

\begin{assump}[Infinitesimal reducibility]\label{ass:infinitesimal_reducibility}
	The system under study is composed of an infinitesimally reducible homogeneous material and each part undergoes deterministic and reversible evolution.
\end{assump}

\begin{rationale}
	The idea is that time evolution specifies a map for the state space $\mathcal{S}$ of each infinitesimal part. Knowing how the parts evolve tells us how a composite state $\mathcal{c}$ evolves as well.
	
	Consistent with what we said in assumption \ref{ass:determinism}, if we defined a state for each particle, then a deterministic and reversible evolution on that state must exist. Yet, the idea that we can assign states to infinitesimal parts should be considered only a \emph{simplifying} assumption. The obvious reason is that we know this does not work in practice: as we keep decomposing the material we end up with molecules, atoms and then subatomic particles. But it is instructive to understand when and how the assumption breaks down at a more conceptual level.
	
	The first problem is methodological. To be able to talk about the states of a part we need a physical process that is able to distinguish between them. For a billiard ball we can imagine marking one spot with a red marker. This allows us to track it as the ball moves or collides with other balls. For an electron, instead, we do not have a way to mark a piece. In fact when two electrons scatter we can't even tell which is which, let alone what portion went where. The classical assumption may not hold because we do not have suitable physical processes at our disposal.
	
    Even if we are able to track parts, the assumption requires them to be infinitesimal, the limit of a process of infinite recursive subdivision. The best we can do experimentally is to confirm that the assumption holds up to the smallest precision available. Therefore even in the best of cases it cannot be considered an experimentally validated assumption but a reasonable default assumption (e.g. ``it worked so far").
	
	The second, more conceptual, problem is that the idealizations that allowed us to define states in assumption \ref{ass:determinism} may not hold as the parts get smaller. Recall the cannonball whose motion is sufficiently unaffected by the photons that scatter off its surface. As we consider smaller and smaller parts, at some point we will find an amount of material that is affected by the interaction with the air or photons scattering off of it. At that point the parts are no longer sufficiently isolated to define an independent state, their evolution depends on the environment. Therefore the deterministic evolution of the whole cannot be reduced to the deterministic evolution of its parts.
	
	Another issue is that defining a state requires some kind of asymmetry between system under study and environment. The future state of the system does not depend on the state of the environment, so that we can define an independent state, yet the future state of the environment is affected by the state of the system, so that we can have external processes that allow us to define physical distinguishability. If we try to assume that both the system and environment are really made of the same classical material, then the claim to that asymmetry is lost: the infinitesimal parts of both system and environment must affect one another in the same way.
	
	Another issue arises if we include the measuring device in our description. Ideally, we'd like to require the following:
	\begin{itemize}[noitemsep]
		\item The system under study is deterministic and reversible
		\item The measuring device ascertains the state of the system
		\item System and measuring device, as a whole, is deterministic and reversible
	\end{itemize}
	Unfortunately, these three requirements together are inconsistent. Say $a \in A$ is the initial state of our system and $b \in B$ the initial state of our measuring device: the final state $\hat{a}=\hat{a}(a)$ because the system is deterministic; $a=f(\hat{b})$ because we are able to know the initial state of the system by looking at the final state of the measuring device; $b=b(\hat{a},\hat{b})$ because the system as a whole is reversible. Since $\hat{a}$ is determined by $a$, and $a$ can be known from $\hat{b}$, $b = b(\hat{a}(f(\hat{b})), \hat{b}) = b(\hat{b})$ can be determined from just $\hat{b}$. But since $a=f(\hat{b})$, the whole past state could be reconstructed just by looking at the future state of the measuring device. But $(\hat{b}) \rightarrow (a, b)$ cannot be an injection, therefore we arrive at a contradiction.
	
	Of those three requirements we can only pick two. If we want determinism and reversibility for the combined system, either both systems are not independently deterministic or the second is not a measuring device. In other words: we can't expect to have deterministic and reversible evolution at all levels of aggregation and have the parts interact in any physically meaningful way.
	
	The same problem of physically meaningless interaction surfaces when we consider the states of infinitesimal parts. Under our classical assumption, if we were to take the state of an infinitesimal part to really be its full description, with no unstated part, then each piece would evolve independently, oblivious to the other parts. Each particle would essentially reside in its own separate physical universe, as it cannot physically distinguish anything else. This is not physically meaningful.
	
	If we assumed the system is deterministic and reversible only as a whole, then each part could evolve depending on the states of the other parts. This would seem much better, as the state of the pieces is still exhaustive yet they are physically connected to each other. But this does not actually solve the problem of independent oblivious components. First, for this to work, one would have to specify how the states of the parts were defined since their evolution is individually no longer deterministic (the future of each part depends on the state of other parts). But ignoring that issue, the bigger problem is that we can always locally separate the evolution into independent degrees of freedom. For example, the position and momentum of two particles may affect each other during the evolution, but the average and difference in position and momentum may evolve independently.\footnote{Also note that any Hamiltonian is locally isomorphic to a free particle~\cite{Linares}} We can always find such a local decomposition, and the easiest way to see that is in terms of the transported variables: they retain the original value; they clearly evolve independently. And since the pieces are infinitesimal, a local decomposition is all that is needed. While such description may be cumbersome to achieve in practice, conceptually it is still possible. So, even if each particle evolution depends on the state of the others, the system is described by degrees of freedom that evolve independently, oblivious to each other.	
	
	The moral is that the classical idea of being able to assign to all systems a state which represents their full description does not work. The state is only what we can describe and it can't be the full description. As the division between system and environment is subjective, each system must be able to function as both. Therefore it will have a part whose evolution depends only on the system itself, the state, and a part whose evolution depends on other systems, the unstated part. While they may not be the same in all circumstances, they must exist. The state is what gives the identity to the system, the part we can study and describe. The unstated part is what allows continuous interaction between the system and the environment; it's what allows us to study and describe the state. As we'll see later, it is precisely this problem that quantum mechanics conceptually addresses better than classical mechanics.
	
	One final problem is with time itself: an infinitesimally reducible system undergoing deterministic and reversible evolution cannot tell time. If it did, we would be able to find a quantity that changes through time but is invariant under state variable changes. The problem is that deterministic and reversible time evolution is equivalent to a state variable change. In fact, at a fixed time, we can choose to describe the system with the transported state variables of any past or future times. So, all quantities that are invariant under state variable changes are also invariant under time evolution. The somewhat ironic result is that while time is essential for the very definition of deterministic and reversible evolution, time itself has to be defined by some other type of process.\footnote{For example, if the phase-space volume occupied by a distribution increased, it would give us a coordinate-invariant quantity that changes in time. This, however, is not possible under Hamiltonian evolution.}
	
	In light of what we discussed, we cannot take the classical assumption to be fundamental, in the sense that we cannot take it to strictly apply to all of the universe. While ultimately flawed, the  assumption can be considered valid for a great number of macroscopic systems, and is also useful as a default assumption of sorts. Understanding its shortcomings will help us later to see how the quantum assumption solves, at least partially, some of these issues.
	
	As a final note: we caution against automatically thinking that the classical assumption is valid for all macroscopic systems. It is conceptually possible to have a macroscopic system where a clear independent state cannot be assigned to each part, in which case the assumption would not hold.\footnote{For example, this is the case in Bose-Einstein condensates~\cite{BEC}.}
\end{rationale}

\subsection{Hamiltonian mechanics}

We are finally ready to write the equations of motion. As per \ref{prop:homeomorphism} our evolution is at least a self-homeomorphism $f:\mathcal{S} \leftrightarrow \mathcal{S}$. But this is not enough.

The evolution must map the distribution point-wise so that the value associated at the initial state is the same as the value at the final state. All the material that starts in $\mathcal{s}$ has to end up in $\hat{\mathcal{s}}$. In math terms $\rho_{\hat{\mathcal{c}}} (\hat{\mathcal{s}}) = \rho_\mathcal{c}(\mathcal{s})$.

In the same way we expect the total amount of material to be conserved. If $U \in \mathcal{S}$ is a set of initial particle states and $\Lambda_U(\mathcal{c})$ is the amount of material associated with that set, then we expect it to be equal to the amount of material $\Lambda_{\hat{U}}(\hat{\mathcal{c}})$ associated with the set of final states $\hat{U} \in \mathcal{S}$.

But probably the easiest way to look at it is that initial and final sets of states have to possess the same count of states and possibilities as defined by the symplectic form $\omega$. Therefore the area within a degree of freedom, which represents the number of possibilities in said d.o.f., needs to be mapped to an equal area within the transported d.o.f. (i.e. the d.o.f. defined by the transported state variables). Independent d.o.f. must remain independent, and therefore transported orthogonal d.o.f. remain orthogonal. This means that the product of possibilities of independent d.o.f. is also conserved. These statements are the physical justification of Gromov's non-squeezing theorem~\cite{Gromov,deGosson,Stewart} and Liouville's theorem. And they give intuitive insight on the geometry of Hamiltonian systems.

Mathematically, under the classical Hamiltonian assumption,  deterministic and reversible evolution is a self-symplectomorphism (or self-isometry or canonical transformation depending on your math training). That is, it does not just preserve the topology but also the symplectic form $\omega$.

\begin{prop}\label{prop:symplectomorphism}
	A deterministic and reversible evolution map for the particles of a classical material is a self-symplectomorphism. That is: $\mathcal{T}_{\Delta t}: T^*\mathcal{Q} \rightarrow T^*\mathcal{Q}$ and $\mathcal{T}_{\Delta t}^*\omega = \omega$ where $\mathcal{T}_{\Delta t}^*$ is the pullback of $\mathcal{T}_{\Delta t}$.
\end{prop}

\begin{justification}
	We claim $\mathcal{T}_{\Delta t}$ is a self-homeomorphism on $T^*\mathcal{Q}$. The state space for the particles of a classical material is $T^*\mathcal{Q}$ by \ref{prop:symplectic_manifold}. $\mathcal{T}_{\Delta t}$ is a deterministic and reversible evolution map and by \ref{prop:homeomorphism} is a self-homeomorphism.
	
	We claim $\mathcal{T}_{\Delta t}$ is a symplectomorphism on $(T^*\mathcal{Q}, \omega)$. The distribution on final states must be defined  therefore the Jacobian for $\mathcal{T}_{\Delta t}$ exists and is non-zero: $\mathcal{T}_{\Delta t}$ is a diffeomorphism. A deterministic and reversible process conserves the number of states and possibilities. $\omega$ is the two-form that returns the count of possibilities. $\omega$ is invariant under deterministic and reversible evolution. $\mathcal{T}_{\Delta t}$ is a symplectomorphism by definition.
\end{justification}

If we assume continuous time evolution we have the following:

\begin{prop}\label{prop:hamiltons_equations}
	A continuous deterministic and reversible evolution for the particles of a classical material admits a Hamiltonian $H \in C^2(T^*\mathcal{Q}, \mathbb{R})$ that allows us to write the laws of evolution as
\begin{align*}
d_{t}q^i &= \partial_{p_i} H \\
d_{t}p_i &= - \partial_{q^i} H
\end{align*}
\end{prop}

\begin{justification}
We claim the vector field $S \in T\mathcal{S}$ for the infinitesimal displacement $S^a = d\xi^a/dt$ admits a potential $H \in C^2(T^*\mathcal{Q}, \mathbb{R})$ such that $S^{a} \omega_{ab} = \partial_{b}H$. The state space $\mathcal{S}$ for the particles of a classical material is a symplectic manifold by \ref{prop:symplectic_manifold}. The map for infinitesimal time evolution $\mathcal{T}_{dt}$ is an infinitesimal self-symplectomorphism by \ref{prop:symplectomorphism}. By \ref{symplectomorphism_generator} the infinitesimal displacement $S$ admits a potential $H \in C^2(T^*\mathcal{Q}, \mathbb{R})$ such that $S^{a} \omega_{ab} = \partial_{b}H$

We claim the state variables evolve according to Hamilton's equations. $S^{a} \omega_{ab} = d_t\xi^a \omega_{ab} = \partial_{b}H$. For $b=\{1,...,n\}$ we have $d_tp_i (-1) = \partial_{q_i} H$. For $b=\{n+1,...,2n\}$ we have $d_tq^i (+1) = \partial_{p_i} H$.
\end{justification}

We recognize the familiar set of Hamilton's equations. They are the set of equations that describe the deterministic and reversible motion of the infinitesimal parts of an infinitesimally reducible homogeneous material. Note that the argument goes the other way as well. A system governed by Hamilton's equations is deterministic and reversible: the equation of motion given by $H$ are differentiable since $H$ is twice differentiable. Therefore the equations are Lipschitz continuous and, by the Picard-Lindel\"{o}f theorem, a unique solution exists~\cite{Grant}.

In other words: the forces that conserve energy (i.e. the value of a suitable Hamiltonian) are exactly the ones that provide deterministic and reversible motion. The challenge in their derivation mainly lies in the necessary use of different branches of mathematics (e.g. topology, measure theory, differential geometry, symplectic geometry). The conceptual meaning, on the other hand, is hopefully straightforward.

It's important to realize that, during the derivation, multiple mathematical features (e.g. invariant densities, cotangent bundle for phase space, symplectomorphism) were justified by the same physical assumption. The math itself gives us no indication that the different features stem from the same source; therefore, the math itself does not give us a conceptually unified picture. This is one of the reasons we believe that starting from the physical description is objectively better if we are to come to a better understanding of our physical theories.

Note that we could have taken different approaches. For example, we could have appealed to information theory and required our invariant distributions to preserve Shannon's information entropy~\cite{Shannon}, as no information should be gained or lost during a deterministic and reversible process. Or we could have appealed to statistical mechanics and required that the determinant of the covariance matrix be conserved, as a deterministic and reversible process should be defined at the same level of precision. With suitable treatment of independent d.o.f. both these approaches recover Hamiltonian mechanics as well. While the full treatment is outside the scope of this work, we want to underline that the definitions used in this work go a long way to building bridges at the core of different mathematical and scientific branches.

\section{Time dependent evolution}
\label{sec:relativity}

We now generalize our discussion to include time dependent dynamics. This is needed when the evolution map is not the same at all instants or when state variables depend on time (e.g. changing to a moving frame).

To do this we will redefine the particle state space to include the temporal degree of freedom. We'll find that the dynamics is described by relativistic Hamiltonian mechanics in the extended phase space. We'll also find that particle states divide into standard and anti-states depending whether time is aligned or anti-aligned with the evolution parameter. Note that no extra assumption is needed, which makes relativistic mechanics simply the correct time dependent description of the deterministic and reversible evolution of an infinitesimally reducible system.

\subsection{Time changes}

The discussion so far has been limited to the case where both the state variables and the evolution map never change throughout the evolution. This is too restrictive as there are very reasonable situations in which this does not hold.

The first obvious case is that the map may be time dependent. This does not affect our definition of determinism and reversibility as we still can tell final state from initial state and vice-versa. But at this point we have no way to specify a time dependent evolution: the Hamiltonian $H: T^*\mathcal{Q} \rightarrow \mathbb{R}$ as we derived it is just a function of the state.

The second case is when we perform a transformation for the time parameter $\hat{t}=\hat{t}(t)$. The equations of motion transform and we have $d_{t}\hat{t} \, d_{\hat{t}}q^i = \partial_{p_i} H$. $q^i: T^*\mathcal{Q} \rightarrow \mathbb{R}$ and $H: T^*\mathcal{Q} \rightarrow \mathbb{R}$ are real functions of the state space and cannot be redefined to include $d_{t}\hat{t}$. Physically the transformation is introducing fictitious forces that are not conservative, so they cannot be expressed by a Hamiltonian. But this means that we have an ill defined mathematical framework as the notion of determinism and reversibility is not defined in a way that is independent of time transformations. We need a framework that is capable of handling the fictitious forces as well.

The third case is when we perform a state variable transformation $\hat{q}=\hat{q}(q,t)$ that is time dependent (e.g. $\hat{q}=q+vt$). Our composite state distribution $\rho(q)$ becomes $\rho(\hat{q},t)$: it is no longer defined at equal time. This means that our measure $\mu$ needs to be modified to define integrals over the time variable.

As we can see, the framework we have is ill suited to handle these cases. Therefore we cannot simply stick time in the Hamiltonian and expect everything to work out. But we can't ignore the problem either as the three cases outlined are common situations to study. So we need to go back to our definitions and amend them properly.

\subsection{Complete state space}

The first thing to do is to amend our definition of state space to include all states at all times. It may seem like we are extending our definition of state and state space but, at a closer look, we are not. Since we defined state as a physical configuration at a particular time, the set of all states (i.e. the state space) is more accurately defined as the set of all configurations at all times. In the previous sections we tacitly simplified the problem by ignoring time, which made it easier to study the time independent case. This approach is common in physics and engineering, which is why it fits the standard names, but this is also the source of the above problems.

We'll therefore call \emph{complete state space} of particle states the set of all configurations at all times and we'll indicate it by $\bar{\mathcal{S}}$. In that space, a state $\mathcal{s} \in \bar{\mathcal{S}}$ will be given by $(q^i, p_i, t)$, making $\bar{\mathcal{S}}$ a $2n+1$ dimensional manifold. A region at equal time will be a hypersurface $\mathcal{S}_{t=t_0} \subset \bar{\mathcal{S}}$ that includes all possible configurations at a particular time. Such hypersurface will need to cut across all evolutions once and only once, therefore not all functions $t : \bar{\mathcal{S}} \rightarrow \mathbb{R}$ are suitable time variables. A time evolution is a line  $\bar{\mathcal{s}} \subset \bar{\mathcal{S}}$ that passes once through each equal time surface. Specifying an evolution map means giving a set of lines such that each state is traversed once and only once.

\begin{prop}\label{prop:complete_particle_state_space}
	The complete state space $\bar{\mathcal{S}}$ for the particles of a classical material is a $2n+1$ dimensional differentiable manifold. The state space $\mathcal{S}_{t=t_0}$ at a particular time is a hypersurface of $\bar{\mathcal{S}}$.
\end{prop}

\begin{justification}
	We claim $\bar{\mathcal{S}}$ is a $2n+1$ dimensional differentiable manifold. Locally there exists $U\subseteq \bar{\mathcal{S}}$ such that $U \cong T^*\mathcal{Q}\times \mathbb{R}$ as a state $\mathcal{s} \in U$ exists for each configuration and each time instant. Locally $\bar{\mathcal{S}} \supseteq U \cong\mathbb{R}^{2n+1}$. $\bar{\mathcal{S}}$ is a $2n+1$ dimensional manifold. $T^*\mathcal{Q}$ is a differentiable manifold. Time evolution is a diffeomorphism by \ref{prop:symplectomorphism} therefore derivatives between state variables and time variables always exist. $\bar{\mathcal{S}}$ is a differentiable manifold.
	
	We claim $\mathcal{S}_{t=t_0}$ is a hypersurface of $\bar{\mathcal{S}}$. Let $t : \bar{\mathcal{S}} \rightarrow \mathbb{R}$ be a time variable. $\mathcal{S}_{t=t_0} = \{ \mathcal{s} \in \bar{\mathcal{S}} \; | \; t(\mathcal{s}) = t_0 \}$ is a level set. $t$ is monotonic as it is a time variable. $t$ has no critical points. $t_0$ is a regular value. $\mathcal{S}_{t=t_0}$ is an embedded submanifold of co-dimension $1$. $\mathcal{S}_{t=t_0}$ is a hypersurface.
\end{justification}

For the composite state we need to take a different approach. A distribution of material over states at equal time $\rho_\mathcal{c} : \mathcal{S}_{t=t_0} \rightarrow \mathbb{R}$ will depend on the choice of time variable, therefore the complete state space will not be useful as it will not be time invariant. What we are really interested in is the \emph{state evolution space} of a classical material, which we'll indicate by $\bar{\mathcal{C}}$. To each state evolution $\bar{\mathcal{c}} \in \bar{\mathcal{C}}$ will correspond a distribution $\rho_{\bar{\mathcal{c}}} : \bar{\mathcal{S}} \rightarrow \mathbb{R}$ over all particle states at all times. This distribution will be differentiable: it's differentiable over all equal time hypersurfaces and time evolution is differentiable. The distribution, though, is not integrable over the complete state space $\bar{\mathcal{S}}$. Such integration will operate over time as well, giving the amount of material multiplied by time. If time is taken to be infinite the value is not finite. The distribution, instead, must be integrable over all equal time hypersurfaces: no matter what time variable we choose, the amount of material found at a particular time has to be finite.

\begin{prop}\label{prop:state_evolution_space}
	The state evolution space $\bar{\mathcal{C}}$ for a classical material is isomorphic to the space of differentiable functions that are Lebesgue integrable on any equal time hypersurface. That is $\bar{\mathcal{C}} \cong  C^1(\bar{\mathcal{S}})\cap L^1_t(\bar{\mathcal{S}}, \mu)$ where $L^1_t(\bar{\mathcal{S}}, \mu)=\{{\rho: \bar{\mathcal{S}} \rightarrow \mathbb{R}} \; | \; {\int_{\mathcal{S}_{t=t_0}} |\rho| d\mu < \infty} \; \forall \, t, t_0 \}$.
\end{prop}

\begin{justification}
	We claim $\bar{\mathcal{C}} \cong  G \subseteq C^1(\bar{\mathcal{S}})\cap L^1_t(\bar{\mathcal{S}}, \mu)$. Let $\rho_{\bar{\mathcal{c}}} : \bar{\mathcal{S}} \rightarrow \mathbb{R}$ be the distribution of material associated with an evolution $\bar{\mathcal{c}} \in \bar{\mathcal{C}}$. Let $t : \bar{\mathcal{S}} \rightarrow \mathbb{R}$ be a time variable. 
	The restriction of $\rho_{\bar{\mathcal{c}}}$ over an arbitrary level set $\mathcal{S}_{t=t_0}$ is continuous and integrable by \ref{prop:differentiable_manifold}. $\rho_{\bar{\mathcal{c}}} \in L^1_t(\bar{\mathcal{S}}, \mu)$. Time evolution is a diffeomorphism by \ref{prop:hamiltons_equations}. $\rho_{\bar{\mathcal{c}}}$ is differentiable along the time variable as well as the state variables. $\rho_{\bar{\mathcal{c}}} \in C^1(\bar{\mathcal{S}})$.
	
	We claim $\bar{\mathcal{C}} \cong  C^1(\bar{\mathcal{S}})\cap L^1_t(\bar{\mathcal{S}}, \mu)$. As in \ref{prop:differentiable_manifold}, if two distributions are different there exists an outcome that can tell them apart, therefore they represent two distinct physical evolutions.
\end{justification}

Before proceeding, it's useful to get a better intuition for the geometry of these spaces. Note that the transported state variables are perfect to label the particle state evolution $\bar{\mathcal{s}} \subset \bar{\mathcal{S}}$: as they do not change during the evolution we can take them to label not just the values at a particular time, but the evolutions themselves. The distribution is transported deterministically and reversibly over those lines, so the value of the distribution will be constant along each particle state evolution $\bar{\mathcal{s}}$. To identify the full evolution, then, we just need to specify the distribution only on a region that cuts across all particle evolutions, i.e. a hypersurface at constant time.

\subsection{Relativistic cotangent bundle}

We now want to extend the symplectic form so we can count states and possibilities over the complete state space $\bar{\mathcal{S}}$. The first thing to do is to complete our $n$-dimensional $\mathcal{Q}$, the manifold that defines our units, to include the time variable $t$. This, again, highlights our previous oversight: we didn't include time within the quantities needed to specify a physical configuration at a particular time. We call $\mathcal{M}$ such a space which is a manifold of dimension $n+1$: state variables and the time variable are independent (i.e. any combination of state and time variables is valid).

As we said before, the time variable is not simply another state variable: it doesn't identify extra configurations and it's the variable on which we have defined deterministic and reversible evolution. When changing variables, we still need to make sure that deterministic and reversible motion can be defined on the new time variable. Transformations like $\hat{t}=q, \; \hat{q}=t$ (exchanging time with another state variable) or $\hat{t}=\sqrt{t^2 + q^2}, \; \hat{q}=\tan^{-1}(t/q)$ (polar coordinates between time and state variable) clearly do not guarantee deterministic and reversible motion over the new time variable. To be meaningful, the change of time variable must preserve time ordering therefore the transformation has to be strictly monotonic between the two time variables. On the other hand, the change of state variables must preserve the ability to uniquely identify states at each instant in time. In this sense, time is not just like another state variable.\footnote{In relativity, for example, it is common to use foliations and the $3+1$ formalism to recover the special character of time when performing calculations or to gain better physical insight~\cite{Gourgoulhon}.}

As in the time independent case, we are interested in invariant densities which are defined on co-vectors. The co-vector $k_i dq^i$ is completed with a time component $\bar{\omega} dt$.\footnote{$\bar{\omega}$ is the classical analogue of the wave frequency. We use $\bar{\omega}$ to distinguish from the symplectic form $\omega$.} As each state variable $q^i$ has a conjugate quantity $p_i\equiv \hbar k$, the time variable $t$ will have a conjugate quantity $E\equiv\hbar\bar{\omega}$. We call the combination of $(t, E)$ the temporal degree of freedom. As time is special, so is the temporal degree of freedom, which is treated differently when defining a distribution $\rho_\mathcal{c}$ and therefore by the measure $\mu$ and the symplectic form $\omega$.

Most of all, the temporal degree of freedom is not an independent degree of freedom because $E$ cannot be an independent variable. It cannot add any configurations, as those are defined by $T^*\mathcal{Q}$ only, and cannot add any time instants, as those are defined by $t$ alone. Therefore there must exist a constraint such that, locally, $E=E(t,q^i,p_i)$. The complete state space $\bar{\mathcal{S}}$ is a hypersurface of $T^*\mathcal{M}$.

As the temporal degree of freedom $(t, E)$ is not independent of the other d.o.f.~$(q^i, p_i)$; it is not orthogonal to them in $T^*\mathcal{M}$. States are defined on the plane where $( q, p )$ (maximally) change. This is not the plane of constant $( t, E )$ (they are not orthogonal) where the area given by $dq \wedge dp$ is defined. It is the plane perpendicular to constant $( q, p )$. And the plane of constant $( q, p )$ is where the area given by $dt \wedge dE$ is defined. That is: the plane where we can properly count states and define our symplectic form $\omega$ is perpendicular to $dt \wedge dE$. Intuitively, states are defined independently of time, therefore they are defined on a surface perpendicular to the temporal d.o.f.

We have a right triangle-like relationship between the plane where $\omega$ is defined and its projections on the planes defined by each d.o.f., similar to the multiple d.o.f.:
\begin{align*}
\textrm{multiple d.o.f.} \;\;\; &dq^1 \wedge dp_1 + dq^2 \wedge dp_2 = \omega \\
\textrm{temporal d.o.f.} \;\;\; &dt \wedge dE + \omega = dq \wedge dp
\end{align*}
But in the previous case, the right angle was between the two independent d.o.f. In this case, the right angle is between $\omega$ and the plane of constant $(q, p)$ where $dt \wedge dE$ is defined. We rewrite it as $dq \wedge dp - dt \wedge dE = \omega$. This corresponds to the Minkowski product across d.o.f. and the vector product within. In terms of unified state variables $\xi^a\equiv \{t, q^i, E, p_i\}$ we have:

\begin{align*}
\omega_{ab} =  \left[
\begin{array}{cc}
0 & 1 \\
-1 & 0 \\
\end{array}
\right] \otimes \left[
\begin{array}{cc}
-1 & 0 \\
0 & I_n \\
\end{array}
\right] =
\left[
\begin{array}{cccc}
0 & 0 & -1 & 0 \\
0 & 0 & 0 & I_n \\
1 & 0 & 0 & 0 \\
0 & -I_n & 0& 0 \\
\end{array}
\right] \\
\end{align*}

\begin{figure}
		\begin{tikzpicture}[scale=0.4]

		% INPUTS:
		\newcommand{\avar}{5};		% omega horizontal length
		\newcommand{\bvar}{6};		% height of first dt dE 
		\newcommand{\dvar}{4};			 % height of second dt dE
		\newcommand{\xvar}{20};			% space between the two

		% styles
		\tikzstyle{triangle}=[ultra thick,fill=white,opacity=0.5];
		\tikzstyle{projection} = [fill=yellow,fill opacity=0.1,text opacity =1,text=black];
	
		% outer rectangle and name of diagram
		\draw (-8,-7) rectangle (34,12.5);
		
		% the different points that will be connected to draw the shapes
		\coordinate (p2) at (\avar,0);
		\coordinate (p3) at (\avar,\bvar);
		\coordinate (p5) at (\avar + \bvar,\bvar);	
	    \coordinate (p7) at (\avar,- \avar);
	    \coordinate (p8) at (\avar - \bvar,\avar + \bvar);

		% angle made by the coordinates to rotate the first dq dp square
		\newcommand{\thetavar}{atan2(\bvar , \avar)};
		
		% drawing the first shape
		\filldraw[style=triangle] (0,0) -- ($(p2)$) -- ($(p3)$) -- cycle;
		
		\filldraw[style=projection,rotate around={0:(0,0)}] (0,0) rectangle ($(p7)$) node[pos=.5] {\large $\omega$};
				
		\draw[style=projection,rotate around={\thetavar:(0,0)}] (0,0) rectangle ($(p8)$) node[pos=.5, rotate=\thetavar] {\large $dq\wedge dp$};
			
		\draw[style=projection,rotate around={0:(0,0)}] ($(p2)$) rectangle ($(p5)$) node[pos=.5] {\large $dt\wedge dE$};	
		
		% coordinates to the second shape (to the right)
		\coordinate (h0) at (\xvar,0);
		\coordinate (h2) at (\xvar + \avar,0); 
		\coordinate (h3) at (\xvar + \avar, \dvar);
		\coordinate (h5) at (\xvar + \avar + \dvar, \dvar);
		\coordinate (h7) at (\xvar + \avar,- \avar);
		\coordinate (h8) at (\xvar + \avar - \dvar,\avar + \dvar);

		% angle made by the coordinates to rotate the first dq dp square
		\newcommand{\phivar}{atan2(\dvar , \avar)};

		%drawing the shape on the right
		\filldraw[style=triangle] ($(h0)$) -- ($(h2)$) -- ($(h3)$) -- cycle;

		\filldraw[style=projection,rotate around={0:($(h0)$)}] ($(h0)$) rectangle ($(h7)$) node[pos=.5] {\large $\omega$};

		\draw[style=projection,rotate around={\phivar:($(h0)$)}] ($(h0)$) rectangle ($(h8)$) node[pos=.5, rotate=\phivar] {\large $d\hat{q}\wedge d\hat{p}$};

		\draw[style=projection,rotate around={0:($(h0)$)}] ($(h2)$) rectangle ($(h5)$) node[pos=.5] {\large $d\hat{t}\wedge d\hat{E}$};

	\end{tikzpicture}
	\caption{The symplectic form $\omega$ and its components along temporal and standard degrees of freedom. The two d.o.f. are not independent and therefore not orthogonal. The temporal d.o.f. is orthogonal to the symplectic form as states are defined (and counted) independently of time. Compare with FIG.~\ref{fig:independent_dof}. Note: this is a 2D conceptualization of 4D space.}
	\label{fig:temporal_dof}
\end{figure}
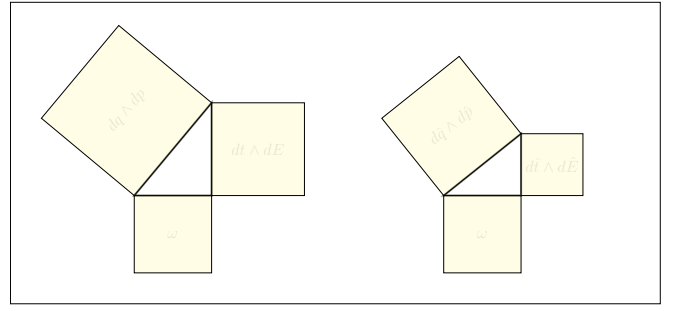

The symplectic form $\omega$ is still the canonical two-form expressed with the mathematically non-canonical, but physically meaningful, variable $E$. This allows us to better understand its physical meaning. $\omega$ allows us to count the possibilities within a degree of freedom, adjusting it on $T^*\mathcal{M}$ to avoid the counting problems introduced by time variable changes.

The symplectic form $\omega$ also allows us to distinguish temporal from standard d.o.f. Given a two-dimensional surface $U \subset T^*\mathcal{M}$, how can we tell if it should be charted by a $(t,E)$ or whether we can define a distribution on it? Consider $\int_U \omega$. The result will be positive if the integration is over a standard degree of freedom, and it will be negative if the integration is over the temporal degree of freedom. If the contribution in every subregion of $U$ is positive, then the whole $U$ is always oriented along a standard degree of freedom.

The invariance of $\omega$ also prevents mixing between standard and temporal d.o.f. The Euclidean signature across standard d.o.f. allows us to move ranges of possibilities from one d.o.f. to the other while conserving the total: the generators of rotation are in fact divergence free. The Minkowski signature, instead, only allows hyperbolic rotation (i.e. Lorentz boosts) across standard and temporal d.o.f., whose generators are curl free. Therefore, as required, we cannot transform a distribution over a standard d.o.f. into a distribution over a temporal one.

We can capture the above discussion by stating that the complete state space $\bar{\mathcal{S}}$ is a hypersurface of the symplectic manifold $(T^*\mathcal{M}, \omega)$, where $\omega \equiv \sum dq^i \wedge dp_i - dt \wedge dE$. And if we set $q^\alpha = (t, q^i)$ and $p_\alpha = (-E, p_i)$, we can write $\omega \equiv \sum dq^\alpha \wedge dp_\alpha$. This is the manifold that allows us to define time-and-state-variable-invariant densities. The symplectic form $\omega$ allows us to measure the count of possibilities as in the time independent case.

\begin{prop}\label{prop:relativistic_symplectic_manifold}
	The complete state space $\bar{\mathcal{S}}$ for the particles of a classical material is a hypersurface of the symplectic manifold formed by a cotangent bundle $T^*\mathcal{M}$ equipped with the canonical two-form $\omega = \sum dq^i \wedge dp_i - dt \wedge dE = \sum dq^\alpha \wedge dp_\alpha$.
\end{prop}

\begin{justification}
	We claim the space of physical objects that allows time-and-state-variable-invariant densities is a symplectic manifold $(T^*\mathcal{M}, \omega)$ where $\omega = \sum dq^i \wedge dp_i - dt \wedge dE$. Let $\mathcal{M}$ be the manifold that defines the unit system including time. Locally $\mathcal{M} \cong \mathcal{Q} \times \mathbb{R}$ as time and state variables are independent quantities. As in \ref{prop:symplectic_manifold}, invariant densities are defined on the symplectic manifold $(T^*\mathcal{M}, \omega)$. Let $\bar{\omega}$ be the co-vector component associated with $t$. Let $E\equiv\hbar \bar{\omega}$. Let $\omega$ be the symplectic form. The components of $\omega$ across different d.o.f. must be zero as they are not invariant under change of units. $\omega=\sum dq^i \wedge dp_i \pm dt \wedge dE$. The $+$ case is excluded as the temporal degree is not an independent d.o.f.: $\omega = \sum dq^i \wedge dp_i - dt \wedge dE$.
	
	We claim $\bar{\mathcal{S}}$ is a hypersurface of $(T^*\mathcal{M}, \omega)$. $\bar{\mathcal{S}} \subseteq (T^*\mathcal{M}, \omega)$ as states are physical objects that allow time-and-state-variable-invariant densities. Consider the map $f : T^*\mathcal{M} \rightarrow \bar{\mathcal{S}} \; | \; (q^i, p_i, t, E) \mapsto (q^i, p_i, t)$. $f$ is an open map. The topology of $\bar{\mathcal{S}}$ is the subspace topology. Let $\rho : T^*\mathcal{M} \rightarrow \mathbb{R}$ be an invariant density. $\rho$ is differentiable by \ref{prop:differentiable_manifold}. Its restriction on $\bar{\mathcal{S}}$ is also an invariant distribution and therefore differentiable. The inclusion map $\bar{\mathcal{S}} \hookrightarrow T^*\mathcal{M}$ is a smooth embedding. $\bar{\mathcal{S}}$ is an embedded submanifold of co-dimension $1$. $\bar{\mathcal{S}}$ is a hypersurface.
\end{justification}

\subsection{Relativistic Hamiltonian mechanics}

As we are using time as a variable to label states, we will use a different quantity as the parameter for the evolution. The trajectory of a particle in the extended phase space will be given by the evolved variables $\xi^a(s)$ in terms of a parameter $s$. As before, deterministic and reversible evolution will preserve $\omega$, as the possibility count on each independent d.o.f. will be conserved. Mathematically, deterministic and reversible evolution is a self-symplectomorphism.

\begin{prop}\label{prop:relativistic_symplectomorphism}
	A time-dependent deterministic and reversible evolution map for the particles of a classical material is a self-symplectomorphism on $T^*\mathcal{M}$. That is: $\mathcal{T}_{\Delta s}: T^*\mathcal{M} \rightarrow T^*\mathcal{M}$ such that $\mathcal{T}_{\Delta s}^*\omega = \omega$ where $\mathcal{T}_{\Delta s}^*$ is the pullback of $\mathcal{T}_{\Delta s}$.
\end{prop}

\begin{justification}
	We claim $\mathcal{T}_{\Delta s}$ is a self-homeomorphism on $T^*\mathcal{M}$. The complete state space for particles of a classical material is $\bar{\mathcal{S}} \subset T^*\mathcal{M}$ by \ref{prop:relativistic_symplectic_manifold}. $\mathcal{T}_{\Delta s}$ is a deterministic and reversible evolution map and by \ref{prop:homeomorphism} is a self-homeomorphism. We can extend the map over $T^*\mathcal{M}$.
	
	We claim $\mathcal{T}_{\Delta s}$ is a symplectomorphism on $(T^*\mathcal{M}, \omega)$. The distribution on final states must be defined. The Jacobian for $\mathcal{T}_{\Delta s}$ exists and is non-zero. $\mathcal{T}_{\Delta s}$ is a diffeomorphism. A deterministic and reversible process conserves the number of states and possibilities. $\omega$ is the two-form that returns the count of possibilities. $\omega$ is invariant under deterministic and reversible evolution. $\mathcal{T}_{\Delta s}$ is a symplectomorphism by definition.
\end{justification}

If we assume continuous time evolution we have the following:

\begin{prop}\label{prop:relativistic_hamiltons_equations}
	A continuous time-dependent deterministic and reversible process for the particles of a classical material admits an invariant Hamiltonian $\mathcal{H}: T^*\mathcal{M} \rightarrow \mathbb{R}$ that allows us to write the laws of evolution as
	\begin{align*}
	d_{s}t &= - \partial_{E} \mathcal{H} \\
	d_{s}E &= \partial_{t} \mathcal{H} \\
	d_{s}q^i &= \partial_{p_i} \mathcal{H} \\
	d_{s}p_i &= - \partial_{q^i} \mathcal{H}
	\end{align*}
\end{prop}

\begin{justification}
We claim the vector field $S \in T(T^*\mathcal{M})$ for the infinitesimal displacement $S^a = d_s\xi^a$ admits a potential $\mathcal{H}$ such that $S^{a} \omega_{ab} = \partial_{b}\mathcal{H}$. The map for infinitesimal evolution $\mathcal{T}_{ds}$ is an infinitesimal self-symplectomorphism by \ref{prop:relativistic_symplectomorphism}. By \ref{symplectomorphism_generator} the infinitesimal displacement $S$ admits a potential $\mathcal{H}$ such that $S^{a} \omega_{ab} = \partial_{b}\mathcal{H}$

We claim the state variables evolve according to the extended Hamilton equations. $S^{a} \omega_{ab} = d_s\xi^a \omega_{ab} = \partial_{b}\mathcal{H}$. For $a = 0$ we have $d_s t \, (-1) = \partial_{E} \mathcal{H}$. For  $a=\{1,...,n\}$ we have $d_s q^i \, (+1) = \partial_{p^i} \mathcal{H}$. For $a=n+1$ we have $d_s E \, (+1) = \partial_{t} \mathcal{H}$. For $a=\{n+2,...,2n + 1\}$ we have $d_s p_i \, (-1) = \partial_{q^i} \mathcal{H}$.
\end{justification}

We recognize the equations as those for Hamiltonian mechanics on the extended phase space~\cite{Synge,Lanczos,Struckmeier}.

The trajectory in time $t(s)$ is of particular importance. Since the evolution is deterministic and reversible, for each value of $s$ we need to have one and only one value of $t$. Therefore $t(s)$ is invertible, strictly monotonic and $d_{s}t$ along a trajectory cannot change sign. This means there are two classes of states: the ones where the parametrization $s$ is aligned with time $t$, which we call standard states, and the ones where the parametrization $s$ is anti-aligned with time $t$, which we call anti-states. Let us call $\lambda : T^*\mathcal{M} \rightarrow \mathbb{R}$ the function $\lambda (\xi^a) \mapsto d_s t |_{\xi^a}$ that returns the change of $t$ along $s$. $\lambda > 0$ for all standard states while $\lambda < 0$ for all anti-states.  Note that since the parametrization is conventional and can be changed to $s'=-s$, what we call standard and anti-states is also conventional. What is physical and not conventional, though, is that standard and anti-states cannot be connected by deterministic and reversible evolution.\footnote{This represents the classical analogue for quantum anti-particle states.}

\begin{prop}\label{prop:antistates}
	Let $\mathcal{T}_{\Delta s}: T^*\mathcal{M} \rightarrow T^*\mathcal{M}$ be the time dependent deterministic and reversible evolution map for the particles of a classical material. The map partitions the extended state space into standard states, those connected by a trajectory where $d_{s}t>0$, and anti-states, those connected by a trajectory where $d_{s}t<0$.
\end{prop}

\begin{justification}
	We claim $t(s)$ is strictly monotonic. The motion is deterministic and reversible. At each value of $t$ there must be only one possible state: if there is more than one state the motion is non-deterministic (one state is not enough to identify more than one state), if there is no state the motion is non-reversible (no state cannot identify a state). $t(s)$ is invertible. $t(s)$ is strictly monotonic.
	
	We claim $d_{s}t$ partitions the space. Let $\mathcal{s}_1, \mathcal{s}_2 \in T^*\mathcal{M}$ be two physical states connected by deterministic and reversible evolution. $\mathrm{sign}(d_{s}t(\mathcal{s}_1)) = \mathrm{sign}(d_{s}t(\mathcal{s}_2))$. Let $U \equiv \{\mathcal{s} \in T^*\mathcal{M} \; | \; d_{s}t(\mathcal{s}) > 0 \}$ and $V \equiv \{\mathcal{s} \in T^*\mathcal{M} \; | \; d_{s}t(\mathcal{s}) < 0 \}$. $U \cap V = \varnothing$. Let $\gamma : [s_0, s_1] \rightarrow T^*\mathcal{M}$ the trajectory given by deterministic and reversible evolution, where $[s_0, s_1]$ is the range of the parametrization. Either $\gamma([s_0, s_1]) \subseteq U$ or $\gamma([s_0, s_1]) \subseteq V$.
\end{justification}

With the invariant Hamiltonian we are able to write the equations of motion for any choice of time and state variables. For example: $d_t q^i = d_t s \, d_s q^i = d_s q^i / d_s t = - \partial_{p_i} \mathcal{H} / \partial_{E} \mathcal{H}$. These equations must be the same as the ones given by a standard Hamiltonian in those coordinates. That is: $d_t q^i = \partial_{p_i} H = - \partial_{p_i} \mathcal{H} / \partial_{E} \mathcal{H}$. The partial derivatives of $H$ and $\mathcal{H}$ are related, and therefore the functions themselves are related. Working through the math, we find that the most general relationship is of the form: $\mathcal{H} = \lambda(t,q^i,E,p_i) \, (H(t,q^i,p_i) - E - \mathcal{h}(t))$ where $\lambda = d_s t$, $\mathcal{h}(t)$ is an arbitrary function and, most importantly, $E = H(t,q^i,p_i) + \mathcal{h}(t)$.

The arbitrary function $\mathcal{h}(t)$ has no physical consequence: it is just a constant added to the Hamiltonian. We can either set it to zero or simply redefine the Hamiltonian to include it. Therefore we are left with $E = H(t,q^i,p_i)$. In other words: $E$ is the value of the Hamiltonian, the energy of the system. We have found the constraint that identifies the hypersuface of $T^*\mathcal{M}$ introduced in \ref{prop:relativistic_symplectic_manifold}.

Note that $\mathcal{H}=0$ for all states $\mathcal{s} \in \bar{\mathcal{S}} \subset T^* \mathcal{M}$, since $\mathcal{H}=\lambda(H-E)$ and $E = H$. For all states we also have $\lambda \neq 0$. Therefore $\mathcal{H}=0$ only because $E = H$. We can extend $\mathcal{H}$ on all of $T^* \mathcal{M}$ such that $\mathcal{H} \neq 0$ outside of the complete state space $\bar{\mathcal{S}}$. This way we can directly use the constraint $\mathcal{H}=0$ to identify $\bar{\mathcal{S}}$ as a subspace of $T^* \mathcal{M}$.

This way $\mathcal{H}$ gives both the constraint $\mathcal{H}=0$ to identify states and the map $d_s \xi^a=\omega^{ab} \partial_b \mathcal{H}$ to identify their evolution.

\begin{prop}\label{prop:form_of_invariant_hamiltonian}
	The form of the invariant Hamiltonian is given by $\mathcal{H} = \lambda(t,q^i,E,p_i) \,  (H(t,q^i,p_i) - E)$ where $H(t,q^i,p_i)$ is the standard Hamiltonian at each instant $t$, the conjugate time variable $E = H(t,q^i,p_i)$ represents the value of the Hamiltonian, and  $\lambda : \bar{\mathcal{S}} \rightarrow \mathbb{R}$ where $\lambda(\xi^a) = d_s t$. $\bar{\mathcal{S}}$ is the level set for $\mathcal{H} = 0$.
\end{prop}

\begin{justification}
	We claim the form of the extended Hamiltonian is $\mathcal{H} = \lambda(t,E,q^i,p_i) \, (H(t,q^i,p_i) - E)$. Let $\mathcal{H}$ be the invariant Hamiltonian and $H$ be the Hamiltonian for a particular choice of time and state variables. We have $d_s q^i = \partial_{p_i} \mathcal{H} = d_s t \, d_t q^i = - \lambda \, \partial_{p_i} H$ where $\lambda : T^*\mathcal{M} \rightarrow \mathbb{R}$ such that $\lambda \equiv d_s t$. $d_s p_i = - \partial_{q^i} \mathcal{H} = d_s t \, d_t p_i = \lambda \, \partial_{q^i} H$. As $\lambda \neq 0$ for all physical states, we can set $\mathcal{H} = \lambda(t,q^i,E,p_i) \, (H(t,q^i,p_i) - E + \mathcal{h}(t,q^i,E,p_i))$ without loss of generality. We have $\partial_{q^i} \lambda \, (H - E + \mathcal{h}) + \lambda \, (\partial_{q^i} H + \partial_{q^i} \mathcal{h}) = \lambda \, \partial_{q^i} H$. This holds for all choices of $t$. In particular for $t=s$, $\lambda \, (\partial_{q^i} H + \partial_{q^i} \mathcal{h}) = \lambda \, \partial_{q^i} H$. $\partial_{q^i} \mathcal{h} = 0$. $\mathcal{h}$ is not a function of $q^i$. Repeat the same procedure for $p_i$. $\mathcal{h}$ is not a function of $p_i$. $\lambda = d_s t = - \partial_E \mathcal{H} = - \partial_{E} \lambda (H - E + \mathcal{h}) - \lambda \, (-1 + \partial_{E} \mathcal{h})$. In particular for $t=s$, $\lambda = - \lambda (-1 + \partial_{E} \mathcal{h})$. $\partial_{E} \mathcal{h} = 0$. $\mathcal{h}$ is not a function of $E$. $\mathcal{h}$ is only a function of $t$. We can redefine $H = H + \mathcal{h}$ as the result is still a function of $q^i, p_i , t$ with the same equations of motion. $\mathcal{H}=\lambda (H - E)$.
	
	We claim $E = H(t,q^i,p_i)$ for physical states. $\lambda = d_s t = - \partial_E \mathcal{H} =\journal{\break} - \partial_{E} \lambda (H - E) - \lambda (-1)$ regardless of the choice of t and therefore for all $\lambda$. $H - E = 0$. 
	
	We claim $\bar{\mathcal{S}}$ is the level set for $\mathcal{H} = 0$. $\mathcal{H} = \lambda (H - E) = 0$ over $\bar{\mathcal{S}}$. We can extend $\mathcal{H}$ over $T^* \mathcal{M}$ such that $\mathcal{H} \neq 0$ over $T^*\mathcal{M} \,\backslash\, \bar{\mathcal{S}}$. $\bar{\mathcal{S}}$ is the level set for $\mathcal{H} = 0$.
\end{justification}

As an example, the Hamiltonian and invariant Hamiltonian for a relativistic free particle are:
\begin{equation}\label{free_hamiltonians}
\begin{aligned}
H &= c \sqrt{m^2 c^2 + p^i p_i} \\
\mathcal{H} &= \frac{1}{2m} ( p^i p_i - (E/c)^2 + m^2c^2) \\
\ifjournal
 &= \frac{1}{2mc^2} (c \sqrt{m^2 c^2 + p^i p_i} + E) (c \sqrt{m^2 c^2 + p^i p_i} - E) \\
\else
 &= \frac{1}{2mc^2} (c \sqrt{m^2 c^2 + p^i p_i} + E)\\ &(c \sqrt{m^2 c^2 + p^i p_i} - E) \\
\fi
\lambda &= \frac{1}{2mc^2} (c \sqrt{m^2 c^2 + p^i p_i} + E)
\end{aligned}
\end{equation}
Therefore we can see that the invariant Hamiltonian has the form $\mathcal{H}=\lambda (H - E)$. In the next section we'll derive the above expression and look more closely at the case of particles under potential forces.

Now that we have characterized the motion of the particles, we should also characterize the evolution of composite states.

Suppose we have a distribution of material $\rho_{\bar{\mathcal{c}}}(s)$ associated with an evolution $\bar{\mathcal{c}}$. Its support, the region where the density is non-zero, is within $\bar{\mathcal{S}}$. Therefore  $\mathcal{H} \rho_{\bar{\mathcal{c}}}(s) = 0$ since $\mathcal{H} = 0$ in $\bar{\mathcal{S}}$ where $\rho_{\bar{\mathcal{c}}}(s) \neq 0$. As the density is transported over each trajectory, its value does not change over $s$. Therefore we have $d_s \rho_{\bar{\mathcal{c}}}= d_s \xi^a \partial_a \rho_{\bar{\mathcal{c}}} = \omega^{ab} \partial_b \mathcal{H} \partial_a \rho_{\bar{\mathcal{c}}} = \{ \rho_{\bar{\mathcal{c}}}, \mathcal{H} \} = 0$ where $\{f , g\} = \partial_a f \omega^{ab} \partial_b g$ is the Poisson bracket.

\begin{prop}\label{prop:relativistic_state_evolution}
	Let $\rho_{\bar{\mathcal{c}}}$ be the distribution of material associated with the evolution $\bar{\mathcal{c}}$. $\mathcal{H} \rho_{\bar{\mathcal{c}}} = 0$ and $d_s \rho_{\bar{\mathcal{c}}} = \{ \rho_{\bar{\mathcal{c}}}, \mathcal{H} \} = 0$.
\end{prop}

\begin{justification}
	We claim $\mathcal{H} \rho_{\bar{\mathcal{c}}} = 0 \; \forall \bar{\mathcal{c}} \in \bar{\mathcal{C}}$. Let $\rho_{\bar{\mathcal{c}}} : T^*\mathcal{M} \rightarrow \mathbb{R}$ be the distribution associated with an evolution $\bar{\mathcal{c}}$. $supp(\rho_{\bar{\mathcal{c}}}) \subseteq \bar{\mathcal{S}}$ as the distribution is defined on the complete state space of the particles. $supp(\mathcal{H}) = T^*\mathcal{M} \,\backslash\, \bar{\mathcal{S}}$ by \ref{prop:form_of_invariant_hamiltonian}. $supp(\mathcal{H} \rho_\mathcal{c}) = supp(\mathcal{H}) \cap supp(\rho_\mathcal{c}) = \varnothing$.
	
	We claim $d_s \rho_{\bar{\mathcal{c}}} = \{ \rho_{\bar{\mathcal{c}}}, \mathcal{H} \} = 0$. $d_s \rho_{\bar{\mathcal{c}}}= \partial_a \rho_{\bar{\mathcal{c}}} d_s \xi^a = \partial_a \rho_{\bar{\mathcal{c}}} \omega^{ab} \partial_b \mathcal{H} = \{ \rho_{\bar{\mathcal{c}}}, \mathcal{H} \}$. Let $\gamma : \mathbb{R} \rightarrow T^*\mathcal{M}$ be a parametrization for a particle evolution $\bar{\mathcal{s}}$. Then $\rho_{\bar{\mathcal{c}}}(\lambda(s)) = \rho_{\bar{\mathcal{c}}}(\lambda(s + \Delta s))$ as the material does not change over the particle evolution. $d_s \rho_{\bar{\mathcal{c}}} = 0$.
\end{justification}

As we have derived relativistic motion in an unusual way, a few comments are in order.

As mentioned at the beginning of the section, no new assumption was required. Relativistic Hamiltonian mechanics is simply the correct way of handling time dependent motion. As such, there was no mention of the speed of light $c$, $\mathcal{M}$ being a pseudo-Riemannian manifold or its metric $g$. These will be introduced in the next section with the kinematic equivalence assumption. In light of this, it is probably better to consider relativistic mechanics as a feature of the geometry of the complete state space more than the geometry of space-time. Also note that this feature arises from the use of the standard topology on $\mathbb{R}^n$ and the need for invariant densities. It would not have arisen with the use of a discrete topology.

We have also seen that anti-particle states, generally associated with quantum field theory, are actually a feature of relativistic Hamiltonian mechanics. This is usually missed because in classical mechanics only the standard Hamiltonian is typically used. Note how the invariant Hamiltonian in \ref{free_hamiltonians} is quadratic in $E$, which can have both positive and negative solutions given the constraint $\mathcal{H}=0$.

Another important feature is that this formulation of relativistic Hamiltonian mechanics is formally very close to the quantum case. Note how $\mathcal{H} \rho = 0$ and $\mathcal{H} \ket{\psi} = 0$ are formally equivalent. If we apply to the invariant Hamiltonian in \ref{free_hamiltonians} the usual formal classical/quantum substitution, we have the Klein-Gordon equation using space-like convention rescaled by a factor of $2m$. As we'll see later, to this parallel in the mathematical formalism corresponds a parallel in the physical description.

As a final thought, we note how the physics maps less elegantly to the math in the time dependent case. The function space for $\bar{\mathcal{C}}$ is not simply the Lebesgue integrable differentiable functions, but those that are integrable on particular hypersurfaces. The complete state space $\bar{\mathcal{S}}$ is a hypersurface of $T^*\mathcal{M}$. The invariant Hamiltonian is needed to identify that subspace. This may be a hint that a better formulation is lurking, one where we consider the particle evolutions $\bar{\mathcal{s}}$ as the primary objects and label those (instead of the states). The distribution of material would simply be a distribution over the particle evolutions, a Lebesgue integrable differentiable function. Yet, the need to be able to describe the evolution in terms of physically meaningful state variables will most likely reimpose the extra structure.

\section{Kinematics and minimal materials}
\label{sec:Lagrangian}

We will now focus our attention on the simplest possible classical material: one that is fully described by its motion in space, i.e.~the internal dynamics can be ignored. What we'll find is that the dynamics of its infinitesimal parts follows classical Lagrangian mechanics. Furthermore, the Hamiltonian of the system is constrained to the one of a free particle under scalar and vector potential forces.

\subsection{Kinematic equivalence}

The link between kinematics, the study of motion through quantities like velocity and acceleration, and dynamics, the study of state evolution and its causes, is crucial in physics. In Newtonian physics, $\vec{F}=m\vec{a}$ provides such a link. In Hamiltonian (and Lagrangian) mechanics, the link is implicitly given by the choice of the Hamiltonian (or Lagrangian).

Our aim is to make that link more conceptually crisp and explicit with the following assumption:

\begin{assump}[Kinematic equivalence]\label{ass:kinematic_equivalence}
	For the system under investigation, studying its kinematics  (i.e. trajectories in space) is equivalent to studying its dynamics  (i.e. states and their evolution).
\end{assump}

\begin{rationale}
	The idea is that the system internal dynamics has no effect on the overall motion and therefore can be disregarded. Consider the motion of our cannonball under gravitational and inertial forces: its temperature or precise chemical composition does not affect the overall trajectory. We can assume kinematic equivalence. Suppose we now add magnetic forces to our consideration. The magnetization of the cannonball may now affect the motion and therefore we cannot always reconstruct the full state by simply looking at the trajectory. We cannot assume kinematic equivalence.
	
	Kinematic equivalence can't therefore be assumed for all systems. Yet, the assumption is fundamental because the particles of any material, during their evolution, will \emph{at least} describe a trajectory in space and what we learn about the link between kinematic and dynamical quantities applies to all materials, even the ones where the kinematic assumption does not hold.
	
	The first natural question to ask is: why does state evolution of a particle of a classical material always describe a trajectory in space? It's because spatial extent is the one way in which all materials can be decomposed. As we divide materials into smaller pieces, physical location is one state variable we can assign to the particles of all materials. If we assume kinematic equivalence, the kinematic variables will be enough to identify a particle state. In such a case, we will call the material \emph{minimal} to highlight the idea that the state is described by the smallest possible number of d.o.f.
	
	When the system is not minimal, the state of the particles will be described by some additional independent state variables. The overall dynamics will be the spatial d.o.f. coupled with the internal d.o.f. Therefore by studying a minimal material, we are also studying a component of the dynamics that is present in all materials. In other words: any material is an extension of a minimal material.
	
	It should also be noted that, once we reduce particles to be infinitesimal at a point, the rest of the description must be a local object invariant under coordinate transformations. This means that the remaining internal state must be identified by scalar, vector or tensor quantities. This may be used in the future to justify the extension of this work to field theories. 
\end{rationale}

\subsection{Space-time manifold and Lagrangian mechanics}

As we need to describe trajectories, we introduce $\mathcal{M}$ as the $n+1$ dimensional manifold that defines physical space and time. A point in $\mathcal{M}$ will be identified by $x^\alpha = (t, x^i)$, its time and space coordinates. A trajectory will be given by a set of functions $x^\alpha(s)$ where $s$ is an arbitrary parametrization of the curve. What we need to understand is the link between the complete state space of the particles of a classical material and the space of all the possible trajectories under the kinematic assumption.

We first note that the system of units for the trajectories is fully defined by $\mathcal{M}$. Any trajectory can be specified once a set of coordinates $x^\alpha$ is chosen. Under the kinematic assumption, these units must also be sufficient to describe the dynamics. Therefore $\mathcal{M}$ has conceptually two roles: it is the space where the motion unfolds and it is the space on which we define the system of units for our states. This already tells us that the state space for the particles of the material is a hypersurface of $T^*\mathcal{M}$.

Under kinematic equivalence, given a trajectory we must identify a state and vice-versa. There must be a homeomorphism between the complete state space and the space of possible trajectories. We already know what the complete state space is. We just need to find a way to label the trajectories so that we can express the correspondence.

Because of the double role of $\mathcal{M}$, we can set half of our state variables to be equal to the space-time variables: $q^\alpha = (t, q^i) = x^\alpha$. This already tells us that the trajectory is differentiable since $q^\alpha(s)$ is differentiable. The four-velocity $u^\alpha = d_s x^\alpha$ is therefore well defined and it can be used to identify trajectories. The magnitude of the four-velocity can be rescaled by changing the evolution parameter, $u^\alpha = d_s x^\alpha = d_s \hat{s} \, d_{\hat{s}} x^\alpha = d_s \hat{s} \, \hat{u}^\alpha$, and is therefore not physical. Only the direction is physical so the four-velocity gives us another $n$ independent variables to identify trajectories. The combined position and velocity, the kinematic variables, identify $2n + 1$ trajectories using a hypersurface of the tangent bundle $T\mathcal{M}$. This is exactly the dimensionality of the complete state space $\bar{\mathcal{S}} \subset T^*\mathcal{M}$. We cannot add more curves, as a homeomorphism preserves dimensionality, therefore we identified the suitable space for the trajectories of particles of a classical material under kinematic equivalence.

Being able to identify particles through either state or kinematic variables is not enough: we also need to be able to express a composite state as the distribution of particle states over position and velocity. As such, the differentiable distribution over state variables must become a differentiable distribution over kinematic variables. This tells us the transformation is differentiable.

Mathematically, the state space $\bar{\mathcal{S}} \subset T^*\mathcal{M}$ is diffeomorphic to a hypersurface of the tangent bundle $T\mathcal{M}$ of the space-time manifold.

\begin{prop}\label{prop:tangent_bundle}
	Let $\mathcal{M}$ be the $n+1$ dimensional space-time manifold over which trajectories are defined. The complete state space $\bar{\mathcal{S}} \subset T^*\mathcal{M}$ for the particles of a minimal classical material is diffeomorphic to a hypersurface of the tangent bundle $T\mathcal{M}$.
\end{prop}

\begin{justification}
	We claim the complete state space is a hypersurface of the cotangent bundle of the space-time manifold, $\bar{\mathcal{S}} \subset T^*\mathcal{M}$. Let $\bar{\mathcal{s}} \in \bar{\mathcal{S}}$ be the evolution of a particle for a minimal classical material. Let $f(\bar{\mathcal{s}}) \mapsto x^\alpha(s)$ be the function that returns the space-time trajectory associated with the evolution. $f$ is injective as the material is minimal. $\bar{\mathcal{s}}$ can be identified by a trajectory $x^\alpha(s)$ which can be fully described by coordinates defined on $\mathcal{M}$. Let $\mathcal{s} \in \bar{\mathcal{S}}$ be the particle state for a minimal classical material. Let $f(\bar{\mathcal{s}}, t) \mapsto \mathcal{s}$ the function that returns the state at a particular time $t$ along an evolution $\bar{\mathcal{s}}$. Such function is a bijection as the evolution is deterministic and reversible. $\mathcal{s}$ can be identified by an evolution $\bar{\mathcal{s}}$ and an instant $t$, both fully described by coordinates defined on $\mathcal{M}$. $\mathcal{M}$ defines the unit system. $\bar{\mathcal{S}} \subset T^*\mathcal{M}$ by \ref{prop:relativistic_symplectic_manifold}.
	
	We claim there exists a homeomorphism between $\bar{\mathcal{S}} \subset T^*\mathcal{M}$ and a hypersurface of $T\mathcal{M}$. Let $\mathcal{s} \in \bar{\mathcal{S}}$ be the particle state for a minimal classical material. Let $x^\alpha(s)$ be the evolution associated with that state. Let $t(\mathcal{s})$ be the instant of time at which $\mathcal{s}$ is defined. $q^\alpha(\mathcal{s}) \equiv x^\alpha(s(t(\mathcal{s})))$ are state variables as they are real function of the state. $q^\alpha$ is differentiable therefore $d_s x^\alpha$ exists. Suppose $d_s x^\alpha$ is not part of the state. Then $d_s x^\alpha = d_s q^\alpha$ must be determined by the law of evolution. $x^\alpha(s)$ is determined by the law of evolution therefore $q^\alpha$ are the only state variables needed from the trajectory. As the material is minimal, $q^\alpha$ determines the full state $\mathcal{s}$. This is a contradiction: $\mathcal{s}$ is not identified by just $q^\alpha$. Therefore $u^\alpha=d_s x^\alpha$ is part of the state. $p_\alpha = p_\alpha(x^\alpha, u^\alpha)$. The map $(x^\alpha, u^\alpha) \mapsto (q^\alpha, p_\alpha)$ is continuous by \ref{prop:continuity} and is bijective as the material is minimal. It is a homemorphism between $T^*\mathcal{M}$ and $T\mathcal{M}$. $\bar{\mathcal{S}}$ is a hypersurface of $T^*\mathcal{M}$ therefore its image is a hypersurface of $T\mathcal{M}$.
	
	We claim $(x^\alpha, u^\alpha) \mapsto (q^\alpha, p_\alpha)$ is a diffeomorphism. Let $\bar{\mathcal{c}}$ be the evolution of a minimal classical material. Let $\rho_{\bar{\mathcal{c}}} : \bar{\mathcal{S}} \rightarrow \mathbb{R}$ be the distribution associated with $\bar{\mathcal{c}}$. $\rho_{\bar{\mathcal{c}}}(\mathcal{s}(x^\alpha, u^\alpha))$ as the material is minimal. $\rho_{\bar{\mathcal{c}}}(x^\alpha, u^\alpha)$ is differentiable by the same logic expressed in \ref{prop:differentiable_manifold}. $(x^\alpha, u^\alpha) \mapsto (q^\alpha, p_\alpha)$ is a diffeomorphism as it is a homeomorphism that maps a differentiable function to a differentiable function.
\end{justification}

This relationship allows us to write the kinematic variables in terms of the state variables. In particular we have $u^\alpha=u^\alpha(q^\alpha, p_\alpha)$. As the relationship is invertible, these functions must be monotonic in $p_\alpha$. Since $\partial_{p_\alpha} \mathcal{H}=d_s q^\alpha= u^\alpha$, $\mathcal{H}$ is convex in $p_\alpha$. Under such conditions we can define a Legendre transformation $\mathcal{L}=u^\alpha p_\alpha - \mathcal{H}$. Such transformation is the Lagrangian of the system and the motion is given by the Euler-Lagrange equation. For the time independent case, we define $L=v^i p_i - H$ where $v^i=d_t x^i$.

\begin{prop}\label{prop:lagrangian}
	Let $\mathcal{H}$ be an invariant Hamiltonian defined on $T^*\mathcal{M}$. Under the kinematic assumption, we can define the Legendre transform $\mathcal{L}=u^\alpha p_\alpha - \mathcal{H}$ which we call the \emph{extended Lagrangian}.
\end{prop}

\begin{justification}
	We claim $\mathcal{H}$ is convex in $p_\alpha$. By \ref{prop:tangent_bundle} we can write the velocity $u^\alpha=u^\alpha(q^\alpha, p_\alpha)$ as a differentiable function of the state variables. At fixed $q^\alpha$, $u^\alpha$ must be invertible. The components $u^\alpha$ are monotonic in $p_\alpha$. $\partial_{p_\alpha} \mathcal{H}=d_s q^\alpha= u^\alpha$. The derivatives of $\mathcal{H}$ in $p_\alpha$ are monotonic in $p_\alpha$. $\mathcal{H}$ is convex in $p_\alpha$.
	
	We claim there exists an extended Lagrangian $\mathcal{L}=u^\alpha p_\alpha - \mathcal{H}$. $\mathcal{H}$ is convex in $p_\alpha$. It admits a Legendre transform. The extended Lagrangian $\mathcal{L}$ is such transform.
\end{justification}

Not only must we be able to express distributions in terms of both state and kinematic variables, but integration has to carry over as well. To do that, we need to be able to count trajectories in the same way we are able to count states. Just as the symplectic form $\omega$ gives us the count of possibilities for a d.o.f. in terms of $dq^\alpha$ and $dp_\alpha$, the ranges of the state variables, there should be a bilinear form $g$ that gives the the count of possibilities in terms of $dx^\alpha$ and $du^\alpha$, the ranges of the kinematic variables. Under the kinematic assumption, these counts must be consistent with each other. That is: $\omega(dq^\alpha,dp_\alpha) \propto g(dx^\alpha,du^\alpha)$.

Since both $dx^\alpha$ and $du^\alpha$ are four vector components, $g$ is a rank two tensor on $\mathcal{M}$. To each direction in space-time $dx^\alpha$ is associated a direction in the tangent space of velocities $du^\alpha$ with which it forms a d.o.f. Across d.o.f., that is $\alpha \neq \beta$, we have $\omega(dq^\alpha, dp_\beta) = 0 =\omega(dq^\beta, dp_\alpha)$. Therefore if $dx^\alpha$ and $du^\beta$ do not form a d.o.f., $g(dx^\alpha, du^\beta) = 0 = g(dx^\beta, du^\alpha)$: $g$ is reflexive. We also know that $g$ cannot be alternating, that is $g(dx^\alpha, du^\alpha) \neq 0$: a choice of units for $x^\alpha$ is a choice of units for $u^\alpha$ and therefore they cannot belong to different d.o.f. This means $g$ is a symmetric non-degenerate tensor. We recognize it as the metric tensor and therefore $\mathcal{M}$ is a pseudo-Riemannian manifold.

Since $g$ is symmetric, it can be diagonalized: $\omega(dq^\alpha, dp_\beta)= dq^i \omega_{q^ip_i} dp_i + dt \omega_{tE} dE \propto g(dx^\alpha, du^\beta) = dx^i g_{ii} du^i + dt g_{00} du^0$. As $\omega_{q^ip_i} > 0$ we must have $g_{ii} > 0$. As $\omega_{tE} < 0$ we must have $g_{00} < 0$. We can further change the units along the spatial directions so that $g_{ii}=1$: if we use the same units for the spatial coordinates, each direction will contribute the same number of possibilities. Since time must use a different unit, we set $g_{00}=-c^2$ where $c^2$ is the constant that allows one to convert the possibility count over the temporal d.o.f. into the possibility count of a spatial d.o.f. Therefore, locally $g(dx^i, dx^i) = (dx^i)^2 - c^2 dt^2$ is Minkowskian.

\begin{prop}\label{prop:riemannian_manifold}
	The space-time manifold $\mathcal{M}$ is a pseudo-Riemannian manifold $(\mathcal{M}, g)$. $g$ has signature $(n,1)$ and locally is $g = (dx^i)^2 - c^2 dt^2$.
\end{prop}

\begin{justification}
	We claim $(\mathcal{M}, g)$ is a pseudo-Riemannian manifold where $g$ is a metric tensor induced by $\omega$ under kinematic equivalence. Let $(dx^\alpha, du^\alpha)$ be an infinitesimal range of possibilities for the kinematic variables of the particles of a minimal classical material. There exists a map $g : T\mathcal{M} \times T\mathcal{M} \rightarrow \mathbb{R}$ such that $g(dx^\alpha, du^\alpha) \propto \omega(dq^\alpha, dp_\alpha)$. The map $g$ is linear, since a linear combination of infinitesimal ranges leads to the linear combination of possibilities. $g$ is a rank two tensor. Suppose $g(dx^\alpha, du^\beta) = 0$. $g(dx^\alpha, du^\beta) \propto \omega(dq^\alpha, dp_\beta) = 0 = \omega(dq^\beta, dp_\alpha) \propto g(dx^\beta, du^\alpha) = 0$. $g$ is reflexive. Position and velocity along the same direction use related units. They cannot form independent d.o.f. $g(dx^\alpha, du^\alpha) \neq 0$. $g$ is not alternating. $g$ is symmetric. $g$ is a symmetric rank two tensor. $(\mathcal{M}, g)$ is a pseudo-Riemannian manifold.
	
	We claim $g$ is locally Minkowskian. Let $x^\alpha$ be a choice of coordinate that diagonalizes $g$. $x^\alpha$ exists as $g$ is symmetric. $\omega(dq^\alpha, dp_\beta) = dq^i \omega_{q^ip_i} dp_i + dt \omega_{tE} dE \propto g(dx^\alpha, du^\beta) \propto dq^i g_{ii} du^i + dt g_{00} du^0$. $dq^\alpha = dx^\alpha$. We require $g(dx^\alpha, du^\beta)>0$ when the arguments are positive. Therefore $g_{ii} > 0$ and $g_{00} < 0$. $g$ has signature $(n, 1)$. Choosing the same unit for all $dx^i$ and a temporal unit for $dt$ we can choose coordinates such that locally $g = (dx^i)^2 - c^2 dt^2$. $g$ is locally Minkowskian.
\end{justification}

The metric tensor $g$ and the constant $c$ found here are indeed the ones used in the context of special relativity. Yet, they have acquired a different meaning. $g$ does not measure space-time distances: it counts possibilities of a d.o.f., like $\omega$. Directions in space are orthogonal not because the angle between them is $90^\circ$, but because they identify independent d.o.f. That is, the choice of units along one direction is independent from the choice of units along orthogonal directions. Space and time, instead, are not independent d.o.f. and the geometry between them is different. In other words: $g$ inherits its signature from $\omega$. The geometry of space-time is the geometry of the ranges for possible states.

This also allows us to give $c$ a more insightful physical meaning. It is not a velocity: the metric is defined at each point without motion. It is the constant of proportionality that maps a range of initial conditions from the time variable $t$ to the space variable $x^i$. For example, once the SI is adopted, $1$km would correspond to a thousand times more possible states than $1$m, and, in the same way, $1$sec would correspond to $299,792,458$ more possible states than $1$m. Therefore when we write $x^\alpha = (ct, x^i)$ or $p_\alpha = (- E/c, p_i)$ we are really expressing everything in quantities that are proportional to the count of states and possibilities. These coordinates make the bookkeeping of number of states much easier which, as we saw, is the only thing that Hamiltonian and Lagrangian mechanics are doing.

This also allows us to understand heuristically why the speed of deterministic and reversible motion cannot be greater than $c$. Consider the trajectory of a point particle. Over an interval $dt$ it will travel $dx$ in space. Since the motion is deterministic and reversible, and therefore continuous, it will travel through all the states defined in the interval $dx$. It will also travel through every moment in time, which corresponds to $cdt$ possibilities. As it moves, time will always change but not necessarily space. The number of states traveled in space cannot exceed the number of states traveled in time. Therefore we have:
\begin{align*}
\frac{|n_x(dx)|}{|n_t(dt)|} &=  \frac{|dx|}{|c dt|} \leq 1 \\
\frac{|dx|}{|dt|} & \leq c \\
\end{align*}
Note, though, that this constraint applies only to deterministic and reversible motion. It was under that premise that we were able to show that the trajectories are continuous, that we could define $\omega$ and therefore $g$. Non-local quantum effects, such as the ones featured in EPR~\cite{EPR} experiments, are not examples of deterministic and reversible evolution. Therefore they do not have anything to violate: the maximum speed does not apply to their ``motion."\footnote{The very notion of velocity becomes ill defined in such cases precisely because we do not have a unique well-defined evolution.} On the other hand, only deterministic and reversible motion can carry information and can establish before/after relationships, therefore everything remains consistent.

\subsection{Massive particles under scalar and vector potential forces}

The link between $\omega$ and $g$ allows us to constrain the possible diffeomorphisms between state variables and kinematic variables. If we express the possibilities in terms of both state variables and kinematic variables we have: $dq^\alpha \wedge dp_\alpha \propto dx^\alpha du_\alpha$. Since  $p_\alpha=p_\alpha(x, u)$, we have $dp_\alpha = m du_\alpha$, where $m$ is the constant of proportionality. We recognize $m$ as the inertial mass. Integrating the expression, we have $p_\alpha = m u_\alpha + A_\alpha(x)$ where $A_\alpha(x)$ is an arbitrary vector field.\footnote{Note that $A_\alpha(x)$ corresponds to the product the vector potential and the charge. This simplifies the notation and avoids having to find a symbol for the charge itself.} This is the most general link between velocity and conjugate momentum, which is the sum of kinetic momentum and a scalar potential.

This constrains the form of the extended Hamiltonian. $\frac{1}{m}g^{\alpha\beta}(p_\beta-A_\beta) = u^\alpha = d_s x^\alpha = d_s q^\alpha = \partial_{p_\alpha} \mathcal{H}$. Integrating we have $\mathcal{H}=\frac{1}{2m}(p_\alpha-A_\alpha)g^{\alpha\beta}(p_\beta-A_\beta)+V(x)$ where $V(x)$ is an arbitrary function. We can further constrain the Hamiltonian by choosing the natural parameterization for $s$ (also called unit speed~\cite{Lee} or arc length parametrization) for which the modulus of the velocity is constant. We have $u^\alpha u_\alpha = - c^2 = \frac{1}{m^2}(p^\alpha-A^\alpha)(p_\alpha-A_\alpha)$. $\mathcal{H}=\frac{1}{2m}((p^\alpha-A^\alpha)(p_\alpha-A_\alpha) + m^2 c^2)$. Applying the same procedure in the time independent case yields the Hamiltonian $H=\frac{1}{2m}(p^i-A^i)(p_i-A_i)+V$.

We recognize the Hamiltonian for massive particles under scalar and vector potential forces.

\begin{prop}\label{prop:kinetic_hamiltonian}
The extended Hamiltonian for the particles of a classical material under the kinematic assumption is $\mathcal{H}=\frac{1}{2m}((p_\alpha-A_\alpha)g^{\alpha\beta}(p_\beta-A_\beta) + m^2 c^2)$.
\end{prop}

\begin{justification}
	We claim the conjugate momentum expressed as a function of the kinematic variables is $p_\alpha = m u_\alpha + A_\alpha(x)$ where $A_\alpha \in \mathfrak{X}(\mathcal{M})$ is a differentiable vector field. $\omega(dq^\alpha, dp_\beta) = dq^\alpha \wedge dp_\alpha \propto g(dx^\alpha, du^\beta) = dx^\alpha du_\alpha$. Let $m$ be the proportionality constant. $dq^\alpha \wedge dp_\alpha = dx^\alpha \wedge \partial_{u_\beta} p_\alpha du_\beta = m dx^\alpha du_\alpha$. $\partial_{u_\beta} p_\alpha = m \delta_\alpha^\beta$. After integration $p_\alpha = m u_\alpha + A_\alpha(x)$ where $A_\alpha$ is an arbitrary vector field.
	
	We claim the Hamiltonian is $\mathcal{H}=\frac{1}{2m}((p_\alpha-A_\alpha)g^{\alpha\beta}(p_\beta-A_\beta) + m^2 c^2)$.  $\frac{1}{m}g^{\alpha\beta}(p_\beta-A_\beta) = u^\alpha = d_s x^\alpha = d_s q^\alpha = \partial_{p_\alpha} \mathcal{H}$. After integration $\mathcal{H}=\frac{1}{2m}(p_\alpha-A_\alpha)g^{\alpha\beta}(p_\beta-A_\beta)+V(x)$ where $V$ is an arbitrary function. Let $s$ be the natural parametrization such that $u^\alpha u_\alpha = - c^2$. $V= \frac{1}{2}m c^2$ as $\mathcal{H}=0$ by \ref{prop:form_of_invariant_hamiltonian}. $\mathcal{H}=\frac{1}{2m}((p_\alpha-A_\alpha)g^{\alpha\beta}(p_\beta-A_\beta) + m^2 c^2)$.
\end{justification}

The above discussion gives the inertial mass $m$ a new physical meaning: it is the conversion constant that tells us how to convert a range of velocity to a range of conjugate momentum with the same number of possibilities. In other words: the more massive the particles, the more states per unit velocity. Intuitively, a more massive body is harder to accelerate because it has to go through more states for the same change in velocity. Also note that the Hamiltonian excludes massless particles because the kinematic assumption does that implicitly. A massless particle always has the same velocity, $c$, therefore the trajectory is not enough to tell us the dynamics. The motion of the photon is not enough to tell us its energy or its momentum. Therefore it makes sense that, by assuming kinematic equivalence we derived an equation valid only for massive particles.

We have seen that the inertial mass $m$, the invariant speed $c$ and the reduced Planck constant $\hbar$ all serve similar purposes: convert one quantity to another preserving the number of possibilities. $dp=\hbar dk$ converts between ranges of wave number and conjugate momentum. $dx = c dt$ converts between ranges of time and space. $dp = m du$ converts between ranges of velocity and conjugate momentum. They allow us to keep proper bookkeeping when mapping sets of states. Using natural units means measuring all these quantities through the possibility count. There is a difference among them, though: while $c$ and $\hbar$ are simply determined by the unit system, $m$ also depends on the type of material. That is, different homogeneous materials can have different inertial mass. The reason is that the choice of units for $x^\alpha$ already determines the others ($x^\alpha \rightarrow u^\alpha$ and $x^\alpha \rightarrow q^\alpha \rightarrow k_\alpha \rightarrow p_\alpha$) therefore the link between units of $u^\alpha$ and $p_\alpha$ is not a new unit conversion. In other words: $m$ is not just another arbitrary constant set by the choice of units.

Another important thing to note is that the relationship between conjugate momentum and velocity, $p_\alpha = m u_\alpha + A_\alpha(x)$, includes an arbitrary function. As it is arbitrary, we cannot give $p$ a well defined physical meaning in terms of the space-time trajectory. Suppose, in fact, we changed $A$ and $p$ while leaving $u$ unchanged, the physical description in terms of trajectories remains the same.

We should also ask whether any choice of $A$ actually represents a different case leading to different trajectories. The simplest and more insightful way to study this problem is to use the kinematic quantities $x^\alpha$ and $mu_\alpha$ as state variables. They are not going to be conjugates of each other but they do capture the entire physical description. Any change in $A$ that does not affect the relationship between $x^\alpha$ and $mu_\alpha$ does not change the physics. Expressing the Hamiltonian and the Poisson brackets in terms of the kinematic variables we have:
\begin{align*}
\mathcal{H}&=\frac{1}{2}mu_\alpha u^\alpha + \frac{1}{2} m c^2 \\
\{x^\alpha, x^\beta\} &= 0 \\
\{x^\alpha, mu_\beta\} &= \delta^\alpha_\beta \\
\{mu_\alpha, mu_\beta\} &= \{p_\alpha, - A_\beta\} + \{-A_\alpha, p_\beta\} \\
&= \partial_\alpha A_\beta - \partial_\beta A_\alpha = F_{\alpha\beta}
\end{align*}
Note how the constraint $\mathcal{H}=0$ simply sets the rest kinetic energy. Note how the Poisson bracket between kinetic momentum is the force tensor. The weaker the force, the more $(x^i, mu_i)$ become independent d.o.f. The force tensor, then, tells how much $(x^i, mu_i)$ fail to be independent d.o.f. With these expressions, the Poisson bracket between any function of position and kinetic momentum can now be calculated. For example:
\begin{align*}
d_s x^\alpha &= \{x^\alpha, \mathcal{H}\} = u^\alpha \\
d_s (mu^\alpha) &= \{mu^\alpha, \mathcal{H}\} = F^{\alpha\beta} u_\beta
\end{align*}
where we recognize the covariant form of the Lorentz force.

Any change of $A$ that leaves $F_{\alpha\beta}$ unchanged does not affect the evolution of any function of position and velocity. Such transformation is what we call a gauge transformation. Note the same group of transformations is found by fixing $A$ and changing $p$ such that the laws of motion don't change. In other words: $A$ and $p$ are defined relatively to each other. $A$ defines the $0$ for $p$ and vice-versa. The gauge group expresses the arbitrariness of such a definition.

The extended Hamiltonian we found is quadratic in energy, therefore negative energy states are possible. These correspond to the particle anti-states mentioned in \ref{prop:antistates}. As we fixed the natural parametrization, $u^\alpha u_\alpha = - c^2$, there are two cases: the parametrization is aligned with time, therefore $s=\tau$ where $\tau$ is proper time, or the parametrization is anti-aligned with time, therefore $s=-\tau$. The trajectory velocity $u^\alpha= \pm d_\tau x^\alpha$ is plus or minus the physical velocity depending on whether it's a standard particle state or an anti-state. Expanding the equation for the force, we have:
\begin{align*}
\pm d_\tau (\pm m d_\tau x^\alpha) &= m d_\tau^2 x^\alpha = \pm  F^\alpha_\beta d_\tau x^\beta
\end{align*}
Therefore the force acts on anti-particle states in the opposite way: they act as if they have opposite charge. Note that nothing physical is actually going backwards in time, just the parametrization.

In this setting, charge conjugation $C$ is simply changing the parametrization from $s$ to $-s$. This will permute particle standard and anti-states as they follow the trajectories in the opposite direction. But if we apply parity $P$ and time reversal $T$, changing $x^\alpha$ to $-x^\alpha$, we would have the same net effect: the parametrization stays the same, but it now moves backwards in space and time. As they have the same net effect, if we apply all three we go back where we started. This gives us a very intuitive classical particle mechanics analogue for the $CPT$ theorem~\cite{Weinberg,ClCPT}.

Again we want to stress just how the math takes a very simple and physically meaningful form when studying the motion in phase space using position and kinematic momentum. It is most expressive because, by using the kinematic variables of a deterministic and reversible motion of an infinitesimal amount of material, it makes use of all the assumptions about the system we are studying.

To sum up, the kinematic assumption has more inherent consequences than one would expect at first. It means the particles are identified by position and velocity and therefore the system of equations is second order. It links the space of initial conditions on $T\mathcal{M}$ with the space of states on $T^*\mathcal{M}$. Such a link allows phase space symplectic form $\omega$ to induce the space-time metric tensor $g$. It allows us to define Lagrangians. It constrains the motion to massive particles under scalar and vector potential forces. Moreover, it allowed us to reunderstand basic concepts, like the invariant speed $c$, the inertial mass $m$ and the force tensor $F_{\alpha\beta}$ in terms of state counting and d.o.f. independence. Finally, it allowed us to construct classical equivalents for anti-particles and the $CPT$ theorem.

We now have all the building blocks of classical mechanics, the consequences of three seemingly simple assumptions. We have more insight into why these blocks fit together the way they do. But we have also placed ourselves in a very good position to move forward to quantum systems by using the new physical insights we have gained.

\section{Prelude to quantum}
\label{sec:quantum_prelude}

In section \ref{sec:reducibility} we made the key assumption of infinitesimal reducibility which led to classical Hamiltonian mechanics. We are now going back to that point, that is we set aside the kinematic assumption and time dependent evolution, and take a different path. But before introducing assumption \ref{ass:irreducibility} and deriving quantum mechanics as its consequence, we first give a semi-classical account of the conceptual picture that emerges. This allows us to build the intuition more gradually and to show that classical and quantum mechanics are not as different as one may first imagine.

We will see that, as a first approximation, quantum particles are more suitably described by classical composite states (i.e. distributions in phase space) like any other object made of a finite amount of material. We'll also see that classical distributions already possess qualitative properties that are typically associated with quantum systems, including a classical analogue of the uncertainty principle.

\subsection{The case against point particles}

As we defined them in previous sections, classical particles are the infinitesimal limit of subsequent subdivisions and therefore should not be understood as point-like objects of finite mass. It was their being infinitesimal and spread around a small neighborhood that gave us conjugate pairs of state variables and Hamiltonian evolution. Any finite amount of classical material, therefore, is described by a distribution. If we take a quantum particle, say an electron, it has a finite amount of charge/mass/energy. Therefore, to a first approximation, it must be described by a distribution as well.

The idea that an electron is not a point particle may, to some, seem preposterous at first. But we contend that it is no more preposterous than to think it point-like, described by a probability distribution that self-interacts. In fact, the idea of an electron being a distribution of material is already conceptually more compatible with two features of quantum mechanics.

The first one is that quantum particles are not associated with a well defined value for position and momentum. If we misidentify the electron with a classical particle (i.e. the infinitesimal part of a classical material), this would indeed be worthy of note and puzzlement. But if the electron is a distribution, there is no issue because a phase space distribution does not have a single value for position or momentum. The earth is, so to speak, in both North America and Europe at the same time. In fact, based on what we saw before, any object of finite mass is a distribution over a region of phase space, and discontinuous distributions are ruled out by \ref{prop:continuity}. An electron is an object of finite mass: why should it be any different? In fact, if one claimed an electron to be point-like, he would be claiming an electron is fundamentally different than all other physical objects of finite mass/energy (which are the only ones we experimentally work with). Shouldn't \emph{he} have the burden of proof?

The other feature is particle-wave duality. If the spread of the distribution is narrow enough for the problem at hand, the motion can be described by just the position and momentum of the center of mass, making it behave like a classical particle. If the spread is not negligible (e.g. the object encounters obstacles or slits of a smaller scale) then we have to use the full distribution. Therefore, even in this classical approximation, the electron would have this dual nature depending on the problem.

While this picture is quite reasonable, some people may reject it because they believe that a fundamental object must have no size, it must be a point. But that does not follow: lack of spatial extent and lack of substructure are independent concepts.

In fact, we can conceptually construct an infinitesimal object that is not elementary. Suppose we have dots moving on the surface of a sphere. Their motion can be expressed in terms of their polar and azimuthal angles. If we shrink the sphere, the dynamics and its description stays the same. We can collapse the whole sphere to a point: it is now a point-like object that is not elementary, as there is motion detectable along different angles of incidence.

Whether a physical system has a distinguishable substructure depends on whether there exists a physical process that can interact with one of its parts independently of the others. If we are always forced to interact with the whole system, then we can't find any discernible part. The fact that the system is extended in space is a separate property, which does not necessarily imply the existence of such a process.

If we abandon the idea of point-like objects of finite mass, we avoid these problems and get a picture more consistent across the different realms of physics. And we can also see that classical mechanics already has some features that are commonly ascribed to quantum mechanics.

\subsection{Classical uncertainty principle}

Suppose we have our classical electron described by a finite amount of classical material distributed over a single d.o.f. $\int \rho_\mathcal{c} dq \wedge dp = 1$ where $\rho_\mathcal{c}$ is the normalized distribution. In this context, we will call a \emph{fragment} an infinitesimal part of our classical electron. The Shannon entropy~\cite{Shannon} associated with the distribution is finite as well. $I(\rho_\mathcal{c}) = - \int \rho_\mathcal{c} \log (\rho_\mathcal{c})  dq \wedge dp = I_0$. Intuitively, this value can be thought of as the amount of information required to identify a fragment within the distribution.\footnote{To be precise, Shannon's differential entropy is actually the number of bits required to identify an element of the given distribution relative to one of a uniform distribution of unitary range on all axes.}

As the state evolves into its final state, the amount of material remains the same: $\int \rho_{\hat{\mathcal{c}}} dq \wedge dp = 1$. The final entropy will also remain the same: $I(\rho_{\hat{\mathcal{c}}}) = I_0$. Mathematically, because $-\rho_\mathcal{c} \log (\rho_\mathcal{c})$, like the distribution itself, is just another integrable function, its integral is conserved by Hamiltonian evolution. Physically, the entropy remains unchanged because the amount of information required to identify a fragment given the distribution has to remain the same under deterministic and reversible evolution. Therefore, in the same way that each fragment is constrained to move through states at equal energy, the composite state is constrained to move through distributions at equal entropy.

We now ask the following question: once we fix the entropy, what is the relationship between the spread of the marginal distributions along $q$ and $p$? Using Lagrange multipliers, we fix the total amount of material to $1$ and the total entropy $I_0$ and find the distribution that minimizes the product of the variance $\sigma_q^2 \sigma_p^2 \equiv \int (q-\mu_q)^2 \rho_{\mathcal{c}} \, dqdp \int (p-\mu_p)^2 \rho_{\mathcal{c}} \, dqdp$.
\begin{align*}
L = &\int (q-\mu_q)^2 \rho_{\mathcal{c}} \, dqdp \int (p-\mu_p)^2 \rho_{\mathcal{c}} \, dqdp \\
\journal{ 
	&+ \lambda_1(\int \rho_{\mathcal{c}} dqdp - 1) + \lambda_2(- \int \rho_{\mathcal{c}} \ln \rho_{\mathcal{c}} \, dqdp - I_0)\\ }
\arxiv{ 
	&+ \lambda_1(\int \rho_{\mathcal{c}} dqdp - 1) \\ &+ \lambda_2(- \int \rho_{\mathcal{c}} \ln \rho_{\mathcal{c}} \, dqdp - I_0)\\ }
\delta L = &\int \delta \rho_{\mathcal{c}} [(q-\mu_q)^2 \sigma_p^2 + \sigma_q^2 (p-\mu_p)^2 + \\ &\lambda_1 - \lambda_2 \ln \rho_{\mathcal{c}} - \lambda_2 ] dqdp = 0 \\
\lambda_2 \ln \rho_{\mathcal{c}} = &\lambda_1 - \lambda_2 + (q-\mu_q)^2 \sigma_p^2 + \sigma_q^2 (p-\mu_p)^2 \\
\rho_{\mathcal{c}} = &e^{\frac{\lambda_1 - \lambda_2}{\lambda_2}}e^{\frac{(q-\mu_q)^2 \sigma_p^2}{\lambda_2}}e^{\frac{\sigma_q^2 (p-\mu_p)^2}{\lambda_2}}\\
\end{align*}
We solve the multipliers and have:
\begin{align*}
\rho_{\mathcal{c}} = &\frac{1}{ 2 \pi \sigma_q \sigma_p} e^{-\frac{(q-\mu_q)^2}{2\sigma_q^2}} e^{-\frac{(p-\mu_p)^2}{2\sigma_p^2}} \\
I_0 = &\ln (2\pi\sigma_q\sigma_p) + 1
\end{align*}
We find that the distribution that minimizes the spread is the product of two independent Gaussians. Recall that in quantum mechanics, the Gaussian wave packet is the state that minimizes uncertainty: this is the classical analogue.

As the distribution evolves, the entropy is conserved and therefore the product $\sigma_q^2 \sigma_p^2$ can never be less than the one given by the Gaussian distribution of the same entropy. We have:
\begin{align*}
\sigma_q\sigma_p \geq \exp (I_0 - 1) / 2 \pi 
\end{align*}
This is strikingly similar to the Heisenberg uncertainty principle, except that the value depends on the initial entropy $I_0$.

Suppose, though, that while the evolution of the distribution of our electron as a whole is deterministc and reversible, each fragment is undergoing non-deterministic motion. Suppose $\rho$ is a distribution at equilibrium of such motion. While each fragment moves randomly, the overall distribution $\rho$ stays the same: the non-deterministic motion is effectively just reshuffling the fragments within the same distribution. The value of the entropy of the distribution, and therefore its spread, is set by the non-deterministic motion of the fragments: the stronger it is, the greater the entropy.

Suppose now our classical electron undergoes quasi-static evolution. At each instant we can assume $\rho$ to be at equilibrium. At each instant the fragments are reshuffled. At each instant the entropy of the distribution is the same one associated with the random motion. The overall evolution is still Hamiltonian, but the trajectories predicted by the framework are now just the average motion of the fragments. 

If we assumed that any classical electron were affected by the same source of randomness, all fragments would exhibit a random walk of the same magnitude as defined by the information entropy. The spread of all distributions for the material would satisfy the same inequality. This gives us the classical uncertainty principle: $\sigma_q\sigma_p \geq \exp (I_\varnothing - 1) / 2 \pi $ where $I_\varnothing$ is the entropy associated with the source.

Such an assumption is not far fetched. The source of entropy could be identified in the environment (i.e. a minimum interaction that cannot be avoided as perfect isolation cannot be achieved) or the internal self interaction of the fragments (i.e. the evolution of the unstated part). Whatever it may be, if it is characterized by an entropy comparable to the one of the distribution (i.e. if we are trying to localize the electron better than the source can allow) the classical assumption breaks down as the random motion of each fragment is not negligible.

The generalization to multiple d.o.f. is outside of the scope of this work. Note, though, that this result is related to the Gromov non-squeezing theorem and the symplectic camel, which are already generalized to multiple d.o.f. It is the link with information theory that gives us further physical insight and is of particular interest given the strong connection between thermodynamics and Shannon's entropy.

This discussion shows, once again, how we can already find in classical mechanics some of the ingredients that are typically assumed to be characteristics of quantum mechanics. It has its own flavor of anti-particles, $CPT$ theorem and uncertainty principle. These should really be considered general features of deterministic and reversible evolution of distributions (i.e. Hamiltonian mechanics). Quantum mechanics just inherits these traits.

The discussion also gives us an intuitive physical picture that we will develop further in the next section: a quantum particle is a well determined distribution of not well determined fragments. The distribution as a whole can be studied and evolves deterministically while the fragments do not.

\section{Irreducible systems}
\label{sec:irreducibility}

In this section we explore a different relationship between the whole and its parts. We will consider an irreducible material, one where the state and the evolution of the whole does not tell us anything about the state and the evolution of the parts. The material is still conceptually infinitesimally decomposable, made of infinitesimal parts which we call fragments. But each fragment cannot be studied independently and therefore cannot be assigned a state: the internal dynamics is undetermined. We'll find that the evolution of an irreducible material is suitably described by the standard framework of quantum mechanics.

\subsection{Quantized material}

In section \ref{sec:reducibility} we introduced the idea of a classical material, one that is homogeneous and infinitesimally reducible, which represented a way to relate the state of the whole system to the state of its parts. In this section we explore an alternative case.

We will call a \emph{quantized material} a homogeneous material that is finitely reducible. That is, as we decompose the system we can only assign a state to finite amounts of material, which we call \emph{quantum particles}. As the material is homogeneous, all its particles carry the same amount of material and are described by the same state space. Each quantum particle is \textit{irreducible}: its state does not tell us anything about the state of its parts. While this idea should be familiar enough, it is best to iron out some details to avoid possible misconceptions.

The first possible source of confusion is between \textit{divisibility}, the idea that an object can be divided into parts, and \textit{reducibility}, the idea that the state of the whole is equivalent to the state of the parts. More precisely, a system is divisible if there is a time evolution $\mathcal{T}_{\Delta t}: \mathcal{S} \rightarrow \mathcal{S}_1 \times \mathcal{S}_2$ that takes the state of a system and returns the state of two separate isolated systems. It is reducible if $\mathcal{S} = \mathcal{S}_1 \times \mathcal{S}_2$, the state of the system is the state of the two subsystems. Most of the time these two concepts overlap, but they are separate as we can see in the following examples.

Consider a planarian flatworm. If we cut it into multiple pieces, each part will grow into a new worm. Yet, the state of each worm is not the state of multiple worms put together. The planarian is divisible into multiple worms but it is not reducible to multiple worms. Consider a magnet. Its state is described, to a first approximation, by the position of the north and south poles. Yet, we cannot divide the magnet into a north and south pole. The magnet is reducible to the poles but is not divisible into the poles.

Consider a muon. It will decay into an electron and two neutrinos. But the state of a muon is not the state of an electron and two neutrinos. The muon is divisible into the three particles but it is not reducible to them. Consider a proton. Its state is given by the state of its three valence quarks and the gluon field. Yet, we cannot divide the proton into the three isolated quarks. The proton is reducible to quarks, but not divisible into quarks.

The fundamental property of a quantum particle, as introduced here, is not indivisibility: it's irreducibility. It is the state that cannot be further partitioned, not the material itself.

Even more subtle is the difference between divisibility and \textit{decomposability}, the idea that the object has parts. That is, a system is decomposable if there is a law $+ : \mathcal{S} \times \mathcal{S} \rightarrow \mathcal{S}$ to combine two states to get a third. While irreducibility says nothing about whether a quantum particle can be divided, it implicitly assumes that it can be decomposed. Irreducibility describes the relationship between the state of the whole and the state of the parts, therefore the existence of parts is implicit in the definition. But it also tells us that we are limited in what description we can give to these parts: we cannot isolate them as we want and act on them independently because if we were able to do that we could assign them an independent state.

Therefore, while we must be able to talk about parts of a quantum particle, not all questions will make sense. Can we ask what fraction of a quantum particle is in a particular region of space? Yes, since the particle can be spatially decomposed some amount of material may be inside the volume. Can we ask where the fragments in that same fraction will be at a later time? No, since the particle is irreducible we cannot describe the evolution of a part.

Note that if the material is infinitesimally reducible (i.e. the classical case), reducibility/decomposability/divisibility are all possible. This is why one may not dwell on their difference and use the word ``part" interchangeably in the three cases. In the quantum context, instead, we assume irreducibility, so it is better to use more precise language to avoid confusion. Therefore we will call \emph{fragments} the infinitesimal parts of the material that cannot be described independently. We will call \emph{components} the parts that the material can be decomposed into. A component will have a state associated with it, and therefore will always be a collection of fragments.

The proper way to think of a quantized material is therefore the following: a continuum of infinitesimal fragments clumped together to form quantum particles, the smallest amount of material to which we can assign a state. With this picture in mind, we can see, even before developing the math, how some features typical of quantum behaviors emerge.

The first feature is the importance of statistics. Consider a fragment of a quantum particle in an initial state. As we cannot track its motion, it can correspond to any fragment of the particle in its final state. That is, the motion of a fragment of a particle is a random walk within the evolution of a particle.

The second feature is non-locality. Suppose we have a quantum particle distributed in space. Suppose we exert a force only on a subregion of its distribution. As the particle is irreducible, we cannot interact with only part of it: we must also affect the part where no force is exerted. We cannot simply interact with the part at a specific region, as this would allow us to assign a state to that part.

Wave-particle duality is also accounted for. The quantum particle behaves like a wave in the sense that its parts form a distribution in space. Yet, it behaves as a single unit, as the state is irreducible and cannot be split.

The most striking feature, though, is that we can predict the functional form for particle decay over time. Suppose a quantum particle spontaneously decays into two or more particles. Such a process cannot be deterministic as the number of d.o.f. increases. As the decay is spontaneous, it is not triggered by the environment (i.e. external state). It must be triggered by the internal dynamics of the fragments. As we assumed the internal dynamics is unknown, we cannot learn anything about it by looking at the product of the decay. In particular, all states have to have equal chance of decay or we could tell that the internal dynamics is more stable for some states. The probability of decay is constant for all states at all times. The probability of survival follows the exponential distribution.

As we see qualitatively how the features typical of quantum mechanics naturally emerge from this picture, we are ready to make this more quantitatively and mathematically precise.

\subsection{Irreducible material and complex vector spaces}

It is time to turn these concepts into a precise mathematical formulation as we did for classical mechanics. We want to formalize the state space for an \emph{irreducible material} as a quantum particle will be a unit amount of such material.

As for the classical material, the irreducible material is thought to be decomposable: each amount can be seen as the sum of different components (e.g. spatially distributed). Therefore \ref{prop:abelian_group} applies and the state space $\Psi$ of an irreducible material is an abelian group.

Yet, the description of each component cannot be more accurate than the description of the whole. To capture this idea precisely we need to characterize the aspects of the unstated part that influence the state definition. How do we do that? We can look at the classical derivation for guidance. In that case, we looked at the set of transformations that increased/decreased the amount of material. The idea is that we can think of these transformations as characterizing what the state can and cannot say about the infinitesimal parts.

Suppose we have a state $\mathcal{c}_1$ for an amount $a_m$ of classical material. Suppose we assigned a numeric identifier to each particle of the material. This identifier would be just a dummy label we assign, not a physical quantity. As the classical material is homogeneous, permutations between classical particles within the same state have no effect. Therefore the ordering given by the identifier does not matter. Identifying a classical particle within a composite state is equivalent to picking a number from a random variable $A$ uniformly distributed over an interval of size $a_m$.

Suppose we have another state $\mathcal{c}_2 = \tau(a) \mathcal{c}_1$ where the only difference is an increase/decrease of the amount of material by a factor $a$. The random variable $B$ that identifies a particle of $\mathcal{c}_2$ is now a uniform distribution $B=aA$ where the range is the new amount of material: there are now proportionally more/less labels to assign. This is also the only possible transformation that can link the two uniform random variables $A$ and $B$. Therefore increase/decrease of material is the only operation that is captured by the state among those that affect the labeling of the fragments. In other words: the classical state does not care about the permutation of the fragments, but it cares how many fragments are there.

For the irreducible material we follow the same strategy and study the transformations that affect the labeling of what the irreducible state cannot tell us: the configuration of each fragment. Similarly, the irreducible state will not care about permutations of fragment configurations, but it will care how many fragment configurations are there.

Suppose we have a state $\psi_1$ for an amount $a_m$ of irreducible material characterized by one d.o.f. The configurations of the infinitesimal fragments will be identified by two values, as a d.o.f. is two-dimensional. Each fragment configuration is identified by a pair of independent random variables $(A, B)$, which are, as before, just dummy variables we use to name the fragment configurations. As the material is homogeneous, the distribution is uniform or we could tell something about the fragment configurations. We choose A and B to be uniform distributions of interval $\sqrt{a_m}$: identifying a fragment configuration is like picking a point at random over a square of area $a_m$, the amount of material.

Suppose we have another state $\psi_2$, where the labeling for the fragment configurations has changed. We can label the fragment configurations of $\psi_2$ in the same way, with two independent uniform distributions $(C, D)$. But since the labeling is the only thing that has changed, we must be able to relate the new distribution over $(C,D)$ to the old distribution over $(A,B)$. The relationship between the two pairs of variables must be linear since all distributions are uniform. $C=aA+bB$. To preserve the square distribution we must then have $D= -bA + aB$. We can write $\psi_2 = \tau(a,b) \psi_1$ as $a$ and $b$ are the only parameters of the transformation. Or more compactly, $\psi_2 = \tau(c) \psi_1$ where $c=a+\imath b$.

The transformation $\tau(c)$ increases/decreases the area by the modulus $|c|^2=a^2+b^2$, which physically corresponds to the increase/decrease of amount of material. It also rotates the square by the phase $\theta_c \equiv arg(c)$, which corresponds to a change in correlation: if $\theta_c = 0$, $C = |c| A$ are clearly correlated while if $\theta_c = \pi$, $C = - |c| A$ are clearly anti-correlated.

This line of reasoning is fairly abstract and it may seem to have no bearing with quantum states. Therefore, jumping a bit ahead, let's see if we can already find quantitative consequences of this model and see whether they are consistent with our knowledge of quantum mechanics.

If the amount of material is connected to the area of two uniform distributions, it should combine like the variance of a random variable. And in fact it does. Suppose $X$ and $Y$ are two random variables,
\begin{equation}\label{eqn:correlation}
\sigma^2_{X+Y} = \sigma^2_{X} + \sigma^2_{Y} + 2 \, \sigma_{X} \sigma_{Y} \rho_{X,Y}
\end{equation}
where $\sigma^2$ is the variance, $\sigma$ the standard deviation, $\rho_{X,Y}= COV(X,Y)/\sigma_{X} \sigma_{Y}$ is the Pearson correlation coefficient~\cite{Grimmet}. As $-1\leq\rho_{X,Y}\leq+1$, the total variance can be greater or smaller than the sum of the variance of the components. Now suppose $\psi$ is the state for a unit amount of irreducible material (i.e. a quantum particle). We can write $c \psi = (c_1 + c_2) \psi = c_1 \psi + c_2 \psi$. We have:
\begin{equation}\label{eqn:correlation_state}
|c_1+c_2|^2=|c_1|^2 + |c_2|^2 + 2 |c_1||c_2|\cos(\theta_{c_2} - \theta_{c_1})
\end{equation}
The two expressions are formally identical and it shouldn't come as a surprise. The more material we have, the more labels for the fragment configurations to choose from, the bigger the spread of the random variables that assign the labels. $-1\leq\cos(\theta_{c_2} - \theta_{c_1})\leq+1$ takes the place of the Pearson correlation coefficient. Quantum interference can then be understood in terms of the combination of correlated/anti-correlated distributions.

If the phase determines the correlation between two states, then it should only be meaningful when comparing two states. And in fact only differences in phases are physically distinguishable in quantum mechanics. A more meaningful name for the phase, then, would be \emph{correlation angle}.

If quantum states are uniform distributions over fragment configurations, we should see a hint of that when computing distributions over observables. And in fact we do. Suppose we have $Q\psi = q \psi$, an eigenstate of some observable. Then it is also an eigenstate of the infinitesimal transformation of the conjugate variable $1+\frac{Qdp}{\imath\hbar}$. Since it has a symmetry over $p$, the distribution has to be equal at each value of $p$: the distribution is uniform.

If quantum particle states are really just unit amounts of material uniformly distributed, we should be able to calculate the Shannon entropy associated with them. Using our dummy variables, $I(\psi) = - \int \rho \ln (\rho) dA dB = - \int 1 ln (1) dA dB = 0$. For a quantum mixed state $\tilde{\rho} = \sum \tilde{\rho}_i |\psi_i\rangle \langle \psi_i |$, we have $I(\tilde{\rho}) = - \sum \tilde{\rho}_i \ln \tilde{\rho}_i - \tilde{\rho}_i I(\psi_i) = - \sum \tilde{\rho}_i \ln \tilde{\rho}_i = -tr(\tilde{\rho} \ln \tilde{\rho})$. And in fact this is the Von Neumann entropy~\cite{Von Neumann} used in quantum information, that implicitly assumes all states provide the same amount of entropy (i.e. zero).

If two quantum particle states $\psi_1$ and $\psi_2$ are really just two states at equilibrium with the same entropy, we should always be able to find a deterministic and reversible evolution to connect the two. And in fact we can. Take the plane identified by $\psi_1$ and $\psi_2$. Find the infinitesimal rotation $1 + \frac{Od\alpha}{\imath\hbar}$ along that plane. $\psi_2 = e^{\frac{O\alpha}{\imath\hbar}} \psi_1$ is the unitary evolution of the other. Note that this is \emph{not} always possible in classical mechanics, as two states $\mathcal{c}_1$ and $\mathcal{c}_2$ may have different entropy.

In other words, our simple model is not only consistent with quantum mechanics but it has predictive power: it provides intuition that is well supported by these findings. It is not ``just a mathematical trick."

If we increase the number of independent d.o.f. the group of transformations remains the same. While we have more fragment configurations to label, each d.o.f. is labeled independently. Therefore we have $n$ uniform distributions, each over a square of the same amount of material. This means that, as we combine components, the labels have to combine in the same way along the different d.o.f.: the increase/decrease of material and the change in correlation has to be the same across all d.o.f.

Mathematically, the state space $\Psi$ for an irreducible material is a complex vector space. It is so because it is the state space for a decomposable material (i.e. an abelian group) where we can increase/decrease the amount and change the correlation of the internal description (i.e. a set of transformations conveniently parametrized by a complex number).

\begin{prop}\label{prop:complex_vector_space}
	The state space $\Psi$ for an irreducible material is a vector space over $\mathbb{C}$.
\end{prop}

\begin{justification}
	We claim $\Psi$ is an abelian group. $\Psi$ is the state space for a decomposable homogeneous material and is therefore an abelian group  by \ref{prop:abelian_group}.
	
	We claim the set of physically distinguishable transformations $T$ that relabel the fragment configurations is isomorphic to $\mathbb{C}$. Consider an amount $a_m \in \mathbb{R}$ of irreducible material characterized by one d.o.f. The configuration of a fragment of the material is characterized by two random variables $A, B: \Omega \rightarrow \mathbb{R}$ where $\Omega$ is the sample space for the configuration of the fragments. As the material is homogeneous, $A$ and $B$ can be chosen such that the distribution $\rho$ of material over such variables is uniform over a range of $\sqrt{a_m}$ for both variables. We have $a_m=\int \rho dA \wedge dB = \int dA \wedge dB$. Consider a second amount $\hat{a}_m$ of irreducible material where we changed the labeling of the fragment configurations. Its fragment configurations can also be characterized by a uniform distribution over two variables $C, D: \Omega \rightarrow \mathbb{R}$ over a range of $\sqrt{\hat{a}_m}$. As the two amounts of material differ just by the labeling of fragment configuration, $(C,D)$ is a function of $(A,B)$. We have $\hat{a}_m = \int dC \wedge dD = \int |J| dA \wedge dB$ where $|J|$ is the Jacobian determinant. As both distributions are uniform, $|J|$ has to be constant. The transformation has to be linear: $C=aA + bB$ and $D=-bA + aB$. As the change in amount of material, the only physically distinguishable quantity associated with the fragment configurations, is fully specified by $a$ and $b$, no other transformation is physically relevant for the state. Let $\tau: \mathbb{C} \rightarrow T$ represent such a transformation parametrized by $c=a+\imath b$. $|c|^2 = |J| = (a^2 + b^2) = \hat{a}_m / a_m$ is the ratio of the two amounts of materials and is physically distinguishable. Define on $T$ an addition $+: T \times T \rightarrow T$ and a multiplication $*: T \times T \rightarrow T$ such that $\tau(z_1) + \tau(z_2) = \tau(z_1+z_2)$ and $\tau(z_1) * \tau(z_2) = \tau(z_1*z_2)$, $z_1,z_2 \in \mathbb{C}$, the sum and product of the transformation is equal to the sum and product of their respective factors. The phase $arg(c)$ under addition is physically distinguishable as it can affect the modulus of the result. $\tau$ is an isomorphism between $T$ and $\mathbb{C}$ as fields. Assume now that the material is characterized by multiple independent d.o.f. The transformation on each d.o.f. is defined independently by $c_i \in \mathbb{C}$. As $\hat{a}_m / a_m$ is the same under all d.o.f. $|c_i|=|c_j|$. The relationship has to always hold under addition as well, therefore $arg(c_i) = arg(c_j)$. $T\cong \mathbb{C}$ also for multiple d.o.f.

	We claim $\Psi$ is a vector space over $\mathbb{C}$. The abelian group $\Psi$ can be extended with the operations defined by $T$ as was done in \ref{prop:real_vector_space} for the infinitesimally reducible case.	
\end{justification}

\subsection{Complex valued distributions}

In the previous sections we saw how a finite amount of classical material was described by a distribution over the state variables of infinitesimal parts. For an irreducible material, we cannot assign states to the fragments, therefore we cannot properly talk about state variables. Therefore we call \emph{fragment variables} the quantities that we use to label the fragment configurations, and it will be over these variables that we will define the distribution of irreducible material.

Naturally, we cannot give a complete description for the configuration of a fragment, therefore we cannot give a joint distribution for all fragment variables at the same time.\footnote{This recovers the notion of complementarity introduced by Bohr~\cite{Bohr}.} We can, though, find a maximal set of fragment variables: one that can be used to provide the best description possible. While there isn't a unique maximal set, we can show that one such set is provided by the fragment variables $q^i : \mathcal{Q} \rightarrow \mathbb{R}$ that define the units required to describe the system.

In fact, we cannot add a new d.o.f. (as this would mean changing $\mathcal{Q}$) and we cannot add a new fragment variable in an existing d.o.f. (as this would fully specify said d.o.f.). Therefore the set of fragment variables $q^i$ already provides a maximal set of fragment variables. This means that a state $\psi$ must tell us for each point in $\mathcal{Q}$ the amount and correlation angle of the material. That is: the state will be identified by $\psi : \mathcal{Q} \rightarrow \mathbb{C}$, the wave function over $\mathcal{Q}$.

The square modulus at each point will give us the density of the material. As in the classical case, such density must be integrable to give us the total amount. Therefore the function $\psi(q) \in L^2(\mathcal{Q}, \mu)$ is square integrable. The norm $|\psi|^2 = \int_{\mathcal{Q}} |\psi(q)|^2 dq$ will give us the amount of material associated with the state.

As the material is homogeneous, all its quantum particles are made of the same amount. We can set, by convention, the amount of material for a particle to be unitary, therefore quantum particle states are the subset of $\Psi$ such that $|\psi|^2 = 1$. As we have seen, changing the phase by a constant only affects the unstated part. Therefore normalized distributions that differ by a total phase are not physically distinguishable and represent the same quantum particle state.

\begin{prop}\label{prop:wavefuntion}
	The state space $\Psi$ for an irreducible material is isomorphic to a subspace of the space of Lebesgue square integrable continuous functions. That is $\Psi \cong \Psi(\mathcal{Q}, dq^i) \subseteq C(\mathcal{Q}, \mathbb{C}) \cap L^2(\mathcal{Q}, \mu, \mathbb{C})$ under the isomorphism $\Upsilon_q : \Psi \leftrightarrow \Psi(\mathcal{Q}, dq^i)$ where $\Psi(\mathcal{Q}, dq^i)$ is the space of wave functions over $q^i$. The state space for a quantum particle is isomorphic to the projective space $\mathsf{P}(\Psi)$.
\end{prop}

\begin{justification}
	We claim the manifold $\mathcal{Q}$ that defines the units to describe a fragment of an irreducible material also provides a maximal set of fragment variables: no other fragment variable can be added. Let $\mathcal{Q}$ be the manifold that defines the units. Then for each state $\psi \in \Psi$ there exists a distribution $\rho_\psi(q^i)$ that represents the quantity of material at each point. Suppose this distribution could be further refined. Then another independent variable $k$ would exist such that for each state $\psi \in \Psi$ we would have a distribution $\rho_\psi(q^i, k)$. $k$ cannot be part of a new d.o.f. as it would use a unit defined independently from $\mathcal{Q}$. This is a contradiction as $\mathcal{Q}$ defines all units. $k$ cannot lie in the same d.o.f. as of one of the $q^i$ or it would provide a complete description for the fragment. This is a contradiction as the material is irreducible. $\mathcal{Q}$ provides a maximal set of fragment variables.
	
	We claim $\Psi \cong \Psi(\mathcal{Q}, dp^i) \subseteq C(\mathcal{Q}, \mathbb{C}) \cap L^2(\mathcal{Q}, \mu, \mathbb{C})$ as a vector space. The maximal physical description of a fragment configuration is given by a point in $\mathcal{Q}$ parametrized by $q^i$ and by a transformation $\tau \in T$ parametrized by $c \in \mathbb{C}$. For each $\psi \in \Psi$ there must exist a wave function $\psi(q^i) : \mathcal{Q} \rightarrow \mathbb{C}$ that provides the maximal physical description for the state. $\rho_\psi(q^i) = |\psi(q^i)|^2$ as the square modulus of $\psi(q^i)$ is the increase/decrease from the unit amount of material. $\rho_\psi(q^i) \in C(\mathcal{Q}, \mathbb{R})$ by \ref{prop:continuity}. $\rho_\psi(q^i) \in C^1(\mathcal{Q}, \mathbb{R})$ by the same arguments in the justification for \ref{prop:differentiable_manifold}. Therefore $\psi(q^i) \in C(\mathcal{Q}, \mathbb{C})$. $\rho_\psi(q^i) \in L^1(\mathcal{Q}, \mu, \mathbb{R})$ by the same arguments in the justification for \ref{prop:integration}. Therefore $\psi(q^i) \in L^2(\mathcal{Q}, \mu, \mathbb{C})$. Let $\Upsilon_{q^i} : \Psi \rightarrow C(\mathcal{Q}, \mathbb{C}) \cap L^2(\mathcal{Q}, \mu, \mathbb{C})$ be the function that given a state returns its wave function. As the wave function captures the maximal description of the system, two different states must represent different distributions. $\forall \psi_1, \psi_2 \in \Psi, \psi_1 \neq \psi_2 \implies \Upsilon_{q^i}(\psi_1) \neq \Upsilon_{q^i}(\psi_2)$. $\Upsilon_{q^i}$ is injective. $\Upsilon_{q^i}$ is a bijection between $\Psi$ and $\Psi(\mathcal{Q}, dq^i)\equiv\Upsilon_{q^i}(\Psi)$. Let $\psi=\psi_1+\psi_2$, then $\Upsilon_{q^i}(\psi)=\Upsilon_{q^i}(\psi_1)+\Upsilon_{q^i}(\psi_2)$ as combining the distributions means combining the statistical description of the fragment configurations. Let $\psi_1=\tau(c)\psi_2$, then $\Upsilon_{q^i}(\psi_1)=c \Upsilon_{q^i}(\psi_2)$ as $\tau(c)$ transforms the fragment configuration description by the factor $c$. $\Upsilon_{q^i}$ is a homomorphism between vector spaces. $\Psi$ is isomorphic to $\Psi(\mathcal{Q}, dq^i) \subseteq C(\mathcal{Q}, \mathbb{C}) \cap L^2(\mathcal{Q}, \mu, \mathbb{C})$.
	
	We claim the state space for a quantum particle is isomorphic to the projective space $\mathsf{P}(\Psi)$. Let $\psi \in \Psi$ be the state of a quantum particle. $\tau(e^{\imath \theta}) \psi$ with $\theta \in \mathbb{R}$ is not physically distinguishable from $\psi$ as phases are only physically distinguishable through addition. $\tau(\sqrt{\rho}) \psi$ with $\sqrt{\rho} \in \mathbb{R}^+ \backslash \{1\}$ cannot represent a quantum particle as it has a different amount of material. There exists only one distinguishable quantum particle state in the equivalence class $\psi \sim \tau(c) \psi$ with $c \in \mathbb{C}$. The state with no amount of material does not correspond to any particle state. The quotient space $\mathsf{P}(\Psi) \equiv (\Psi \backslash \{0\}) / \sim$ is isomorphic to the state space for quantum particles.
\end{justification}

\subsection{Inner product and operators}

In both the reducible and irreducible case, the material is distributed over $\mathcal{Q}$. The $q^i$ do not label the states of a finite amount of material, but the configurations for the infinitesimal parts. For a classical material, the infinitesimal parts can be given a state, have their own topological space and a set of local state variables $(q^i, p_i)$ to identify them. We can use the richness of differential geometry. For an irreducible material, the infinitesimal parts cannot be given a state, cannot be fully identified by a set of variables. We have to use other mathematical tools, typically vector spaces. As we will often use linear combinations of states, we will write, for brevity, $c\psi$ instead of $\tau(c)\psi$. That is: each complex number still represents a state transformation.

For the tool set to be on par with classical mechanics, such that we can write the equations of motion, we need to address three requirements. We need to be able to compare two states and ask what fraction of them share the same configuration: this will give us an inner product. We need to be able to express distributions and expectations for different quantities: this will give us self-adjoint linear operators. We then need to identify the self-adjoint linear operators for the quantities $k_i$ and $p_i$ and to be able to express the distribution of material in terms of those. Let's start with the first.

The idea is that, given two states $\psi_1$ and $\psi_2$, we should be able to tell how much of $\psi_1$ is prepared in the same way as $\psi_2$. That is $\psi_1 = \alpha \psi_2 + \psi_3$ is the combination of the component that is prepared like $\psi_2$ (e.g. same distributions in position and momentum) and the component $\psi_3$ that is prepared completely differently. Let $\mathsf{P}_{\psi_2} : \Psi \rightarrow \Psi$ be the operator that gives us the component $\mathsf{P}_{\psi_2} (\psi_1) \equiv \alpha \psi_2$ prepared like $\psi_2$. This operator is linear in $\psi_1$: the component of the combined system that is prepared like $\psi_2$ is the sum of the subcomponents that are prepared like $\psi_2$. It also gives us the same result if applied once or twice: once we find the component that is prepared like $\psi_2$, it remains the same. $\mathsf{P}_{\psi_2}$ is therefore a projection: a linear operator such that $\mathsf{P}_{\psi_2}(\mathsf{P}_{\psi_2} (\psi_1))= \mathsf{P}_{\psi_2} (\psi_1)$.

If $\psi_1$ is already prepared like $\psi_2$, then its projection on $\psi_2$ will be equal to $\psi_1$. Conversely, if the projection of $\psi_1$ is equal to itself, then it means that the state is already prepared like $\psi_2$. That is: $\mathsf{P}_{\psi_2} (\psi_1) = \psi_1 \Leftrightarrow \psi_1 = \alpha \psi_2$. Because of this, we can define a non-degenerate complex product $\langle \cdot |\cdot \rangle : \Psi \times \Psi \rightarrow \mathbb{C}$ such that $\mathsf{P}_{\psi_2} (\psi_1) \equiv \frac{\langle \psi_2|\psi_1 \rangle}{|\psi_2|^2} \psi_2$.

The projections have other useful properties. Because the material is irreducible, if some fraction of $\psi_1$ is prepared like $\psi_2$, then the same fraction of $\psi_2$ is prepared like $\psi_1$. If not, studying the fragments of one state would give us a better (or worse) description for the fragments of the other state. Therefore we have $\frac{|\mathsf{P}_{\psi_2} (\psi_1)|^2}{|\psi_1|^2}=\frac{|\mathsf{P}_{\psi_1} (\psi_2)|^2}{|\psi_2|^2}$. This tells us that $|\langle \psi_2|\psi_1 \rangle|^2 = |\langle \psi_1|\psi_2 \rangle|^2$.

Lastly, if we take the projection of $\psi_1$ along $\psi_2$ and we project it again on $\psi_1$, we end up with a component of $\psi_1$. As this is a part of $\psi_1$ prepared as $\psi_1$, its state is just $\psi_1$ rescaled, with no change in correlation. That is: $\mathsf{P}_{\psi_1} (\mathsf{P}_{\psi_2} (\psi_1)) = a \psi_1$ with $a \in \mathbb{R}^+$. This tells us that $\langle \psi_2|\psi_1 \rangle = \langle \psi_1|\psi_2 \rangle^\dagger$.

Because of the properties of the norm and the projections, the product is positive definite, linear and conjugate symmetric. In other words, the state space $\Psi$ for an irreducible material is an inner product space derived by the norm and projections.
 
\begin{prop}\label{prop:inner_product}
	The state space $\Psi$ for an irreducible material is an inner product with $\langle \cdot |\cdot \rangle : \Psi \times \Psi \rightarrow \mathbb{C}$ such that $\langle \psi_1 | \psi_2 \rangle = \int_\mathcal{Q} \psi_1^\dagger (q^i) \psi_2(q^i) d^nq$.
\end{prop}
\begin{justification}
	We claim there exists a map $| \cdot | ^2 : \Psi \rightarrow \mathbb{R}$ such that $|\psi|^2 \ge 0 \; \forall \psi \in \Psi$ and $|\psi|^2 = 0 \Leftrightarrow \psi = 0$. Let $| \cdot | ^2 : \Psi \rightarrow \mathbb{R}$ be the map that returns the amount of material for a given state. Let $\psi \in \Psi$ be a state. $|\psi|^2 = \int \rho(q^i) d^nq = \int \psi^\dagger(q^i) \psi(q^i) d^nq \ge 0$. $|\psi|^2 = 0 \Leftrightarrow \psi = 0$ as the only state with no material is the empty state.
	
	We claim $\forall \psi \in \Psi$ there exists a projection $\mathsf{P}_\psi : \Psi \rightarrow \Psi$ with $\psi \in \Psi$. Let $\psi \in \Psi$ be a state. Let $\mathsf{P}_\psi : \Psi \rightarrow \Psi$ be an operator that returns the component of a given state that is prepared as $\psi$. $\mathsf{P}_{\psi}(c_1 \phi_1 + c_2 \phi_2) = c_1 \mathsf{P}_{\psi}(\phi_1) + c_2 \mathsf{P}_{\psi}(\phi_2)$ is linear as the component of the composition prepared like $\psi$ is the composition of the components so prepared. $\mathsf{P}_\psi \circ \mathsf{P}_\psi = \mathsf{P}_\psi$ is a projection because the result of $\mathsf{P}_\psi$  is prepared as $\psi$. 
	
	We claim $\mathsf{P}_{\psi} (\phi) = \phi \Leftrightarrow \phi = c \psi$ for some $c \in \mathbb{C}$. Let $\psi, \phi \in \Psi$ such that $\mathsf{P}_{\psi} (\phi) = \phi$. The whole state $\phi$ is prepared as $\psi$. We can obtain $\phi$ from $\psi$ by relabeling the fragment configurations. $\phi = c \psi$ for some $c \in \mathbb{C}$. Let $\psi, \phi \in \Psi$ such that $\phi = c \psi$ for some $c \in \mathbb{C}$. Then all of $\phi$ is prepared as $\psi$. $\mathsf{P}_{\psi} (\phi) = \phi$. Therefore $\mathsf{P}_{\psi} (\phi) = \phi \Leftrightarrow \phi = c \psi$.
	
	We claim $|\psi|^2|\mathsf{P}_{\psi} (\phi)|^2=|\phi|^2|\mathsf{P}_{\phi} (\psi)|^2$. Let $\psi, \phi \in \Psi$. As the material is irreducible, $\psi$ and $\phi$ cannot be used to describe smaller parts. The fraction of $\phi$ that is prepared as $\psi$ must be equal to the fraction of $\psi$ that is prepared as $\phi$. $\frac{|\mathsf{P}_{\psi} (\phi)|^2}{|\phi|^2}=\frac{|\mathsf{P}_{\phi} (\psi)|^2}{|\psi|^2}$. Multiplying by $|\psi|^2|\phi|^2$ we have $|\psi|^2|\mathsf{P}_{\psi} (\phi)|^2=|\phi|^2|\mathsf{P}_{\phi} (\psi)|^2$. This holds for the empty state as well. Let $\phi$ be the empty state. $|\phi|^2=0$ because it has no material and $|\mathsf{P}_{\psi} (\phi)|^2 = 0$ as no component of the empty state is prepared like $\psi$.
	
	We claim $\mathsf{P}_{\psi} (\mathsf{P}_{\phi} (\psi)) = r \psi$ for some $r \in \mathbb{R}^+$. Let $\psi, \phi \in \Psi$. Consider $\mathsf{P}_{\psi} (\mathsf{P}_{\phi} (\psi))$: we first take the component of $\psi$ prepared like $\phi$ and then take the component of that prepared like $\psi$. The overall operation is simply returning a part of $\psi$. This must correspond to the same exact configuration for the fragments. The only thing that can change is the amount of material. Therefore $\mathsf{P}_{\psi} (\mathsf{P}_{\phi} (\psi)) = r \psi$ for some $r \in \mathbb{R}^+$.
	
	We claim that the state space $\Psi$ for an irreducible material has an inner product $\langle \cdot |\cdot \rangle : \Psi \times \Psi \rightarrow \mathbb{C}$. $| \cdot | ^2$ and $\mathsf{P}_{\phi}$ satisfy the requirements in \ref{thrm:inner_product}. $\Psi$ is an inner product space with $\langle \cdot |\cdot \rangle : \Psi \times \Psi \rightarrow \mathbb{C}$ such that $\langle \psi |\phi \rangle \psi = | \psi | ^2\mathsf{P}_{\psi}(\phi)$.
	
	We claim that the inner product can be expressed in terms of the wave functions as $\langle \psi_1 | \psi_2 \rangle = \int_\mathcal{Q} \psi_1^\dagger (q^i) \psi_2(q^i) d^nq$. The amount of material associated with a state $\psi \in \Psi$ is $\langle \psi | \psi \rangle = | \psi |^2 = \int_\mathcal{Q} \psi^\dagger (q^i) \psi(q^i) d^nq$. The inner product associated with the $L^2$ norm is $\int_\mathcal{Q} \psi_1^\dagger (q^i) \psi_2(q^i) d^nq$ with $\psi_1, \psi_2 \in \Psi$. $\langle \psi_1 | \psi_2 \rangle = \int_\mathcal{Q} \psi_1^\dagger (q^i) \psi_2(q^i) d^nq$.
\end{justification}

Note that the inner product only has an indirect physical definition, based on the norm and the projections. These latter two concepts therefore provide more physical intuition since the first does not have a straightforward physical meaning.

Also note that the observations made after \ref{prop:differentiable_manifold} for the  classical case are still valid. $\Psi$ cannot be a complete metric space for the norm induced by the inner product as it does not include discontinuous functions. Therefore it is only a pre-Hilbert space. For mathematical convenience, one can take its completion with respect to the norm induced by the inner product and obtain a Hilbert space. Yet, such an object no longer represents a physically meaningful set of states.

Therefore we need to make sure that the rest of our definitions and justifications do not require a complete metric space to be valid or they would not be physically meaningful. In particular, we don't want to use the usual link $\psi(q_0) = \langle q_0 | \psi \rangle$ between the wave function and the inner product as the material cannot be prepared at a specific value $q_0$ when $q$ is a continuous variable. Similarly, we cannot show that to each state variable is associated a self-adjoint operator $Q = \sum q |q\rangle \langle q |$ simply by expressing it in terms of eigenstates as they may not represent physical states.

We proceed in the following way. Suppose we have a fragment variable $f : \mathcal{Q} \rightarrow \mathbb{R}$. Consider the product $f (q) \psi (q)$: it returns the wave function weighted by the value of $f$. We can also think of $f$ as an operator that takes a complex function of $\mathcal{Q}$ and returns another complex function of $\mathcal{Q}$ that represents the distribution of the expectation over the fragments. The operator is linear since $f (q) (a\psi_1(q) + b\psi_2(q)) = a f(q) \psi_1(q) + b f(q)  \psi_2(q)$. The operator is self-adjoint since $f^\dagger = f$. Consider $\int_\mathcal{Q} \psi^\dagger (q) f(q) \psi(q) dq = \int f(q) | \psi(q)|^2 dq$. It integrates the value of $f$ weighted by the amount of the material: it is the expectation of the fragment variable $f$ over the distribution. 

While $\mathcal{Q}$ is enough to define the system of units, it is not enough to define all the fragment variables over which we can express the distributions of irreducible material. The conjugate variables $p_i=\hbar k_i$ cannot be expressed as a function of $q^i$, which reiterates the idea that we cannot specify all variables for each fragment. Suppose there is another maximal set of fragment variables. Then those variables will form a topological space $\hat{\mathcal{Q}}$ such that each $\psi$ can be expressed as a wave function $\psi(\hat{q}) :\hat{\mathcal{Q}} \rightarrow \mathbb{C}$.

As each state $\psi$ can be identified by a wave function on either space, there must be an operator $\Upsilon^q_{\hat{q}}(\psi(q)) = \psi(\hat{q})$ that converts the wave function on one space to the other.\footnote{Note that a neighborhood of $q$ is not necessarily a neighborhood of $\hat{q}$. This relation holds only if $q$ and $\hat{q}$ can be expressed as a bijection of each other $\hat{q}^i(q^i)$, i.e.~if they are different coordinates of the same manifold.} As a linear combination of states corresponds to a linear combination of wave functions, $\Upsilon^q_{\hat{q}}$ is a linear operator. As converting the space again must give us the original wave function, $\Upsilon^{\hat{q}}_{q} \circ \Upsilon^q_{\hat{q}}= I$ is unitary. We can now define a linear operator $F_{[\hat{q}]} \equiv f(\hat{q})$ corresponding to fragment variable $f : \hat{\mathcal{Q}} \rightarrow \mathbb{R}$. Using $\Upsilon^{q}_{\hat{q}}$ we can express $F_{[q]} = \Upsilon^{\hat{q}}_{q} \circ F_{[\hat{q}]} \circ  \Upsilon^{q}_{\hat{q}}$ as an operator over complex functions of $\mathcal{Q}$. It will still be linear and self-adjoint.

That is: for each fragment variable $f$ we have a corresponding linear self-adjoint operator $F_{[q^i]}$ that acts on wave functions over $q^i$, which can be expressed as function multiplication in a suitable choice of $\hat{\mathcal{Q}}$. The (normalized) expectation of $f$ is given by $\frac{1}{|\psi|^2}\int \psi^\dagger(q^i) F_{[q^i]} (\psi(q^i)) d^nq$.

\begin{prop}\label{prop:self_adjoint_operators}
	For each fragment variable $f$ there exists a topological space $\hat{\mathcal{Q}}$ such that $f : \hat{\mathcal{Q}} \rightarrow \mathbb{R}$ and a vector space isomorphism $\Upsilon_{\hat{\mathcal{q}}^i} : \Psi \leftrightarrow \Psi(\hat{\mathcal{Q}}, d\hat{q}^i)$ that takes a state $\psi \in \Psi$ and returns a wave function $\psi(\hat{q}^i)$ from $\Psi(\hat{\mathcal{Q}}, d\hat{q}^i)$, the space of wave functions over $\hat{\mathcal{Q}}$. There also exists a corresponding self-adjoint linear operator $F_{[q^i]} : \Psi(\mathcal{Q}, dq^i) \rightarrow C(\mathcal{Q}, \mathbb{C})$ such that $\int \psi^\dagger(q^i) F_{[q^i]}( \psi(q^i))d^nq $ is the expectation of $f$ multiplied by the amount of material. 
\end{prop}
\begin{justification}
	We claim for each fragment variable $f$ there exists a topological space $\hat{\mathcal{Q}}$ such that  $\hat{q}^i : \hat{\mathcal{Q}} \rightarrow \mathbb{R}^n$ provides a maximal description and $f : \hat{\mathcal{Q}} \rightarrow \mathbb{R}$. Let $f$ be a fragment variable. Let $\hat{q}^i$ be a set of fragment variables such that $f=f(\hat{q}^i)$. If $\hat{q}^i$ does not provide a maximal set of fragment variables then there exist other fragment variables that refine the description. Extend $\hat{q}^i$ with such variables. If still not maximal, continue extending. As the number of possibilities charted by $\hat{q}^i$ cannot exceed the number of possibilities charted by $q^i$, at some point we'll reach a maximal set. We call $\hat{\mathcal{Q}}$ the topological space charted by the fragment variables $\hat{q}^i$.
	
	We claim there exists a vector space isomorphism $\Upsilon_{\hat{\mathcal{q}}^i} : \Psi \leftrightarrow \Psi(\hat{\mathcal{Q}}, d\hat{q}^i)$. Let $\psi \in \Psi$ be a state. $\hat{q}^i$ provides a maximal set of fragment variables. There must exist a wave function $\psi(\hat{q}^i) : \mathbb{R}^n \rightarrow \mathbb{C}$ associated with $\psi$ such that it gives the maximal description for the state as seen in \ref{prop:wavefuntion}. Let $\Upsilon_{\hat{\mathcal{q}}^i} : \Psi \rightarrow \Psi(\hat{\mathcal{Q}}, d\hat{q}^i)$ be the map that associates a state with its wave function over $\hat{\mathcal{Q}}$. $\Upsilon_{\hat{\mathcal{q}}^i}$ is a vector space isomorphism for the same reason discussed in \ref{prop:wavefuntion}
	
	We claim that for each fragment variable $f : \hat{\mathcal{Q}} \rightarrow \mathbb{R}$ there exists a corresponding self-adjoint linear operator $F_{[\hat{q}^i]} : \Psi(\hat{\mathcal{Q}}, d\hat{q}^i) \rightarrow C(\hat{\mathcal{Q}}, \mathbb{C})$ such that $\int \psi^\dagger(\hat{q}^i) F_{[\hat{q}^i]}( \psi(\hat{q}^i)) d^n \hat{q}$ is, if it exists, the expectation of $f$ multiplied by the amount of material. The expectation value of $f$ multiplied by the amount of material is given by $\int f(\hat{q}^i) \rho(\hat{q}^i) d^n\hat{q}$ where $\hat{q}^i$ provide a maximal set of fragment variables such that $f=f(\hat{q}^i)$. We have $\int f(\hat{q}^i) \rho(\hat{q}^i) d^n\hat{q} = \int \psi^\dagger(\hat{q}^i) f(\hat{q}^i) \psi(\hat{q}^i) d^{n}\hat{q}$. Let $F_{[\hat{q}^i]} : \Psi(\hat{\mathcal{Q}}, d\hat{q}^i) \rightarrow C(\hat{\mathcal{Q}}, \mathbb{C})$ such that $F_{[\hat{q}^i]} (\psi(\hat{q}^i)) = f(\hat{q}^i) \psi(\hat{q}^i)$. $F_{[\hat{q}^i]}$ is a linear operator. $F_{[\hat{q}^i]}$ is also self-adjoint since\journal{\break} $\int (F_{[\hat{q}^i]}(\psi(\hat{q}^i)))^\dagger \psi(\hat{q}^i) d^{n}\hat{q} = \int (f(\hat{q}^i) \psi(\hat{q}^i))^\dagger \psi(\hat{q}^i) d^{n}\hat{q} = \int \psi^\dagger(\hat{q}^i) f(\hat{q}^i) \psi(\hat{q}^i) d^{n}\hat{q} = \int \psi^\dagger(\hat{q}^i) F_{[\hat{q}^i]} (\psi(\hat{q}^i)) d^{n}\hat{q}$.
		
	We claim that for any fragment variable $f$ there exists a corresponding self-adjoint linear operator $F_{[q^i]} : \Psi(\mathcal{Q}, dq^i) \rightarrow C(\mathcal{Q}, \mathbb{C})$ such that the integral $\int \psi^\dagger(q^i) F_{[q^i]}( \psi(q^i)) d^nq$ is, if it exists, the expectation of $f$ multiplied by the amount of material. Let $f$ be a fragment variable. There exists a corresponding self-adjoint linear operator $F_{[\hat{q}^i]} : \Psi(\hat{\mathcal{Q}}, d\hat{q}^i) \rightarrow C(\hat{\mathcal{Q}}, \mathbb{C})$ where $\hat{\mathcal{Q}}$ provides a maximal set of fragment variables. As both $\hat{\mathcal{Q}}$ and $\mathcal{Q}$ provide a maximal set, we can define $\Upsilon^{q^i}_{\hat{q}^i} \equiv \Upsilon_{\hat{q}^i} \circ \Upsilon_{q^i}^{-1}$ which transforms a wave function over $\mathcal{Q}$ into a wave function over $\hat{\mathcal{Q}}$. $\Upsilon^{q^i}_{\hat{q}^i}$ is an isomorphism as $\Upsilon_{\hat{q}^i}$ and $\Upsilon_{q^i}$ are isomorphisms. $\Upsilon^{q^i}_{\hat{q}^i}$ is unitary as $\Upsilon^{\hat{q}^i}_{q^i} \circ \Upsilon^{q^i}_{\hat{q}^i} = (\Upsilon^{q^i}_{\hat{q}^i})^{-1} \circ \Upsilon^{q^i}_{\hat{q}^i} = I$. We can extend $\Upsilon^{q^i}_{\hat{q}^i}$ on all continuous functions and define $F_{[q^i]} \equiv \Upsilon^{\hat{q}^i}_{q^i} \circ F_{[\hat{q}^i]} \circ \Upsilon^{q^i}_{\hat{q}^i}$. $F_{[q^i]}$ is a self-adjoint linear operator. $\int \psi^\dagger(q^i) F_{[q^i]}( \psi(q^i)) d^nq = \int \psi^\dagger(\hat{q}^i) F_{[\hat{q}^i]}( \psi(\hat{q}^i)) d^n \hat{q}$ is, if it exists, the expectation of $f$ multiplied by the amount of material.
\end{justification}

The last thing to do is characterize the conjugate variables $k_i$ and $p_i$. To be precise, we want to find $K_{i [q^i]}$, the operators that act on the wave functions over $q^i$ corresponding to fragment variables $k_i$. We know that the operators have to be covariant: under a change of units $\hat{q}^j = \hat{q}^j(q^i)$ we must have $\hat{K}_{j[q^i]} = \partial_{\hat{q}^j} q^i K_{i [q^i]}$. In other words: we are looking for the set of linear operators that obey the chain rule. This turns out to be the space of derivations, therefore $K_{i [q^i]}$ will be the derivatives in the direction of the corresponding $q^i$. They have to be self-adjoint, so they will be of the form $K_{i [q^i]} = a \imath \partial_{i}$ where $a$ is a real number. We choose $a$ to be negative by convention and to be unitary so that $k_i$ is expressed in units of inverse $q^i$ as in the classical case. Therefore we have $K_{i [q^i]} = -\imath \partial_i$ and $P_{i [q^i]} = -\imath \hbar \partial_i$.

If we assume $\mathcal{Q}$ to be $\mathbb{R}^n$, the functional $\Upsilon^{q^i}_{p_i} : \Psi(\mathcal{Q}, dq^i) \leftrightarrow \Psi(\mathcal{P}, dp_i)$ that converts a wave function in position to a wave function in conjugate momentum is the Fourier transform. As the fragment variables for any power of $q^i$ and $p_i$ must exist, the wave functions in both expressions are infinitely differentiable: the space of the wave functions $\Psi(\mathcal{Q}, dq^i) = S(\mathcal{Q}, q^i)$ is the Schwartz space of rapidly decreasing smooth functions.

\begin{prop}\label{prop:momentum_operator}
	Let $\mathcal{Q}=\mathbb{R}^n$. The operator $P_{i [q^i]} : \Psi(\mathcal{Q}, dq^i) \rightarrow C(\mathcal{Q}, \mathbb{C})$ associated with the conjugate quantity of $q^i$ is $P_{i [q^i]} = - \imath \hbar \partial_{i}$. The isomorphism $\Upsilon^{q^i}_{p_i} : \Psi(\mathcal{Q}, dq^i) \leftrightarrow \Psi(\mathcal{P}, dp_i)$ that changes the variable of the wave function from $q^i$ to $p_i$ is $\Upsilon^{q^i}_{p_i} (\psi(q^i)) = \frac{1}{(\sqrt{2\pi})^n} \int_{\mathcal{Q}} e^{\frac{q^i p_i }{\imath \hbar}} \psi(q^i) d^n q $. The space of wave functions $\Psi(\mathcal{Q}, dq^i) = S(\mathcal{Q}, q^i)$ where $S(\mathcal{Q}, q^i)$ is the Schwartz space of rapidly decreasing smooth functions. 
\end{prop}
\begin{justification}
	We claim the operator associated with $k_i$ is $K_{i [q^i]} = - \imath \partial_{i}$. Let $K_{i [q^i]} : \Psi(\mathcal{Q}, dq^i) \rightarrow C(\mathcal{Q}, dq^i)$ be the linear operator associated with $k_i$ as expressed on the space of wave functions over $q^i$. Such an operator must be covariant under the transformation $\hat{q}^j = \hat{q}^j(q^i)$. That is: $\hat{K}_{j [q^i]} = \partial_{\hat{q}^j}q^iK_{i [q^i]}$. In particular, it needs to be covariant when changing only one variable: $\hat{q}^1 = \hat{q}^1(q^1)$, $\hat{q}^j = q^j$ where $j=2...n$. By \ref{thrm:covariant_operator} $K_{i [q^i]}=c\partial_i$ for some $c \in \mathbb{C}$. $K_{i [q^i]}$ is self-adjoint: $K_{i  [q^i]}=\imath r\partial_i$ for some $r \in \mathbb{R}$ by \ref{thrm:antihermitian_derivative}. $k_i$ is expressed in units of inverse $q^i$ therefore $|r|=1$. By convention, we choose $r$ to be negative. $K_{i [q^i]} = - \imath \partial_{i}$.
	
	We claim the operator associated with $p_i$ is $P_{i [q^i]} = - \imath \hbar \partial_{i}$. As $p_i=\hbar k_i$, their expectation is proportional by a factor of $\hbar$. $P_{i [q^i]}=\hbar K_{i [q^i]}= -\imath \hbar \partial_i$.
	
	We claim the operator that changes the variable of the wave function from $q^i$ to $p_i$ is $\Upsilon^{q^i}_{p_i} (\psi(q^i)) = \frac{1}{(\sqrt{2\pi})^n} \int_{\mathcal{Q}} e^{\frac{q^i p_i }{\imath \hbar}} \psi(q^i) d^n q $. Let $\Upsilon^{q^i}_{p_i} : \Psi(\mathcal{Q}, dq^i) \rightarrow \Psi(\mathcal{P}, dp_i)$ be the operator that maps a wave function over $\mathcal{Q}$ into a wave function over $\mathcal{P}$, the space charted by $p_i$. $\Upsilon^{q^i}_{p_i}$ must exist and be unique (up to a total phase) as $\Upsilon^{q^i}_{p_i} = \Upsilon_{p_i}\circ\Upsilon^{q^i} = \Upsilon_{p_i}(\Upsilon_{q^i})^{-1}$ is an isomorphism. $\Upsilon^{q^i}_{p_i}$ must be a linear transform that takes the operator $P_{i [q^i]}= - \imath \hbar \partial_i$ as expressed over $q^i$ to $P_{i [p_i]} = p_i$ as expressed over $p_i$. Such transformation is the Fourier transform. $\Upsilon^{q^i}_{p_i} (\psi(q^i)) = \frac{1}{(\sqrt{2\pi})^n} \int_{\mathcal{Q}} e^{\frac{q^i p_i }{\imath \hbar}} \psi(q^i) d^n q $.
	
	We claim $\Psi(\mathcal{Q}, dq^i) \subseteq S(\mathcal{Q}, q^i)$. Let $\psi(q^i) \in \Psi(\mathcal{Q}, dq^i)$. $(P_{i [q^i]})^m \psi(q^i) = (-\imath \hbar \partial_i)^m \psi(q^i)$ represents the $m^{th}$ power of the conjugate momentum operator applied to the given state. As the quantity is physically meaningful, such an operation must be defined for any $m$ therefore $\psi(q^i) \in C^\infty(\mathcal{Q})$. Let $\psi(p_i) = \Upsilon^{q^i}_{p_i} (\psi(q^i))$ be the wave function over momentum space. $(Q^i_{[p_i]})^m \psi(p_i) =\journal{\break} (+\imath \hbar \partial_{p_i})^m \psi(p_i)$ represents the $m^{th}$ power of the $q^i$ operator applied to the given state. As the quantity is physically meaningful, such an operation must be defined for any $m$ therefore $\psi(p_i) \in C^\infty(\mathcal{P})$. As $\psi(q^i)$ is the Fourier transform of an infinitely smooth function, it decreases faster than any inverse power of $q^i$. That is: $\forall m \in \mathbb{Z}^+ \exists c \in \mathbb{C}$ such that  $|\psi(q^i)|<\frac{|c|}{|q^i|^m} \; \forall |q^i|> 1$. $\Psi(\mathcal{Q}, dq^i) \subseteq S(\mathcal{Q}, q^i)$.
	
	We claim $\Psi(\mathcal{Q}, dq^i) = S(\mathcal{Q}, q^i)$. Let $\psi(q^i) \in \Psi(\mathcal{Q}, dq^i)$ and $\phi(q^i) \in S(\mathcal{Q}, q^i)$ such that $\psi(q^i) \neq \phi(q^i)$. There exists a choice of $a, b, j, k \in \mathbb{Z}^+$, $j,k \leq n$ such that $\lVert \psi(q^i) - \phi(q^i) \rVert_{j,a,k,b} \neq 0$  where $\lVert \psi(q^i) \rVert_{j,a,k,b} = sup|(q^j)^a (\partial_k) ^b \psi(q^i)|$ as $S(\mathcal{Q}, q^i)$ is a Frechet space induced by that family of seminorms. The operator $(Q^j_{[q^i]})^a(P_{k [q^i]})^b$ applied to $\psi(q^i)$ and $\phi(q^i)$ will provide different distributions for the associated fragment variables. $\phi(q^i)$ is physically distinguishable and is associated with a physical state. $\phi(q^i) \in \Psi(\mathcal{Q}, dq^i)$.
\end{justification}

We now have all the basic features of the state space of a quantum particle. And we obtained them by re-deriving them from first principles instead of using some formal analogy. The similarities with the classical framework are simply due to the similarities of the starting points.

It should be evident, though, how the correspondence between mathematical and physical objects is not as satisfying as in the classical case. For instance, while any differentiable function of $T^*\mathcal{Q}$ can be a state variable for a classical particle, not all self-adjoint linear operators can be associated with a fragment variable. Consider $f: \mathbb{R} \rightarrow \mathbb{R}$ for which $f(r) = -r$ if $r$ is an integer and $f(r) = r$ otherwise. Clearly $f(q)$ does not preserve the topology and cannot be used as a fragment variable. Yet, we can still create a self-adjoint linear operator, with the same spectra as the one associated with $q$, just with the eigenvalues switched. The space of physically meaningful operators, then, should be further constrained. A reasonable requirement would be that they preserve the differentiable/integrable structure of the Schwartz space $S(\mathcal{Q}, q^i)$ as classical state variables preserved the differentiable structure of $T^*\mathcal{Q}$. Unfortunately, this is not as easy to justify.

The topology of $\mathcal{P}$, the space charted by conjugate momentum, depends globally on the topology of $\mathcal{Q}$. For example, if $\mathcal{Q}$ is an n-dimensional torus, $\mathcal{P}$ has a discrete topology. That is why, in the end, we restricted ourselves to $\mathbb{R}^n$. Moreover, we characterized the space of functions not because we required them to fall off at infinity, but because we required the ability to express polynomials of $q^i$ and $p_i$. It is not immediately clear how this generalizes over a manifold: further work is needed.

While the link between math and physics may not be as elegant and as general as one would like, we still managed to meet the goal: we identified the states of an irreducible material and saw how these are related to the states in quantum mechanics.

%If the arguments are the same, the projection returns the same state, therefore $<\psi_1|\psi_1> = |\psi_1|^2$: the product is positive definite. Since the projection is linear, the product is linear as well. 

%The coefficient $\alpha$ is a function of both $\psi_1$ and $\psi_2$, so we write it as $\alpha = <\psi_2|\psi_1>$. It is linear in $\psi_1$ as $\mathsf{P}_{\psi_2}$ is linear. The square module $|<\psi_2|\psi_1>|^2$ gives us the amount of material in $\psi_1$ that has the same distribution of $\psi_2$. $|<\psi_1|\psi_1>|^2$ As all states 

\subsection{Irreducibility}

Now that we have fully characterized what we mean by an irreducible material, we can stipulate the following:

\begin{assump}[Irreducibility]\label{ass:irreducibility}
	The system under study is composed of an irreducible homogeneous material and as a whole undergoes deterministic and reversible evolution.
\end{assump}

\begin{rationale}
	The idea is that each amount of material has a state, the physically distinguishable configuration of the distribution, and an unstated part, the physically indistinguishable configuration of the fragments of the material. The system is deterministic and reversible in the sense outlined in assumption \ref{ass:determinism}: given the state of the system at one time (i.e. the whole distribution) we can predict/reconstruct the future/past states. But we cannot describe the evolution of the unstated part (i.e. the motion of the fragments).
	
	As we saw, the state space recovered under this assumption is the one used in quantum mechanics. In the following sections we will derive two types of state evolution: the deterministic and reversible one, corresponding to unitary evolution (i.e.  Schroedinger's equation), and a non-deterministic one useful to describe measurement interactions, corresponding to the projection postulate. Therefore, as the infinitesimal reducibility assumption leads to classical mechanics, the irreducibility assumption leads to quantum mechanics.
	
	As we saw in section \ref{subsec:infinitesimal_reducibility} during the rationale for assumption \ref{ass:infinitesimal_reducibility}, infinitesimal reducibility had its problems. How about the irreducibility assumption? How do they compare?
	
	The methodological problem takes on a different character. In order for the assumption to hold, we need to show that there are no processes at our disposal that can probe the fragments of a quantum particle. For example: an electron as a whole interacts with a photon as a whole. Note that this does not require showing that no such process exists: conceptually, specific parts of the photon may be interacting with specific parts of the electron. But if those processes are not available to us, say because we can only manipulate whole photons and they always interact with a whole electron, they cannot be used to create a more detailed picture. So we simply have to enumerate the processes available to us, which is, at least in principle, feasible. Yet, we can never rule out that someone in the future may discover a new process.

	The issues arising from a lack of perfect isolation are addressed. Each degree of freedom comes with its own unstated part (i.e. the motion of the fragments within that d.o.f.); therefore there is no pretense of a perfect description under perfect isolation. Since we cannot describe the nature of the chaotic motion, we cannot say whether it is internally or externally driven.
	
	The problem of physically meaningless interactions is also considerably lessened. Even if we have multiple particles each with their state evolving deterministically, we can still imagine their unstated parts as interacting with each other. Therefore each quantum particle would not essentially reside in its own separate universe.
	
	In these ways, the irreducibility assumption fares much better than the infinitesimal reducibility assumption. Yet, it is still plagued with some of the same problems.
	
	The problem of incompatibility of measurements with deterministic evolution at all levels still remains. The unstated part of multiple irreducible particles may be interacting, but it cannot be used to pass information from one subsystem to another without making itself distinguishable.
	
	The inability to define time using a deterministic and reversible evolution remains a problem. As the system is homogeneous and irreducible, the unstated part cannot give us a measurable quantity that changes in time as that would provide a description of the unstated part.
	
	There is one additional problem, though, not present in the classical case. Under the irreducibility assumption states are necessarily at equilibrium: they all provide the same level of description for the fragments, they carry the same information entropy. In fact if the internal dynamics were out of equilibrium, the material would not be homogeneous which would contradict the assumption. Deterministic and reversible evolution (which we'll see corresponds to unitary evolution) is necessarily quasi-static: at each moment a state must be well defined, at each moment the system is in equilibrium, the system changes slowly. Therefore quantum states, as formulated here, cannot and should not be expected to provide a valid description during non-deterministic evolution.
	
	In light of this, consider a muon decay: the initial state of the muon and the final states of the electron and neutrinos can certainly be well described by quantum particle states.\footnote{In this context, a state is in equilibrium regardless of its being in a local minimum, easily overcome by fluctuations, or a global minimum, that persists indefinitely.} We have no reason to expect, though, that while the muon is decaying into the resulting particles the state is always in equilibrium and therefore well described mathematically by a vector in a complex inner product space. After all, the state of a single spin-one-half particle cannot determine the future state of three spin-one-half particles.
	
	This notion may seem in contradiction to the current practice. For example, the S-matrix, one of the main devices for calculating scattering amplitudes and cross sections, is derived by assuming unitary evolution~\cite{Weinberg}. Doesn't it imply we are assuming that deterministic and reversible evolution is happening at every instant?
	
	If we look more closely, we realize that the process that the S-matrix captures is a limit where initial states are at $t \rightarrow - \infty$ and final states are at $t \rightarrow + \infty$, a process that takes an infinite amount of time.  That is, to calculate predictions for a fast out-of-equilibrium process we approximate it using slow quasi-static evolution. This is actually not so uncommon: it is a standard first approximation in thermodynamics. The question is why should it work here? As we'll see in more detail when talking about the projection postulate, it's because of the irreducibility assumption itself.
	
	Suppose we are trying to calculate a quantity that does not depend on the particular motion of the unstated part. Then, by its very nature, its value will be the same under an equivalent quasi-static evolution, because this too does not depend on the particular motion of the fragments. The probability distribution of the out-states during scattering is one such quantity. As the final statistical distribution is determined only by the initial state, it is not dependent on the particulars of the unstated part. In these cases, the use of deterministic and reversible evolution as a ``stand-in" to calculate final statistical distributions is justified. However, it should not be taken literally as a physical model of what actually happens. There is no justification for assuming this process to be quasi-static deterministic and reversible evolution of quantum particle states.
	
	In other words, the assumption of irreducibility already contains the seeds of its own demise. On one hand, it tells us that there is a component of chaotic motion within the system. On the other hand, it can only characterize the equilibrium dynamics, which clearly can't be the full range of dynamics. Yet it also offers a way out: since the motion of the fragments is not accessible, we can predict statistical distributions that are independent of it.
	
	In light of what we discussed, the irreducibility assumption is less flawed than infinitesimal reducibility but it still cannot be taken as fundamental, in the sense that we cannot take it to strictly apply to all systems. We should not expect to solve open problems, such as the arrow of time and the measurement problem, under the assumption as formulated in this work.
\end{rationale}

\subsection{Quantum mechanics}

We are now ready to write the equations of motion for an irreducible material. Deterministic and reversible evolution for the whole distribution necessarily means quasi-static evolution: there is a well defined state at each moment in time. Moreover the components must evolve independently as we cannot learn more about one component by observing its evolution in different combinations: the evolution of the composition is the composition of the evolution. In other words: the evolution must preserve the inner product.

Note that formally there is a strong analogy with the reducible material case. In classical mechanics the state space structure was captured by a distribution over a symplectic manifold. Deterministic and reversible evolution was a transformation that preserved that structure, a symplectomorphism. In the quantum case the state space structure is captured by an inner product space. Deterministic and reversible evolution is a transformation that preserves that structure, a unitary transformation. That is: once we have defined what mathematical structure captures the definition of a state, deterministic and reversible evolution must necessarily preserve such structure.

\begin{prop}\label{prop:unitary_evolution}
	A deterministic and reversible evolution map $\mathcal{T}_{\Delta t}: \Psi \rightarrow \Psi$ for an irreducible material is a unitary operator. That is: $\langle \mathcal{T}_{\Delta t} \psi_1, \mathcal{T}_{\Delta t} \psi_2 \rangle = \langle \psi_1 , \psi_2 \rangle$ where $\psi_1, \psi_2 \in \Psi$.
\end{prop}

\begin{justification}
	We claim $\mathcal{T}_{\Delta t}: \Psi \rightarrow \Psi$ is a unitary operator. Let $\psi_1, \psi_2 \in \Psi$ be two states. $\mathcal{T}_{\Delta t}(a\psi_1+b\psi_2) = a \mathcal{T}_{\Delta t} \psi_1 + b \mathcal{T}_{\Delta t} \psi_2$ with $a,b \in \mathbb{C}$ as the composition of the evolution is the evolution of the composition: $\mathcal{T}_{\Delta t}$ is a linear operator. $\mathcal{T}_{\Delta t} \mathsf{P}_{\psi_1} (\psi_2) = \mathsf{P}_{\mathcal{T}_{\Delta t}\psi_1} (\mathcal{T}_{\Delta t}\psi_2)$ as the evolved part of $\psi_2$ that was prepared like $\psi_1$ is the part of the evolved $\psi_2$ that ends prepared like the evolved $\psi_1$. We have $\mathcal{T}_{\Delta t} \mathsf{P}_{\psi_1} (\psi_2) = \mathcal{T}_{\Delta t} \langle \psi_1 , \psi_2 \rangle \psi_1 = \langle \psi_1 , \psi_2 \rangle \mathcal{T}_{\Delta t} \psi_1 = \mathsf{P}_{\mathcal{T}_{\Delta t}\psi_1} (\mathcal{T}_{\Delta t}\psi_2) = \langle \mathcal{T}_{\Delta t}\psi_1 , \mathcal{T}_{\Delta t}\psi_2 \rangle \mathcal{T}_{\Delta t} \psi_1$. $\langle \psi_1 , \psi_2 \rangle = \langle \mathcal{T}_{\Delta t}\psi_1 , \mathcal{T}_{\Delta t}\psi_2 \rangle$. $\mathcal{T}_{\Delta t}$ is unitary.
\end{justification}

\begin{prop}\label{prop:schroedinger_equation}
	A continuous deterministic and reversible evolution for an irreducible material admits a Hamiltonian operator $H : \Psi \rightarrow \Psi$ that allows us to write the laws of evolution as
	\begin{align*}
	\imath \hbar \partial_t \psi = H \psi
	\end{align*}
\end{prop}

\begin{justification}
	We claim the state evolves according to Schroedinger's equation. Let  $\mathcal{T}_{dt}: \Psi \rightarrow \Psi$ be the evolution for an infinitesimal time interval $dt$. As it is unitary, it can be expressed as $\mathcal{T}_{dt} = I + \frac{H dt}{\imath \hbar}$ where $H : \psi \rightarrow \psi$ is a self-adjoint operator. Let $\psi_t \in \Psi$ be a state and $\psi_{t+dt} \in \Psi$ its evolution after an infinitesimal time interval. We have $\mathcal{T}_{dt} \psi_t = \psi_{t+dt} = \psi_t + \frac{H dt}{\imath \hbar} \psi_t$. $\imath \hbar \frac{\psi_{t+dt} - \psi_t}{dt} = \imath \hbar \partial_t \psi_t = H\psi_t$.
\end{justification}

We recognize the Schroedinger equation, the time evolution equation for a quantum particle state.

Note, though, that nothing tells us that the system must be in an eigenfunction of $H$ or that the ground state is somehow preferred. In fact, the eigenfunctions of $H$ may not be part of $S(\mathcal{Q}, q^i)$ and therefore may not be physical. In other words: we have no basis for the time independent Schroedinger equation. This is actually consistent with our assumptions. For example, for a system to reach its ground state it has to be able to radiate energy: it cannot be isolated. Therefore this dynamics cannot be described under the assumption of deterministic and reversible evolution. This already hints that deterministic and reversible evolution is not the only dynamics possible. Other considerations, outside of the scope of this work, can be used to recover this behavior (e.g. thermodynamics, decoherence, ...)~\cite{Weiss}.

Also note that $H$ operates on $\Psi$ which, in the case of $\mathcal{Q}=\mathbb{R}^n$, corresponds to the space $S(\mathcal{Q}, q^i)$ of Schwartz functions. $H$ therefore preserves infinite smoothness and integrability. This, again, reminds us that the space of all self-adjoint operators is far too vast, and does not map well to physically meaningful quantities.

\subsection{Time dependence and kinematic assumption}

The next step would be to integrate the quantum description with time dependent evolution (i.e. time dependent Hamiltonians and relativistic mechanics) and the kinematic assumption (i.e. Lagrangian mechanics). We will only provide a sketch, though, without going through the details. Given that we are limited to single particles (i.e. no fields) without spin (i.e. minimal material only described by position), our results are limited to finding the Klein-Gordon equation modified with electromagnetic interactions. While this is still of note, the equation is of limited physical use and therefore may not warrant the space necessary for a more rigorous derivation.

As in the classical case, the first thing to do is to extend $\mathcal{Q}$ to include time. This gives us the same $\mathcal{M}$ as before. A quantum state is a complex valued distribution $\psi : \mathcal{M}_{t=t_0} \rightarrow \mathbb{C}$ over a hypersurface at constant time. An evolution $\widetilde{\psi} : \mathcal{M} \rightarrow \mathbb{C}$ is a complex valued function defined on the whole $\mathcal{M}$. The operator for the quantity conjugate to time is $E=\imath\hbar\partial_t$.

Next we need to extend integration. In the quantum case, this means to make sure the inner product $\langle \psi | \phi \rangle$ remains unchanged under all variable changes, including time. This allows us to describe the amounts of material and their distribution over fragments in a way that is independent of our description.

As in the classical case, we use $s$ as the evolution parameter. Deterministic and reversible evolution $\mathcal{T}_{ds}$ will be a unitary transformation defined over this extended inner product, as the structure of the state space is conserved. Continuous evolution will admit an invariant Hamiltonian $\mathcal{H}$ such that $\imath \hbar \partial_s \psi = \mathcal{H} \psi$. As for the composite state evolution in the classical case, $\widetilde{\psi}$ remains unchanged by the evolution as it transports the distribution along $s$. Therefore $\mathcal{H} \widetilde{\psi} = \mathcal{h} \widetilde{\psi}$ is an eigenfunction of $\mathcal{H}$. As in the classical case, we use $\mathcal{H}=0$ to identify the state space: the temporal d.o.f. is not independent and therefore there is a constraint among the variables. Therefore we have $\mathcal{H} \psi = 0$. This is the time dependent evolution equation. It means that the wave function does not change along $s$, the parameter of the infinitesimal transformation generated by $\mathcal{H}$.

Note that states are indeed eigenfunctions of the invariant Hamiltonian $\mathcal{H}$: here we do have a justification, as opposed to the standard Hamiltonian $H$ case. For example, a massive particle may have fragments with different energy, but they all must carry the same fraction of mass. Also note that $\mathcal{H} \psi = 0$ is basically the form of all equations (i.e. Klein-Gordon, Dirac, Maxwell, ...) in quantum field theory. The respective invariant Hamiltonian $\mathcal{H}$ is the operator associated with each equation.

Having handled the time dependent case, we can turn our attention to the kinematic equivalence assumption. While the trajectory of a fragment is not well defined in either phase space or physical space, what we can describe remains the same under the assumption. Therefore the distributions still need to be transported from state space to kinematic variables up to a constant factor. This means we still have a linear relationship between velocity and conjugate momentum. In terms of the quantum operators: $m\, d_s Q^\alpha = g^{\alpha\beta} P_{\beta} - A^\alpha = m\frac{[Q^\alpha, \mathcal{H}]}{\imath \hbar}$.

This leads to the same form for the invariant Hamiltonian $\mathcal{H}=\frac{1}{2m}((P_\alpha-A_\alpha)g^{\alpha\beta}(P_\beta-A_\beta) + m^2 c^2) = \frac{1}{2m}((-\imath\hbar\partial_\alpha-A_\alpha)g^{\alpha\beta}(-\imath\hbar\partial_\beta-A_\beta) + m^2 c^2)$. In the free particle case, we have $\frac{1}{2m}(-\hbar^2 \partial^\alpha \partial_\alpha + m^2 c^2) \widetilde{\psi} = 0$ which is the Klein-Gordon equation. We can also write the operator for kinetic momentum as $mU^\alpha = P^\alpha - A^\alpha = -i\hbar D^\alpha$ where $D^\alpha = \partial^\alpha + \frac{A^\alpha}{\imath \hbar}$ is the gauge covariant derivative. We have $(\frac{1}{2}mU^\alpha U_\alpha + \frac{1}{2} m c^2) \widetilde{\psi} = \frac{1}{2m}(-\hbar^2 D^\alpha D_\alpha + m^2 c^2) \widetilde{\psi} = 0$ which is the Klein-Gordon equation with electromagnetic interaction.

Note how the gauge covariant derivative is, up to a constant, kinetic momentum. This motivates why it is such an important operator. As for the Klein-Gordon equation, the only thing it does is to impose that the norm of the four-velocity is the speed of light, a condition much more trivial than one may have anticipated.

It should be clear, even without the details, that all the assumptions come together even in the quantum case. They give more insight to all the pieces in a comprehensive way and help tear down the artificial walls between the different physical theories. This allows us to form a picture of fundamental physics that is more unified, that draws upon a set of common concepts.

\subsection{Non-deterministic evolution and the projection postulate}

As we made it clear from the start, deterministic and reversible evolution is an assumption. And as we have seen in the rationale for assumptions \ref{ass:infinitesimal_reducibility} and \ref{ass:irreducibility} we cannot have determinism at all levels of decomposition while transferring information from one system to another.

In this section we will see how non-deterministic processes become necessary when describing measurements for irreducible systems. But we'll see that, under the same assumption of irreducibility, we are able to make a connection between the output of deterministic and non-deterministic processes.

Suppose we have a quantum particle state $\psi$ we want to identify. Suppose that we have a detector that is able to capture the particle. If it's captured, we will know that the particle was in the region of space enclosed by the detector.

Suppose that we create an array of such detectors. This will allow us to detect the position of the particle with a resolution given by the aperture of each detector.

Suppose now, though, that the quantum particle state is spread over a region significantly greater than a single detector but smaller than the array. What will happen?

It is clear that it has to be absorbed as it is within the aperture of the whole array. It is also clear that it must be absorbed as a whole system: as it is an irreducible system, the detector cannot partially interact with it, it cannot half absorb it. So a single detector will capture it. It is also clear that it cannot always be the same detector: that corresponds to states where the particle is spread within the aperture of that detector. In other words: as the particle hits the array, one detector triggers non-deterministically based on the distribution of the particle.

The process is non-deterministic in the sense that the state of the particle is not enough to predict which detector will trigger. Maybe it depends on the state or unstated part of the detector, maybe it depends on the unstated part of the particle or maybe on the environment. The point is that a single measurement does not tell us much about the state of the incoming particle.

We then repeat the process, making sure the incoming particle is prepared in the same way. Given a single ``pure" incoming particle state, we get a ``mixed" statistical distribution describing the output of the array.\footnote{Note that this is different from the model used often to describe the quantum measurement problem, where the incoming system becomes entangled with the measuring device. We do not make any claim on what the final state of the particle should be: in many cases, like in a photomultiplier tube, the particle is gone and there is no state. We simply have one incoming state for the particle, prepared over and over exactly the same, and the outgoing state for the measurement device, which is a classical statistical distribution over the different measurement takes.} The question is, how is the statistical distribution coming out of the detector related to the material distribution coming in with the particle?

If the detector does not introduce biases, then the final distribution is only a function of the incoming particle: all unbiased detectors will give the same result. The unstated part also cannot affect the final distribution of outcomes: it is not physically distinguishable. The final statistical distribution of the detector must match the incoming material distribution of the particle. That is if $N$ is the total number of takes and $N_A$ is the number of times the detector with aperture in the $A$ region captured the particle, the ratio $\frac{N_A}{N}$ will be approximately the amount of material within the aperture $\int_A \psi^\dagger(q^i) \psi(q^i) d^nq$. If we assume an infinite number of takes, we expect them to coincide: $\int_A \psi^\dagger(q^i) \psi(q^i) d^nq = \lim\limits_{N\rightarrow\infty} \frac{N_A}{N}$. But this is also the frequentist definition of probability. Therefore we have $\int_A \psi^\dagger(q^i) \psi(q^i) d^nq = P(\textrm{measuring particle in } A)$ thus recovering the probabilistic interpretation of the wave function.

We can conceptually extend this picture to the generic case. Suppose we know the outcomes of our measurement are parametrized by a fragment variable $\hat{q}$. We will have a linear functional $\Lambda^{\frac{1}{2}}_{\hat{q}=\hat{q}_0} : \Psi \rightarrow \mathbb{C}$ such that $|\Lambda^{\frac{1}{2}}_{\hat{q}=\hat{q}_0}|^2$ returns the amount of material prepared with the particular value $\hat{q}_0$. We can re-express this quantity in terms of the dual vector $\hat{q}_0$ corresponding to the linear functional $\Lambda^{\frac{1}{2}}_{\hat{q}=\hat{q}_0}$ and the inner product extended to all Lebesgue integrable functions: $|\Lambda^{\frac{1}{2}}_{\hat{q}=\hat{q}_0}|^2 = |\langle \hat{q}_0 | \psi \rangle|^2 $. This is also the probability that the outcome associated with $\hat{q}_0$ is selected by the measuring device, thus recovering the statistical nature of the quantum predictions. 

Note that while this process is of a non-deterministic nature, it is of a specific non-deterministic nature. It requires that no other state can influence the final distribution because the match between the physical material distribution and the statistical measurement distribution was recovered only under the assumption that the detector does not introduce any bias. That is: the non-deterministic nature cannot be understood as an averaging over some unknown external state.

Also note that while we have a way to predict the distribution over final outcomes from the initial state, we have no way of describing what happens in between. When we assume that only the initial state matters, we are saying that we can disregard the particulars of the non-deterministic evolution. So much so that we can substitute it with a deterministic one to compute the probability distribution. In particular, given only the initial state, we have no way of predicting which set of final outcomes the system will project to, nor when that projection actually happens, as that depends on the actual choice of detector setup.

While this may seem like an incomplete description, it is as complete as the irreducibility assumption allows. The state can be known and determines any final statistical distribution. The unstated part cannot be known, and is what gives the statistical distribution in the first place. If we were able to say something more about the evolution of the unstated part, we would be able to further describe it, which we cannot do without violating the assumption. If we cannot learn anything about the unstated part from the final distribution, then we are consistent with our assumption and we are able to use deterministic evolution and projection to determine the final statistical distribution. In other words: the inability to describe the unstated part not only constrains the state space and deterministic evolution, but also the set of non-deterministic processes by excluding those for which the final statistical distribution would depend on the unstated part itself.

We therefore have arrived at all the major underpinnings of quantum mechanics and see how they stem from assumption \ref{ass:irreducibility} of irreducibility. This gives us a more satisfactory account of what quantum mechanics describes, although it does not claim to solve many of the open problems. It does, however, give good reasons as to why we have those problems and how the desire to know ``what is really going on" is hindered by the very assumption that the quantum description is based on.

\section{Discussion and further work}
\label{sec:discussion}

As we went through many details in the derivation, it is time to step back and conclude with a few general remarks.

The first remarkable aspect of this work is that something along these lines can be done. Regardless of whether all the minor details are addressed in the best possible way, the idea that so much can be derived from so little was not something we expected. In hindsight, the idea that physics should be founded on physical principles is nothing new: it just hasn't been thoroughly applied to basic mechanics. If this work had been done incrementally from the start, these concepts would have been developed gradually and the results would feel more mundane.

The second striking aspect is a shift in perspective that this work seems to suggest. The usual reductionist view is that you start with the rules of the small parts and then build up the composite whole. This work flips that perspective: it's not what the parts can tell you about the whole, it's what the whole can tell you about the parts.

The state space for classical particles is the one that allows us to express invariant distributions: that is a requirement that comes from the whole. The fragments of quantum particles do not even have a well defined description: the whole tells us nothing about their configuration. In retrospect, this should not be surprising. Nobody is going to hand us down the laws of infinitesimal parts: we induce  what we can by manipulating finite systems. Therefore what the whole allows us to infer about its parts is physically fundamental.

The third notable aspect is that the assumptions have more to do with our ability to create a reliable description rather than essential features of the system itself. A system is not deterministic per se: it is a particular choice for its description (i.e. state) that may undergo deterministic evolution in some cases. In the same vein, a material is not homogeneous per se. Or infinitesimally reducible. In this context, the state, its properties and the laws of evolution are the description of the system under our idealization, not some intrinsic feature of the system.

The fourth intriguing aspect is that, when framing classical and quantum mechanics in this way, there is a progression. First we assume a homogeneous infinitesimally reducible material. Then we assume a homogeneous irreducible material. Extrapolating from this progression, the next assumption would be an inhomogeneous irreducible material: where the material is not in constant equilibrium and its particles do not all look the same.

This different starting point may have potential connections to some of the problems we mentioned in the rationales for the assumptions and to other interesting physics questions in general. In the foreseeable future, however, we will work on the more mundane task of extending this work to field theories, for which a more concrete path is available.

\section{Conclusion}

We have seen that much of fundamental physics can be derived from few simple physical assumptions. In doing so we obtained a more unified picture of the different branches of fundamental physics centered around the idea of deterministic and reversible evolution of distributions of homogeneous materials. Though the work is extensive we can summarize a few important conclusions with the following points:

\begin{itemize}

\item A state space is not just a collection of states for the system under study but it must also capture physical properties such as physical distinguishability (represented by a topology), the ability to count states (represented by a measure, symplectic form or metric tensor) or decomposability (represented by a vector space structure).

\item Deterministic and reversible evolution is a one-to-one map that preserves the nature of the system (e.g. physical distinguishability, state count and decomposability). Therefore it is an isomorphism on the mathematical structure defined by the state space (e.g. homeomorphism, symplectomorphism, unitary transformation), not just a state-to-state bijection (i.e. isomorphism of sets).

\item Classical Hamiltonian mechanics describes the evolution of a homogeneous material whose state is equivalent to the states of all its infinitesimal parts, the classical particles. The state space for classical particles is a symplectic manifold $(T^*\mathcal{Q}, \omega)$ as this allows us to define state-variable-invariant density distributions; deterministic and reversible evolution for such a state space is a symplectomorphism (i.e. a canonical transformation).

\item Lagrangian mechanics describes the case where the trajectory in space fully characterizes the dynamics of the system. If we want distributions in phase space to be expressed in terms of kinematic variables (i.e. position and velocity), we find that $(\mathcal{Q}, g)$ is a Riemannian manifold and the motion is that of a massive particle under scalar/vector potential forces.

\item Relativistic motion arises by properly handling time dependent deterministic and reversible evolution for the material (i.e. distributions over phase space) without further assumptions. That is, we derive the geometrical structure of space-time from the geometrical structure of phase space, inverting the typical approach.

\item Quantum (Hamiltonian) mechanics describes the evolution of a homogeneous material for which the state can only be reduced up to finite parts, the quantum particles, whose internal motion is undetermined. The state space for quantum particles is a complex inner product space $(\Psi, \langle \cdot , \cdot \rangle)$; deterministic and reversible evolution for such a state space is a unitary transformation.
\end{itemize}

While we are sure that the details of some justifications can be improved, we feel that the overall ``conceptual castle" is solid. The starting physical assumptions are simple, the translation into math is often straightforward and the intuitive ideas that emerge help unveil connections within and across different branches of math and physics. It is the tight fitting of these different pieces into a broader picture that validates the merit of the approach.

While the scope of this work is limited to refounding and reunderstanding physical theories already experimentally established, it would not be unexpected for a better grasp of old theories to lead to insights and ideas for new ones.

%We also have a more unified picture of fundamental physics which, after all ideas are digested, helps the intuition to gain insight that span across the different field. It also illustrate a new approach in fundamental physics: one where we better clarify our starting physical assumption and definition, where physics comes first, and they are later captured with mathematical definitions.

\begin{acknowledgements}
We'd like to thank Erik Curiel, Stephen DiIorio, Jiahua Gu, Josh Hunt, Isaac Mooney and Pavel Okun for providing feedback on drafts of this manuscript, Mark Greenfield for further insights in the topological aspects and Alan Phillips for contributing figures. We additionally thank Paul Stankus for interest and support in the earlier stages of this project. Funding for this work was provided in part by the MCubed program of the University of Michigan.
\end{acknowledgements}

\appendix
\section*{Appendix A. Mathematical proofs}
\stepcounter{section}

In this section we include mathematical demonstrations in support of the physical justifications of the paper.

\begin{thrm}[Extended Riesz theorem]\label{extended_riesz_theorem}
	Let $(S, \tau)$ be a locally compact Hausdorff space. Let $\Lambda = \{\Lambda_U : C(S) \rightarrow \mathbb{R}\}_{U \subseteq S}$ a family of positive linear functionals such that $\forall U \subseteq S \; \Lambda_U = \Lambda_{\mathrm{int}(U)}$ and $\Lambda_{U_1} + \Lambda_{U_2} = \Lambda_{U_1 \cup U_2} + \Lambda_{U_1 \cap U_2} \; \forall U_1, U_2 \subseteq S$. Then there exists a unique Borel measure $\mu$ such that $\Lambda_U (\rho) = \int_{U} \rho d\mu \; \forall \rho \in C(S)$.
\end{thrm}

\begin{proof}
	We claim that at each $s \in S$ there exists a compact neighborhood $U \subseteq S$ endowed with a Borel measure $\mu_U$ such that $\Lambda_U (\rho) = \int_U \rho d \mu_U \; \forall \rho \in C(S)$. $S$ is locally compact. $\forall s \in S \; \exists U \subseteq S$ such that $U$ is compact. $U$ with the subspace topology is a compact Hausdorff topological space on which is defined a positive linear functional $\Lambda_U : C(S) \rightarrow \mathbb{R}$. For the Riesz representation theorem for linear functionals, there exists a unique Borel measure $\mu_U$ such that $\Lambda_U (\mathcal{c}) = \int_U \rho d \mu_U \; \forall \rho \in C(S)$.
	
	We claim that $S$ is endowed with a unique Borel measure $\mu$ such that $\Lambda_U (\rho) = \int_U \rho d \mu$. Let $U \subseteq V \subseteq \mathcal{S}$ two compact subsets. Let $\rho \in C(S)$. $\Lambda_V (\rho) = \int_V \rho d \mu_V = \int_U \rho d \mu_V + \int_{V\backslash U} \rho d \mu_V$. Also $\Lambda_V (\rho) = \Lambda_U (\rho) + \Lambda_{V\backslash U} (\rho) = \int_U \rho d \mu_U + \int_{V \backslash U} \rho d \mu_{\overline{V \backslash U}} $. Therefore\journal{\break} $\int_{U} \rho_{\mathcal{c}} d \mu_U = \int_{U} \rho_{\mathcal{c}} d \mu_V$. $d\mu$ is unique and does not depend on the choice of neighborhood. $S$ is locally compact. It admits a cover $\{U_\alpha\}_{\alpha \in A}$ where each $U_\alpha$ is compact.
	\begin{align*}
	\Lambda_{\mathcal{S}}(\mathcal{c}) &= \sum \limits_{\alpha \in A} \Lambda_{U_\alpha}(\mathcal{c}) - \sum \limits_{\alpha \in A} \sum \limits_{\beta \in A}^{\beta \neq\alpha} \Lambda_{U_\alpha \cap U_\beta}(\mathcal{c}) \\
	&= \sum \limits_{\alpha \in A} \int_{U_\alpha} \rho_{\mathcal{c}} d\mu - \sum \limits_{\alpha \in A} \sum \limits_{\beta \in A}^{\beta \neq\alpha} \int_{U_\alpha \cap U_\beta}(\mathcal{c}) \rho_{\mathcal{c}} d\mu \\
	&= \int_{\mathcal{S}} \rho_{\mathcal{c}} d\mu
	\end{align*}
	The measure is unique on the whole space.
\end{proof}

\begin{thrm}\label{everywhere_integrable_is_lebesgue_integrable}
	Let $\mathcal{L}(S,\mu) \equiv \{ \rho : S \rightarrow \mathbb{R} \; | \;\; |\int_{U} \rho d\mu| < \infty \; \forall U \subseteq S\}$. $\mathcal{L}(S,\mu)=L^1(S,\mu)$ where $L^1(S,\mu) = \{ \rho : S \rightarrow \mathbb{R} \; | \;\; \int_{S} |\rho| d\mu < \infty \}$
\end{thrm}

\begin{proof}
	We claim $\mathcal{L}(S,\mu) \supseteq L^1(S,\mu)$. Let $\rho \in L^1(S,\mu)$. Let $U \subseteq S$. $\lVert \rho \rVert_U \equiv \int_{U} |\rho| d\mu < \int_{\mathcal{S}} |\rho| d\mu < \infty$. $\lVert \rho \rVert_U = \lVert \rho^+ \rVert_U + \lVert \rho^- \rVert_U$. $\lVert \rho^+ \rVert_U < \infty$ and $\lVert \rho^- \rVert_U < \infty$. $|\int_{U} \rho d\mu| = |\int_{U} \rho^+ d\mu - \int_{U} \rho^- d\mu| = |\int_{U} \rho^+ d\mu| - |\int_{U} \rho^- d\mu| = \lVert \rho^+ \rVert_U - \lVert \rho^- \rVert_U < \infty$. $\rho \in\mathcal{L}(S,\mu)$.
	
	We claim $L^1(S,\mu) \supseteq \mathcal{L}(S,\mu)$.  Let $\rho \in \mathcal{L}(S,\mu)$. Let $S^+ \equiv  \{ s \in S \; | \; \rho(s) > 0\}$. Let $S^- \equiv  \{ s \in S \; | \; \rho(s) < 0\}$. $|\int_{S^+} \rho d\mu| = |\int_{S} \rho^+ d\mu| = \lVert \rho^+ \rVert < \infty$ and $|\int_{S^-} \rho d\mu| = |\int_{S} \rho^- d\mu| = \lVert \rho^- \rVert < \infty$. $\lVert \rho \rVert = \lVert \rho^+ \rVert + \lVert \rho^- \rVert < \infty$. $\rho \in L^1(S,\mu)$.
	
	We claim $L^1(S,\mu) = \mathcal{L}(S,\mu)$. $\mathcal{L}(S,\mu) \supseteq L^1(S,\mu)$ and $L^1(S,\mu) \supseteq \mathcal{L}(S,\mu)$.	
\end{proof}

\begin{thrm}\label{symplectomorphism_generator}
	Let $(M, \omega)$ a symplectic manifold. Let $f: (M, \omega) \rightarrow (M, \omega)$ an infinitesimal self-symplectomorphism. Let $S \in TM \; | \; f(\xi^a(m)) = \xi^a(m) + S^a(m)dt \; \forall m \in M$ be the infinitesimal displacement. There exists a function $H: M \rightarrow \mathbb{R}$ such that $S^{a} \omega_{ab} = \partial_{b}H$ and $H \in C^2(M, \mathbb{R})$.
\end{thrm}

\begin{proof}
	We claim the vector field $S \in T\mathcal{Q}$ admits a potential $H$ such that $S^{a} \omega_{ab} = \partial_{b}H$. Let $v, w \in T_m M$ be two vectors defined at a point $m \in M$. Let $v^a, w^b$ be their components. Let $v'\equiv f v, w' \in T_{f(m)}M$ be the pushforward of $v, w$ by $f$. Since $f$ is a symplectomorphism we have $v^{a} \omega_{ab} w^{b} = v'^{a} \omega_{ab} w'^{b}$. The vector components change according to $v'^a = \partial_b \xi(t+dt)^a v^b = (\delta^a_b + \partial_b S^a dt) v^b$. We have:
	\begin{align*}
	v^{a} \omega_{ab} w^{b} &= v'^{a} \omega_{ab} w'^{b}  \\
	&= (v^{a} + \partial_{c} S^{a} v^{c} dt) \omega_{ab} ( w^{b} + \partial_{d} S^{b} w^{d} dt) \\
	&= v^{a} \omega_{ab} w^{b} + (\partial_{c} S^{a} v^{c} \omega_{ab} w^{b} + v^{a} \omega_{ab} \partial_{d} S^{b} w^{d}) dt \\ &+ O(dt^2)
	\end{align*}
	$v^{c} w^{b} \partial_{c} S_{b} - v^{a} w^{d} \partial_{d} S_{a} = 0$ where $S_{b} \equiv S^{a} \omega_{ab}$. The relationship must be true for and pair of vector, therefore $\partial_{a} S_{b} - \partial_{b} S_{a} = curl(S_{a}) = 0$. $S$ is a curl free vector field and admits a potential $H$ such that $S_{a} = \partial_{a}H$.
	
	We claim $H \in C^2(M, \mathbb{R})$. $S$ is differentiable as it is the displacement of a symplectomorphism. The derivatives of $H$ are differentiable as $S_{a} = \partial_{a}H$. $H$ is twice differentiable.
\end{proof}

\begin{thrm}\label{thrm:inner_product}
	Let $V$ be a complex vector space. Let $| \cdot | ^2 : V \rightarrow \mathbb{R}$ such that $|x|^2 \ge 0 \; \forall x \in V$ and $|x|^2 = 0 \Leftrightarrow x = 0$. Let $x \in V$ and $\mathsf{P}_x : V \rightarrow V$ be a family of projections (i.e. $\mathsf{P}_x$ linear and $\mathsf{P}_x \circ \mathsf{P}_x = \mathsf{P}_x$) such that: $\mathsf{P}_x (y) = y \Leftrightarrow y = a x$ for some $a \in \mathbb{C}$; $|x|^2|\mathsf{P}_x (y)|^2=|y|^2|\mathsf{P}_y (x)|^2$; $\mathsf{P}_x (\mathsf{P}_y (x)) = r x$ for some $r \in \mathbb{R}$.
	Then $V$ is an inner product space with the product $\langle \cdot , \cdot \rangle : V \times V \rightarrow \mathbb{C}$ defined by $\langle x , y \rangle x = |x|^2 \mathsf{P}_x (y)$.
\end{thrm}

\begin{proof}
	We claim $\langle \cdot , \cdot \rangle$ is well defined. Let $x, y \in V$. $\mathsf{P}_x (\mathsf{P}_x (y)) =\mathsf{P}_x(y)$ as $\mathsf{P}_x$ is a projection. $\mathsf{P}_x(y) = a x $ for some $a \in \mathbb{C}$ by the first property of $\mathsf{P}_x$. $\langle x , y \rangle = |x|^2 a$ exists and is unique. 
	
	We claim $\langle \cdot , \cdot \rangle$ is positive definite. Let $x \in V$. $\langle x , x \rangle x = |x|^2 \mathsf{P}_x(x) = |x|^2 x $. Therefore $\langle x , x \rangle = |x|^2 \ge 0$ and $\langle x , x \rangle = |x|^2 = 0 \Leftrightarrow x = 0$.
	
	We claim $\langle \cdot , \cdot \rangle$ is linear in the second argument. Let $x,y,z \in V$ and $a,b \in \mathbb{C}$. $\langle x , a y + b z \rangle x = |x|^2 \mathsf{P}_x(a y + b z)=a  |x|^2 \mathsf{P}_x(y) + b  |x|^2 \mathsf{P}_x(z) = a \langle x , y \rangle x + b \langle x, z \rangle x$ by linearity of $\mathsf{P}_x$. Therefore $\langle x , a y + b z \rangle = a \langle x , y \rangle + b \langle x, z \rangle$.
	
	We claim $\langle \cdot , \cdot \rangle$ is conjugate symmetric. Let $x,y \in V$. $ | |x|^2\mathsf{P}_x(y)|^2 = |x|^4|\mathsf{P}_x(y)|^2 = |\langle x , y \rangle|^2 |x|^2$. Using the second property of the projection we have $|\langle x , y \rangle|^2 = |x|^2|\mathsf{P}_x(y)|^2 =  |y|^2|\mathsf{P}_y(x)|^2 = |\langle y , x \rangle|^2$. The modulus of the product is symmetric as a function of the arguments. Now consider $|x|^2|y|^2\mathsf{P}_x(\mathsf{P}_y(x)) = |x|^2 \langle y , x \rangle \mathsf{P}_x(y) = \langle y , x \rangle \langle x , y \rangle x$. The product $\langle x , y \rangle \langle y , x \rangle \in \mathbb{R}$ by the third property of the projection. $\arg(\langle x , y \rangle) + \arg(\langle y , x \rangle)$ is $0$ or $\pi$. $\pi$ is excluded since $\langle x , x \rangle = |x|^2 \ge 0$. $\langle y , x \rangle$ and $\langle x , y \rangle$ have equal modulus and opposite phase. $\langle x , y \rangle = \langle y , x \rangle^\dagger$
\end{proof}

\begin{thrm}\label{thrm:covariant_operator}
	Let $C^1(\mathcal{Q}, \mathbb{C})$ be the space of differentiable complex functions defined over $\mathcal{Q}$. Let $q : \mathcal{Q} \rightarrow \mathbb{R}$ a differentiable function. Let $K_q : C^1(\mathcal{Q}, \mathbb{C}) \rightarrow C^1(\mathcal{Q}, \mathbb{C})$ be a covariant linear operator, that is given an invertible differentiable function $\hat{q} : \mathbb{R} \rightarrow \mathbb{R}$ $K_{\hat{q}} = \partial_{\hat{q}} q K_{q}$. Then $K_q = a \partial_q$ for some $a \in \mathbb{C}$.
\end{thrm}

\begin{proof}
	We claim $K_q(f(q))= \partial_q f K_q(q)$. Let $I : \mathbb{R} \rightarrow \mathbb{R}$ be the identify function, that is $I \circ f = f \; \forall f : \mathbb{R} \rightarrow \mathbb{R}$. $K_q(q) = K_q(I(q)) = \partial_q I(q) K_{I}(I) = K_{I}(I)$ by covariance.
	Let $f : \mathbb{R} \rightarrow \mathbb{R}$ be a differentiable invertible function. $K_q(f(q)) =K_q(I(f(q))) = \partial_q f K_{f}(I(f)) = \partial_q f \partial_f I(f) K_{I}(I) = \partial_q f K_q(q)$. Now let $f : \mathbb{R} \rightarrow \mathbb{R}$ be a differentiable function non necessarily invertible. It can be written as $f = f_1 + f_2$ with $f_1, f_2$ differentiable and invertible. By linearity $K_q(f)=K_q(f_1) + K_q(f_2) = \partial_q f_1 K_q(q) + \partial_q f_2 K_q(q)  = \partial_q f K_q(q)$.
	
	We claim $K_q$ is a derivation. Let $f,g,h$ be differentiable functions such that $h(q)=f(q)g(q)$. $K_q(ln(h(q)))=\partial_q ln(h(q)) K_q(q) = \frac{1}{h} \partial_q h K_q(q) = \frac{1}{h} K_q(h(q)) = K_q(ln(f(q)g(q)))\journal{\break} = K_q(ln(f(q))) + K_q(ln(g(q))) = \frac{1}{f} K_q(f(q)) + \frac{1}{g} K_q(g(q))$. Multiplying by $h=fg$ we have $K_q(h(q)) = g K_q(f(q)) + f K_q(g(q))$. $K_q$ is a linear operator that satisfies the product rule. $K_q$ is a derivation by definition.
	
	We claim $K_q = a \partial_q$ for some $a \in \mathbb{C}$. $K_q$ is a derivation. As $K_q(f(q))= \partial_q f K_q(q)$ it is a derivation along $q$, which can be expressed as $K_q(q)=a \partial_q$ for some $a \in \mathbb{C}$.
\end{proof}

\begin{thrm}\label{thrm:antihermitian_derivative}
	Let $G \equiv C^1(\mathcal{Q}, \mathbb{C}) \cap L^1(\mathcal{Q}, \mathbb{C})$ be the space of differentiable and integrable complex functions defined over $\mathcal{Q}$. Let $K$ be a linear adjoint operator of the form $a \partial_q$ with $a \in \mathbb{C}$. Let $K$ be self-adjoint with respect to the inner product $\langle f, g \rangle = 
	\int_{\mathcal{Q}} f^\dagger(q) g(q) dq$. Then $K= \imath r \partial_q$ with $r \in \mathbb{R}$.
\end{thrm}

\begin{proof}
	We claim $K= \imath r \partial_q$ with $r \in \mathbb{R}$. Let $f,g \in G$. We have $\langle f, K g \rangle = \int_{\mathcal{Q}}f^\dagger a \partial_q g dq = f^\dagger g |_{\partial \mathcal{Q}}  - \int_{\mathcal{Q}} a \partial_qf^\dagger  g dq= \int_{\mathcal{Q}} (- a\partial_q f^\dagger) g dq = \int_{\mathcal{Q}} (- a^\dagger\partial_q f)^\dagger g dq = \langle - a^\dagger\partial_q f, g \rangle$ using integration by part and integrability of $f$ and $g$. $\langle - a^\dagger\partial_q f, g \rangle =\langle K f, g \rangle$ as $K$ is self-adjoint. $- a^\dagger\partial_q = a \partial_q$. $a = - a^\dagger$. $a = \imath r$ with $r \in \mathbb{R}$.
\end{proof}

\end{document}